\documentclass[a4paper,11pt,reqno]{amsart}
\usepackage[utf8]{inputenc}
\usepackage{mathrsfs}
\usepackage{dsfont}
\usepackage{hyperref}
\usepackage{amsmath}
\usepackage{amssymb}
\usepackage{amsthm}
\usepackage{amsfonts}
\usepackage{amstext}
\usepackage{amsopn}
\usepackage{amsxtra}
\usepackage{mathrsfs}
\usepackage{dsfont}
\usepackage{esint}
\usepackage{enumitem}
\usepackage{graphicx}
\usepackage{caption}
\newcommand{\R}{\mathbb{R}}
\newcommand{\Z}{\mathbb{Z}}
\newcommand{\C}{\mathbb{C}}
\newcommand{\N}{\mathbb{N}}
\newcommand{\bS}{\mathbb{S}}
\newcommand{\dps}{\displaystyle}
\newcommand{\ii}{\infty}
\newcommand\1{{\ensuremath {\mathds 1} }}
\renewcommand\phi{\varphi}
\renewcommand\ge{\geq}
\renewcommand\le{\leq}

\newcommand{\wto}{\rightharpoonup}

\newcommand{\cB}{\mathcal{B}}

\newcommand{\cE}{\mathcal{E}}
\newcommand{\cF}{\mathcal{F}}

\newcommand{\dx}{\rd x}
\newcommand{\dy}{\rd y}
\newcommand\pscal[1]{{\ensuremath{\left\langle #1 \right\rangle}}}
\newcommand{\norm}[1]{ \left\| #1 \right\|}

\renewcommand{\geq}{\geqslant}
\renewcommand{\leq}{\leqslant}
\renewcommand{\hat}{\widehat}
\renewcommand{\tilde}{\widetilde}
\newcommand{\eps}{\varepsilon}
\usepackage{color}

\newcommand{\nn}{\nonumber}
\newcommand{\rd}{{\rm d}}

\newtheorem{theorem}{Theorem}[section]
\newtheorem{lemma}[theorem]{Lemma}
\newtheorem{remark}[theorem]{Remark}
\newtheorem{assumption}[theorem]{Assumption}
\newtheorem{corollary}[theorem]{Corollary}
\newtheorem{definition}[theorem]{Definition}
\newtheorem{proposition}[theorem]{Proposition}

\title[Ground states of the Gross-Pitaevskii equation]{Positive-density ground states of the Gross-Pitaevskii equation}

\author[M. Lewin]{Mathieu Lewin}
\address{CEREMADE, CNRS, Université Paris-Dauphine, PSL Research University, Place de Lattre de Tassigny, 75016 Paris, France}
\email{mathieu.lewin@math.cnrs.fr}

\author[P. T. Nam]{Phan Thành Nam}
\address{Department of Mathematics, LMU Munich, Theresienstrasse 39, 80333 Munich, Germany}
\email{nam@math.lmu.de}

\date{\today}

\begin{document}

\maketitle

\begin{abstract}
We consider the nonlinear Gross-Pitaevskii equation at positive density, that is, for a bounded solution not tending to 0 at infinity. We focus on infinite ground states, which are by definition minimizers of the energy under local perturbations. When the Fourier transform of the interaction potential takes negative values we prove the existence of a phase transition at high density, where the constant solution ceases to be a ground state. The analysis requires mixing techniques from elliptic PDE theory and statistical mechanics, in order to deal with a large class of interaction potentials.

\medskip

\tiny \noindent\copyright\;2023 the authors.
\end{abstract}

\tableofcontents

 \section{Introduction}

We study solutions $u:\R^d\to\C$ to the nonlinear equation
\begin{equation}
 \big(-\Delta+w\ast|u|^2\big)u=\mu\,u,
  \label{eq:GP_infinite}
\end{equation}
where
$$\boxed{\mu>0}$$
is a fixed parameter and $w$ is a given real-valued function or measure with $\int_{\R^d}|w|<\ii$ and $\int_{\R^d}w>0$. The positivity of $\mu$ implies that solutions to~\eqref{eq:GP_infinite} must be ``infinitely extended'', in the sense that
$$\int_{B(z,R)}|u|^2\geq \alpha R^d$$
for all $z\in\R^d$ and all $R$ large enough. Here $B(z,R)$ denotes the ball of center $z\in\R^d$ and radius $R$, and the constant $\alpha>0$ is independent of $z$ and $R$. In particular, $u$ cannot tend to 0 at infinity. The simplest solution of~\eqref{eq:GP_infinite} is the constant function
\begin{equation}
u(x)\equiv\left(\frac{\mu}{\int_{\R^d}w}\right)^{\frac12}
 \label{eq:u_cnst_intro}
\end{equation}
but we will see that there are other more interesting solutions.

The nonlinear equation~\eqref{eq:GP_infinite} typically describes a \emph{system of infinitely many quantum particles in a mean-field-type approximation} and $w$ is interpreted as the interaction potential between the particles. This equation appears under different names in the literature, depending on the context. For Cooper pairs in superconductors, it is called the \emph{(non-local) Ginzburg-Landau equation}~\cite{GinLan-50,BetBreHel-94,PacRiv-00,DeLaire-12,FraHaiSeiSol-12}, whereas for the condensate part of an infinite Bose gas it is called the \emph{Gross-Pitaevskii} or \emph{Hartree equation}~\cite{Gross-57,Gross-58,Gross-61,Pitaevskii-61,PitStr-03,LieSeiSolYng-05,LewNamRou-14}. When $u$ is real-valued and $w=\delta$, the same equation can also describe phase separation in binary fluids and living tissues, in which case it is called the \emph{Van Der Waals}~\cite{VanDerWaals-79,Rowlinson-79}, \emph{Cahn-Hilliard} or \emph{Allen-Cahn equation}~\cite{CahHil-58,Miranville-19}. We use the name \emph{Gross-Pitaevskii} throughout.

We will restrict our attention to a particular class of solutions called \emph{infinite ground states}. Those are by definition minimizers of the associated energy functional
\begin{multline}
\cF_\mu(u):=\int_{\R^d}|\nabla u(x)|^2\,\dx+\frac12\iint_{\R^d\times\R^d}|u(x)|^2|u(y)|^2w(x-y)\,\dx\,\dy\\
-\mu\int_{\R^d}|u(x)|^2\,\dx.
 \label{eq:GP_intro}
\end{multline}
Since $u$ is infinitely extended the above integrals are infinite, so that the concept of minimizers is unclear. The traditional way of solving this issue is to simply ask that the energy goes up when $u$ is perturbed locally, whatever the size of the local perturbation. In other words, we require that
\begin{equation}
\text{``}\cF_\mu(v)-\cF_\mu(u)\text{''}\geq0\label{eq:cond_GS}
\end{equation}
for any $v$ that coincides with $u$ outside of a bounded set. The formal difference of the energies in~\eqref{eq:cond_GS} can be given a rigorous meaning by writing integrals of differences, for instance $\int_{\R^d} (|\nabla v|^2-|\nabla u|^2)$ for the first term, see Definition~\ref{def:infinite_GP_GS} below. We call such solutions \emph{infinite ground states} to emphasize that $u$ is infinitely extended and that its energy is infinite. Often we will simply call them \emph{ground states}. In the literature, they have also been called \emph{local minimizers} but this terminology could suggest that only small perturbations are considered. The allowed perturbation are indeed ``local'' in space but arbitrarily large in range and size.

Our goal is to prove the existence of such infinite ground states and study their properties. In particular we would like to know if the constant solution~\eqref{eq:u_cnst_intro} is an infinite ground state and, if not, what are the properties of the ground state(s).

Since~\eqref{eq:GP_infinite} physically describes an infinite system, it seems natural to adopt the point of view of statistical mechanics~\cite{Ruelle}, in particular concerning the existence of \emph{phase transitions} and \emph{symmetry breaking phenomena}. Equation~\eqref{eq:GP_infinite} is invariant under the action of several symmetry groups. We can translate a solution $u$ in space or multiply it by a constant phase factor in the form $e^{i\theta_0}u(x-\tau)$ and we still get a solution. If $w$ is radial, we can also rotate $u$. In statistical mechanics, a \emph{fluid phase} (a gas or a liquid) is usually characterized by the invariance of the equilibrium state under translations and rotations. Such states can be very complicated for general interacting models, but they are completely trivial for the mean-field equation~\eqref{eq:GP_infinite} considered in this paper. In fact, constants are the only functions invariant under space translations and they solve~\eqref{eq:GP_infinite} if and only if $u$ is given by~\eqref{eq:u_cnst_intro}, up to a constant phase factor. This leads us to interpret~\eqref{eq:u_cnst_intro} as an \emph{equilibrium fluid phase} of our system. We then want to investigate whether there are \emph{phase transitions} when the parameter $\mu$ is varied, corresponding to the appearance of other types of solutions. A \emph{solid phase} would typically be a non-trivial periodic solution $u$, for instance. The existence of such phases is going to crucially depend on the properties of the interaction potential $w$.

\subsubsection*{\textbf{Density}}
Infinite ground states will be shown to have a well-defined \emph{average density} (or mass per unit volume) $\rho \in (0,\infty)$, that is,
\begin{equation}
\lim_{R\to\ii}\inf_{z\in \R^d} \frac{1}{|B(0,R)|} \int_{B(z,R)}|u(y)|^2\,\dy =\rho.
 \label{eq:density_intro}
\end{equation}
We will show that the limit also exists with $\sup$ instead of $\inf$. In general the two limits can be different and in that case we think of $u$ as describing an interface between two equilibrium states with different densities but the same $\mu$. As we will see, the two limits will nevertheless coincide for almost all $\mu>0$. We can therefore interpret $\mu$ as a dual variable or penalization, which can be used to find ground states having any desired density $\rho$. Although one can think of fixing $\rho$, it is often easier to fix $\mu$. In statistical mechanics, $\mu$ is called the \emph{chemical potential}.

The density $\rho$ will be proved to be comparable to $\mu$ in the sense that $\mu/C\leq \rho\leq C\mu$ for some positive constant $C$, so that the regimes of high and low density correspond to large and small $\mu$, respectively. At low density ($\mu\ll1$), we expect that the system is a fluid and that the constant function~\eqref{eq:u_cnst_intro} is an infinite ground state. A solid phase could only occur at high density ($\mu\gg1$). This is what we will prove in this paper when the Fourier transform $\widehat w$ changes sign.

\subsubsection*{\textbf{Contact interaction}}
The problem of studying infinite ground states for Equation~\eqref{eq:GP_infinite} has a long history. Many works have been devoted to the special case of the Dirac delta
$$w=\delta_0.$$
The Dirac delta naturally occurs in macroscopic limits where the details of the microscopic interactions between the particles get simplified, or at low density $\mu\ll1$ where the particles interact rarely. The Ginzburg-Landau equation of superconductivity has $w=\delta$~\cite{GinLan-50,BetBreHel-94,FraHaiSeiSol-12}. For Bose gases, one arrives at $w=4\pi a\delta$ at low density, where $a$ is the $s$-wave scattering length~\cite{LieYng-98,LieSeiYng-00,LieSeiSolYng-05}.

Note that for $w=\delta$ the equation has an exact scaling and nothing special can happen when varying $\mu$. In fact, letting $\tilde u(x)=\sqrt\mu u(\sqrt\mu x)$ we are reduced to $\mu=1$. The constant solution~\eqref{eq:GP_infinite} can be proved to be an infinite ground state and solid phases are not expected to arise at all. However, there are some surprises when looking at the possibility of breaking rotational invariance, as we now quickly review.

A huge amount of works has been devoted in the 1990s to the two-dimensional case $d=2$~\cite{BetBreHel-94,PacRiv-00}. For $w=\delta$, it was proved that there are \emph{exactly two infinite ground states in 2D}, up to translations, multiplication by a constant phase factor and complex conjugation~\cite{BetBreHel-94,BreMerRiv-94,Sandier-98,Mironescu-96}:
\begin{itemize}
 \item the constant solution~\eqref{eq:u_cnst_intro},
 \item a vortex of topological degree one, that is, behaving at infinity like
 $$u(x)\underset{|x|\to\ii}\sim \sqrt{\mu}\,\frac{x_1+ix_2}{\sqrt{x_1^2+x_2^2}}.$$
 \end{itemize}
The vortex has an energy infinitely higher than the constant (by a logarithm of the volume) but it is nevertheless an infinite ground state, topologically protected by its degree at infinity. Note that $|u(x)|\to\sqrt\mu$ at infinity, so the solution is not very different from the constant solution if we only look at the modulus. Similar solutions are known to exist in higher dimensions~\cite{BetBreOrl-01,ConJer-17,SanSha-17,Roman-19}, but with vortices concentrated on manifolds of dimension $d-2$ (a line in $d=3$ for instance). Uniqueness is then largely open~\cite{Pisante-14}, however.

Since we are motivated by quantum mechanics, we work with complex-valued solutions $u$. When $u$ is on the contrary assumed to be \emph{real-valued} (Van Der Waals or Cahn-Hilliard model), the situation is different. An infinite ground state \emph{among real-valued functions} (that is, assuming $u$ and $v$ are real-valued in~\eqref{eq:cond_GS}) need not be one among complex-valued functions. In fact, there exists a famous kink solution $u(x)=\tanh(x/\sqrt2)$ in one dimension satisfying $u(x)\to\pm 1$ for $x\to\pm\ii$ and solving~\eqref{eq:GP_infinite} for $w=\delta$ and $\mu=1$. It is known to be a ground state among real-valued functions but it is not an infinite ground state among complex-valued functions. In other words, it becomes unstable under perturbations of the phase. The 1D kink solution was proved to be the basic object describing \emph{phase separation} for real-valued ground states in dimensions $d\leq 7$~\cite[Thm.~2.3]{Savin-09}. This is related to the famous \emph{De Giorgi conjecture} for minimal surfaces~\cite{DeGiorgi-79,Pisante-14,ChaWei-18} which was later proved in~\cite{GhoGui-98,AmbCab-00,Savin-09,PinKowWei-11}.

On the other hand, solutions taking values in a higher dimensional space (that is, $u:\R^d\to\R^m$ with $m\geq d$) have also been studied at length~\cite{MilPis-10,Pisante-11}, but this goes beyond the scope of this paper.

\subsubsection*{\textbf{Nonlocal interactions}}
Although we have some new results on the delta interaction, we mainly focus on \emph{arbitrary interactions $w$}. Then~\eqref{eq:GP_infinite} is sometimes called the \emph{nonlocal Ginzburg-Landau equation}~\cite{DeLaire-12,DeLaiLop-22,DeLaiMen-20} or the \emph{Hartree equation}~\cite{LewNamRou-14}. The original works of Gross in 1957--58~\cite{Gross-57,Gross-58} in fact had a general potential $w$ and the case of the delta potential was only considered in 1961 with Pitaevskii~\cite{Gross-61,Pitaevskii-61}. For simplicity, we call~\eqref{eq:GP_infinite} the \emph{Gross-Pitaevskii equation}.

We will be particularly interested in the situation where the Fourier transform $\widehat w$ has no sign, a case for which it is predicted that there exists a \emph{transition to a solid phase at large densities}. This was noticed in $d=1$ in~\cite{Lewin-15} which has motivated the current work, and is in fact well known in the Physics literature~\cite{Gross-57,KirNep-71,Nepomnyashchii-71,NepNep-71}. In particular, Gross wrote in 1957:
\begin{quote}\it
``We note that there is always a solution of uniform density. (...) But if $w(x)$ is negative in some region of space, there may be other solutions, such as a periodic solution with lower energy than for the uniform solution.''~\cite[p.~161]{Gross-57}
\end{quote}
Our goal is to make some rigorous advances in this direction. Although we cannot prove periodicity of ground states, we will at least confirm the existence of the phase transition.

For Bose-Einstein condensates, a periodic infinite ground state is believed to describe a \emph{supersolid}~\cite{Gross-57,MatTsu-70,Legett-70,PomRic-93,PomRic-94,CinJaiBinMicZolPup-10,HenCinJaiPupPoh-12,MacMacCinPoh-13}.
Although this state of matter has not yet been observed in solid helium~\cite{ChaHalRea-13,Hallock-15}, recent experiments have confirmed the existence of a supersolid phase in ultracold atomic gases, in particular with dipolar interactions~\cite{Kadau-etal-16,LeoMorZupEssDon-17,LiLeeBurShtTopJamKet-17,TanRocLucFamFioGabModRecStr-19,TanLucFamCatFioGabbBisSanMod-19,BotSchWenHerGuoLanPfa-19,Chomaz_etal-19,Norcia_etal-21}. On the other hand, numerical simulations have predicted the existence of a periodic solutions in several cases. One is the step potential
\begin{equation}
w(x)=c\1(|x|\leq R_0)
\label{eq:example_w_cnst}
\end{equation}
which has been largely considered due to its simplicity, even if it is probably not very physical~\cite{PomRic-93,PomRic-94,JosPomRic-07b,SepJosRic-10,KunKat-12,WatBra-12,AncRosToi-13,MacMacCinPoh-13,MasJos-13}. This model was also rigorously studied in~\cite{AftBlaJer-09} for small and large values of $c$, in a box with periodic boundary conditions. On the other hand, it was shown in~\cite{HenNatPoh-10,HenCinJaiPupPoh-12,CinMacLecPupPoh-14} that the Van Der Waals-type potential
\begin{equation}
w(x)=\frac{c}{1+|x|^6}
\label{eq:example_w_VdW}
\end{equation}
is appropriate for Bose-Einstein condensates in which atoms are weakly coupled to a highly excited Rydberg state. The most important property of the two potentials~\eqref{eq:example_w_cnst}--\eqref{eq:example_w_VdW} is that $\widehat{w}$ has no sign, which allows for the existence of periodic solutions with a roton-maxon excitation spectrum~\cite{PomRic-93,CorNatLi-20}. Both potentials will be covered by our rigorous results.

\subsubsection*{\textbf{Main contributions}}
The main new contribution of this work is to bring techniques from statistical mechanics into~\eqref{eq:GP_infinite}. This will allow us to treat a very large class of interaction potentials $w$, at a level of generality which has never been achieved for the nonlinear PDE~\eqref{eq:GP_infinite}, to our knowledge. We will only assume that $w$ is \emph{superstable} in the sense of Ruelle~\cite{Ruelle-70}, that is, remains stable if we remove a small positive part close to the origin. This is a very natural assumption in statistical mechanics which covers all physically interesting interactions.

To obtain the existence of infinite ground states, the usual technique is to resort to a \emph{thermodynamic limit}. Namely, we minimize the energy~\eqref{eq:cF_GP} in a ball of radius $L$ with suitable boundary conditions and obtain a minimizer $u_L$ by elementary compactness arguments. Then we want to show that $u_L$ is locally bounded so as to get an infinite ground state locally when we pass to the limit $L\to\ii$.

In the context of the Ginzburg-Landau equation with $w=\delta$, due to the exact scaling mentioned before it is the same to take $L\to\ii$ or to work in a fixed ball with a large constant $L^2$ in front of the interaction energy. This is in fact how the limit is usually stated in the literature. Obtaining a ground state in infinite space requires zooming at the microscopic scale~\cite{BetBreHel-94}. For a general $w$, it is important to right away work in a large ball.

For Ginzburg-Landau, the boundedness of $u_L$ was obtained in the literature using the maximum principle. This is not working well for our non-local problem, especially under our general assumptions on $w$. We will prove that $u_L$ is bounded in $L^\ii$ using an argument inspired by Ruelle~\cite{Ruelle-70}, relying only on energy comparison techniques. These bounds give the existence of an infinite ground state after passing to the limit $L\to\ii$.

We then study the properties of all infinite ground states, which may or may not come from the previous limit. We first prove that those always possess some natural ``thermodynamic properties''. They have the energy per unit volume and density $\rho$ in~\eqref{eq:density_intro} which can be obtained from the specific free energy of the infinite gas. Using again ideas from statistical mechanics~\cite{Ginibre-67,Ruelle-70}, we prove the important fact that the specific free energy is a \emph{strictly concave function of $\mu$}. This is related to the compressibility of the infinite system and implies that the map $\rho\mapsto \mu(\rho)$ is a well defined non-decreasing function. Any interval where $\mu$ is constant corresponds to a first order phase transition with the coexistence of two phases of different densities.

We then discuss the breaking of symmetries. The first question we ask is whether infinite ground states are automatically real-valued and positive. We can show this is always the case in dimension $d=1$, but not in higher dimensions. In fact, in dimension $d=2$ we can prove the existence of a vortex at low density, using some tools introduced in the context of the Ginzburg-Landau theory~\cite{BetBreHel-94} and as was also predicted by Gross in 1961~\cite{Gross-61}. The vortex represents a rotating quantum fluid, with the rotation seen in the complex phase of $u(x)$.

Next we investigate the translational symmetry of infinite ground states. We show that there exists a critical value $0<\mu_c\leq\ii$, below which the constant function~\eqref{eq:u_cnst_intro} is an infinite ground state and above which it is not anymore. In addition, we have $\mu_c<\ii$ if and only if $\widehat{w}$ changes sign. For small $\mu$, we can prove that the constant is the \emph{unique real-valued infinite ground state}. The main tool we use here is the convergence to the delta case $w=\delta$ in the limit $\mu\to0$, after rescaling, so that we can use some properties of the Ginzburg-Landau equation. For $\mu>\mu_c$ the constant function is no longer a ground state. When $\mu\to\ii$ we obtain convergence to the classical mean-field theory previously studied in~\cite{Suto-11}, which is known to have non-trivial solutions.

It is an interesting open problem to show that infinite ground states are periodic for large $\mu$, even just in 1D. This is one more instance of the famous \emph{crystallization conjecture}~\cite{BlaLew-15}. We can only show that they are really ``far from a constant everywhere''. For instance we will prove that
there exists a positive constant $c>0$ such that
$$\max_{x\in B(z,R)}|u(x)|\geq c+\min_{x\in B(z,R)}|u(x)|,$$
for any center $z\in \R^d$ and any large enough radius $R$, showing that there are fluctuations everywhere in space.

The next section contains a detailed description of all our results, without proofs. The latter are then provided in the following sections.

\subsubsection*{\textbf{Acknowledgement}} This project has received funding from the European Research Council (ERC) under the European Union's Horizon 2020 research and innovation programme (grant agreements MDFT No. 725528 of ML and RAMBAS No. 101044249 of PTN). We thank Rupert Frank, Philippe Gravejat, Radu Ignat, Vincent Millot, Nicolas Rougerie, \'Etienne Sandier and Sylvia Serfaty for useful comments. We also thank the referees for their careful reading and suggestions.

\section{Main results}

In this section we state all our main results, which are then proved in detail in the following sections.

\subsection{Definition of infinite ground states}
We first describe our assumptions on the interaction potential $w$ and introduce the concept of infinite ground states for the Gross-Pitaevskii (GP) equation \eqref{eq:GP_infinite}. For simplicity, we use the shorthand notation $B_r:=B(0,r)$ for the ball centered at the origin.

\begin{assumption}[Interaction]\label{ass:w}
Let
$w\in \R_+ \delta_0+L^1(\R^d)$, with $\mathbb{R}_+=[0,\infty)$, be an even real-valued potential satisfying
\begin{equation}
w = \eps\,\delta_r\ast\delta_r+w_2
\label{eq:decomp_w}
\end{equation}
with $\eps>0$ and $r\geq0$, where
\begin{equation}
\delta_r:=\begin{cases}
|B_r|^{-1}\1_{B_r}&\text{for $r>0$,}\\
\delta_0&\text{for $r=0$.}
            \end{cases}
\label{eq:def_delta_r}
\end{equation}
The potential $w_2$ is assumed to be \emph{stable}, that is,\footnote{We use here a slight abuse of notation, since $w_2$ can contain a Dirac delta at the origin.}
\begin{equation}
\iint_{\R^d\times\R^d}\rho(x)\rho(y)w_2(x-y)\dx\,\dy\geq0,\qquad\forall 0\leq \rho\in L^2(\R^d),
 \label{eq:w2_stable}
\end{equation}
\emph{lower regular}
\begin{equation}
w_2(x)\geq -\frac{\kappa}{1+|x|^{s}},\qquad\forall x\in\R^d,
\label{eq:w2_regular}
\end{equation}
and
\emph{upper regular at infinity}, that is,
\begin{equation}
w(x)=w_2(x)\leq \frac{\kappa}{|x|^{s}},\qquad\forall |x|\geq \kappa,
\label{eq:w2_upper}
\end{equation}
for some $s>d$ and $\kappa\geq 2r$.
\end{assumption}

A potential satisfying the above assumption was called \emph{superstable} by Ruelle in the famous article~\cite{Ruelle-70}. This covers all short range pair potentials of interest in statistical mechanics. The decomposition~\eqref{eq:decomp_w} means that $w$ can be written as a small positive part in a neighborhood of the origin, which is either a contact interaction $\delta_0$ or a strictly positive function (taken in the form of a convolution for convenience), plus a \emph{stable}\footnote{In statistical mechanics, a potential is called stable~\cite{Ruelle} when
$$\sum_{1\leq j<k\leq N}w_2(x_j-x_k)\geq -CN$$
for all $N\geq2$ and all $x_1,...,x_N\in\R^d$. Integrating against $\rho(x_1)\cdots \rho(x_N)/N^2$ and taking the limit $N\to\ii$, we find the mean-field condition~\eqref{eq:w2_stable}.}
short range potential $w_2$ decaying polynomially at infinity.

Assumption~\ref{ass:w} covers for instance any non-negative potential which is strictly positive in a neighborhood of the origin and decays fast enough, including~\eqref{eq:example_w_cnst} and~\eqref{eq:example_w_VdW} in the physical dimensions $d\in\{1,2,3\}$. We also allow $w_2$ to have a negative part. Note however that $w_2(0)\ge 0$ and $\int_{\R^d} w_2\geq0$ (take $\rho=\1_{B_R}$ in~\eqref{eq:w2_stable} and the limits $R\to0$ and $R\to\ii$, respectively) and hence
$$\int_{\R^d}w(x)\,\dx>0.$$
Recall that neutral atoms interact at infinity with an attractive Van Der Waals force decaying like $-|x|^{-6}$ in dimension $d=3$~\cite{LieThi-86}. Our assumptions allow it. A typical example is a truncated Lennard-Jones potential $w(x)=\min(A,|x|^{-12}-|x|^{-6})$ which satisfies all the above assumptions for $A$ large enough by~\cite[Lem.~22]{Triay-18}.

Note that we do not impose any upper bound on $w_2$ close to the origin (in particular, $w_2$ can also include a Dirac delta at 0). This may pose some technical difficulty. In particular, we will often have to restrict ourselves to the functions $u$ of finite local interaction energy, which means that
\begin{equation}
\iint_{K^2}|u(x)|^2|u(y)|^2w_+(x-y)\,\dx\,\dy<\ii
\label{eq:finite_interation_u}
\end{equation}
for every compact set $K\subset \R^d$. The property~\eqref{eq:finite_interation_u} is automatically satisfied in dimensions $d\in\{1,2,3,4\}$ whenever $u\in H^1_{\rm loc}(\R^d)$, as we will always require. But it is necessary in dimensions $d\geq5$ if $w_+$ is too rough.

For $\Omega \subset \R^d$ a smooth domain, we define the Gross-Pitaevskii (GP) energy by
\begin{equation}
\boxed{\cE_\Omega(u):=\int_\Omega|\nabla u|^2+\frac12\iint_{\Omega\times\Omega}w(x-y)|u(x)|^2|u(y)|^2\,\dx\,\dy}
 \label{eq:cE_GP}
\end{equation}
and the grand-canonical free energy
\begin{equation}
\boxed{\cF_{\mu,\Omega}(u):=\cE_\Omega(u)-\mu\int_\Omega|u(x)|^2\,\dx.}
 \label{eq:cF_GP}
\end{equation}
When $\Omega=\R^d$, we simply write $\cE:=\cE_{\R^d}$ and $\cF_\mu:=\cF_{\mu,\R^d}$ as we already did in the introduction in~\eqref{eq:GP_intro}.
Due to the stability condition in Assumption~\ref{ass:w}, $\cE_\Omega$ is a non-negative functional which is finite for all $u\in H^1(\Omega)$ such that
$$\iint_{\Omega^2}w_+(x-y)|u(x)|^2|u(y)|^2\,\dx\,\dy<\ii.$$
In dimension $d\leq4$ the latter is automatic from the $H^1$ regularity. Note that in a bounded set $\Omega$, $\cF_{\mu,\Omega}$ is bounded from below. In fact, from~\eqref{eq:decomp_w} and the stability condition~\eqref{eq:w2_stable} in Assumption~\ref{ass:w}, we have
$$\cE_\Omega(u)\geq \eps\int_{\R^d}(|u|^2\ast\delta_r)^2,$$
where it is here understood that $u$ is extended by $0$ outside of $\Omega$. We then obtain
\begin{align}
\cF_{\mu,\Omega}(u)&\geq \eps\int_{\R^d} (\delta_r\ast|u|^2)^2-\mu\int_{\R^d}|u|^2\nn\\
&= \eps\int_{\R^d} (\delta_r\ast|u|^2)^2-\mu\int_{\R^d}\delta_r\ast |u|^2\geq-\frac{\mu^2}{4\eps}|\Omega+B_r|\label{eq:F_bounded_below}
\end{align}
since $\delta_r\ast |u|^2$ is supported on $\Omega+B_r$.

In this paper we will be interested in minimizers of $\cE_\Omega$ at fixed mass $\int_\Omega|u|^2=\lambda$ (canonical case), or of $\cF_{\mu,\Omega}$ without any constraint (grand-canonical case). This is only a well-defined problem in a bounded domain~$\Omega$, however. In the whole space, we need to resort to the concept of infinite ground states.

\begin{definition}[Infinite ground state]\label{def:infinite_GP_GS} Let $d\geq1$ and $w$ satisfy Assumption~\ref{ass:w}. A function $u\in H^1_{\rm unif}(\R^d,\C)$ is called an \emph{infinite ground state} of chemical potential $\mu$ if it has finite local interaction energy in the sense of \eqref{eq:finite_interation_u} and it is a minimizer of the Gross-Pitaevskii free energy~\eqref{eq:cF_GP} under compact perturbations, in the sense that
\begin{multline}
\int_{\R^d}\big(|\nabla v(x)|^2-|\nabla u(x)|^2\big)\dx-\mu\int_{\R^d}(|v(x)|^2-|u(x)|^2\big)\dx\\
+\frac{1}{2} \iint_{\R^d\times\R^d}\Big(|v(x)|^2|v(y)|^2-|u(x)|^2|u(y)|^2\Big)w(x-y)\,\dx\,\dy\geq0, \label{eq:local_min}
\end{multline}
for all $v\in H^1_{\rm unif}(\R^d,\C)$ having finite local interaction energy and coinciding with $u$ outside of a bounded set.
\end{definition}

We recall that $H^1_{\rm unif}(\R^d)$ consists of functions $u\in H^1_{\rm loc}(\R^d)$ satisfying the uniform bound
$$\sup_{z\in \R^d}\|u\|_{H^1(B(z,1))} <\infty.$$
We will actually prove that any infinite ground state is uniformly bounded. Based on our analysis, it is possible to replace the requirement of $H^1_{unif}$ in Definition \ref{def:infinite_GP_GS} with $H^1_{\rm loc} \cap L^\infty$. However, we prefer to keep the former because it is more naturally related to the energy. 
 
As we have explained in the introduction, the property~\eqref{eq:local_min} means that the energy $\cF_\mu=\cF_{\mu,\R^d}$ in~\eqref{eq:cF_GP}
 increases when $u$ is locally perturbed. Although both $u$ and $v$ have an infinite energy, our assumptions on $u$ and $v$ guarantee that the energy difference in~\eqref{eq:local_min} is finite (see Lemma~\ref{lem:estim_potential} below for the details). We simply require this difference to be non-negative. Even if the condition~\eqref{eq:local_min} looks local (since $v$ is a compact perturbation of $u$), it contains a lot of global information since the support of $v-u$ can be arbitrarily large.

 Definition \ref{def:infinite_GP_GS} is very common in the context of the Ginzburg-Landau and Cahn-Hilliard equation $w=\delta$, see for instance~\cite{BreMerRiv-94}. This is also in the same spirit as the classical ground states discussed in~\cite{Radin-84,Radin-04,BelRadShl-10,Suto-05,Suto-11,Lewin-22}, the infinite Thomas-Fermi-type solutions studied in~\cite{BlaBriLio-03,CanEhr-11} and the signed mean-field minimizers studied in~\cite{GiuLebLie-09}.

We can reformulate~\eqref{eq:local_min} by saying that on any bounded domain $\Omega$, $u$ is a minimizer for the local minimization problem
\begin{equation}
\cF_{\mu,\Omega,u}(u)=\min_{\substack{ v\in H^1(\Omega)\\ v_{|\partial\Omega}=u_{|\partial\Omega}}}\cF_{\mu,\Omega,u}(v)
\label{eq:min_u_v}
\end{equation}
with
\begin{equation}
\cF_{\mu,\Omega,u}(v):=\cF_{\mu,\Omega}(v)+\iint_{\Omega \times (\R^d\setminus\Omega)}|v(x)|^2|u(y)|^2w(x-y)\,\dx\,\dy.
 \label{eq:local_energy_Omega}
\end{equation}
The additional term describes the interaction with the outside. The corresponding external potential $(\1_{\R^d\setminus\Omega}|u|^2)\ast w$ is often called a \emph{boundary condition}, although it actually penetrates in $\Omega$ (when $w$ is not proportional to a Dirac delta). Since $u\in H^1_{\rm unif}(\R^d)$ by assumption, it can be seen to decay fast with the distance to the boundary, so that it mainly affects the minimizer close to $\partial\Omega$. The interpretation of~\eqref{eq:min_u_v} is that $u$ minimizes the energy of a system immersed in an infinite bath $u\1_{\R^d\setminus \Omega}$ depending on itself, instead of hard walls for Dirichlet and a perfect insulator for Neumann boundary conditions. This is in the spirit of the famous Dobru\v{s}in-Lanford-Ruelle (DLR)~\cite{Dobrushin-68a,Dobrushin-68b,Dobrushin-69,LanRue-69} condition for equilibrium states of infinitely many classical particles. Any infinite ground state is automatically the thermodynamic limit of a sequence of finite volume ground states, solving~\eqref{eq:min_u_v}.

Let us now reformulate the infinite ground state property in the following lemma, which also provides first and second order conditions. In particular, infinite ground states have to solve the GP equation \eqref{eq:GP_infinite}.

\begin{lemma}[First and second order conditions for infinite ground states]\label{lem:1st_2nd_order} Let $d\geq1$ and $w$ satisfy Assumption~\ref{ass:w}.

\smallskip
\noindent
{\rm (i)}  Any infinite ground state $u$ as in Definition~\ref{def:infinite_GP_GS} solves the GP equation~\eqref{eq:GP_infinite} in the sense of distributions:
\begin{equation}
(-\Delta+w\ast |u|^2)u=\mu\,u \quad \text{ in } \mathcal{D}'(\R^d).
\end{equation}
Moreover, we have
\begin{equation} \label{eq:GP-ineq-0}
\pscal{h,\left(-\Delta+w\ast |u|^2-\mu\right)h}\\+2\iint_{\R^d\times\R^d}\!\!\Re(\overline{u}h)(x)\Re(\overline{u}h)(y)w(x-y)\,\dx\,\dy\geq0
\end{equation}
for all $h\in H^1(\R^d)$ having compact support and finite interaction energy. If $u$ is real-valued, the latter is equivalent to the positivity of the two operators
\begin{equation}
\begin{cases}
  -\Delta+w\ast |u|^2-\mu\geq0,\\
  -\Delta+w\ast |u|^2+2u(x)u(y)w(x-y)-\mu\geq0,
  \end{cases}
  \label{eq:linearly_stable}
\end{equation}
in the sense of quadratic forms like in~\eqref{eq:GP-ineq-0} with the same assumptions on $h$. It is here understood that $2u(x)u(y)w(x-y)$ is the kernel of the corresponding operator.

\smallskip
\noindent
{\rm (ii)} Conversely, if $u\in H^1_{\rm unif}(\R^d)$ is any solution to the GP equation~\eqref{eq:GP_infinite} with finite local interaction energy, then it is an infinite ground state if and only if
\begin{multline}
\pscal{h,\left(-\Delta+w\ast |u|^2-\mu\right)h}\\+\frac12\iint_{\R^d\times\R^d}(2\Re(\overline{u}h)+|h|^2)(x)(2\Re(\overline{u}h)+|h|^2)(y)w(x-y)\,\dx\,\dy\geq0
\label{eq:local_min_simplified_complex}
\end{multline}
for every complex-valued function $h\in H^1(\R^d,\C)$ having compact support and finite interaction energy.

\smallskip
\noindent
{\rm (iii)} If in addition $u$ is \textbf{real-valued} and \textbf{non-negative}, then $u$ is an infinite ground state if and only if \eqref{eq:local_min} holds for all $v\ge 0$ satisfying the same stated conditions. In this case, {\rm (ii)} remains valid with the additional condition $h\geq -u$ in~\eqref{eq:local_min_simplified_complex}. 

\smallskip
\noindent
{\rm (iv)} If $w$ has a non-negative Fourier transform,
$$\widehat w(k):=\frac1{(2\pi)^{\frac{d}{2}}}\int_{\R^d}w(x)e^{-ik\cdot x}\,\rd x\geq0,$$
and $u\in H^1_{\rm unif}(\R^d,\R)$ is any \textbf{real} solution to the GP equation~\eqref{eq:GP_infinite}, then it is an infinite ground state if and only if
\begin{equation} \label{eq:GP-ineq-1}
  -\Delta+w\ast |u|^2-\mu\geq0
\end{equation}
in the sense of quadratic forms.
\end{lemma}

\begin{proof}
For~(i) we write $v=u+h$ with $h$ of compact support in~\eqref{eq:local_min} and obtain~\eqref{eq:GP_infinite} from the linear term in $h$ and~\eqref{eq:GP-ineq-0} from the quadratic term. If $u$ is real-valued we get~\eqref{eq:linearly_stable} by writing $h=h_1+ih_2$. For~(ii) we write again $v=u+h$. For~(iv) we use that the last term in~\eqref{eq:local_min_simplified_complex} is always non-negative, after passing to Fourier coordinates. For (iii) we note that $|v|=|u|=u$ on the set where $v=u$, so that $|v|$ is an allowed trial function in~\eqref{eq:min_u_v}. We then use the diamagnetic inequality $\int_\Omega|\nabla v|^2\geq \int_\Omega|\nabla |v||^2$; see e.g.~\cite[Thm. 7.21]{LieLos-01}.
\end{proof}

Note that for $u\in H^1_{\rm unif}(\Omega)$, our assumptions on $w$ and~\eqref{eq:finite_interation_u} imply that
$$w_+\ast|u|^2\in L^1_{\rm unif}(\R^d),\qquad w_-\ast|u|^2\in L^\ii(\R^d),\qquad  u(w_+\ast|u|^2)\in L^1_{\rm loc}(\R^d)$$
(see Lemma~\ref{lem:estim_potential} below). This allows us to give a clear meaning to the GP equation \eqref{eq:GP_infinite} and to the quadratic form associated with the operator in~\eqref{eq:GP-ineq-1}. In low dimensions this is easy to check using the $H^1$ regularity of $u$.

Lemma~\ref{lem:1st_2nd_order} has several interesting immediate consequences. For instance, (iii) implies that positive ground states \emph{among real-valued functions} (as considered for instance in the Cahn-Hilliard model of phase transitions~\cite{CahHil-58}) are automatically also ground states within the class of complex-valued functions. This is wrong without the positivity condition! As we have mentioned in the introduction, the 1D kink solution $u(x)=\tanh(x/\sqrt2)$ in the case $w=\delta$ and $\mu=1$ is not an infinite ground state. This also follows from Theorem~\ref{thm:real-valued} below.

Another important consequence is stated in the following

\begin{corollary}[Constant solution for positive-definite $w$]\label{cor:cnst}
Let $d\geq1$ and $w$ satisfy Assumption~\ref{ass:w}. If $w$ has a non-negative Fourier transform, $\widehat w\geq0$,
then the constant solution~\eqref{eq:u_cnst_intro} is an infinite ground state for all $\mu>0$.
\end{corollary}

\begin{proof}
Recall that Assumption~\ref{ass:w} implies $\int_{\R^d}w>0$. The constant solution $u_0\equiv(\mu/\int _{\R^d}w)^{1/2}$ is a positive solution of the GP equation. Since $-\Delta+w\ast|u_0|^2-\mu=-\Delta\geq0$, one can apply (iv) whenever $\widehat{w}\geq0$.
\end{proof}

\subsection{Boundedness of solutions to the GP equation} \label{sec:main-result-1}

It is well known that solutions to the Ginzburg-Landau equation $-\Delta u=(\mu-|u|^2)u$ must satisfy $|u|\leq \sqrt\mu$ everywhere~\cite{BreMerRiv-94,HerHer-96,Farina-98}. The proof uses the maximum principle and it seems hard to extend it to the case of a general interaction potential $w$. The following, which is one of our main results, provides the expected uniform bounds, using only Assumption~\ref{ass:w} on $w$. This is a nonlinear equivalent of a famous result of Ruelle~\cite{Ruelle-70} in classical statistical mechanics, later extended to the quantum case in~\cite{EspNicPul-82,Park-84,Park-85}. If in addition $u$ is an infinite ground state, then there exists also a lower bound on the local mass, which implies that $u$ has positive density, as was mentioned in the introduction.

\begin{theorem}[Uniform bounds]\label{thm:uniform_bound}
Let $d\geq1$ and $w$ satisfy Assumption~\ref{ass:w}. Let $u\in H^1_{\rm unif}(\R^d,\C)$ have finite local interaction energy as in \eqref{eq:finite_interation_u}. Assume that $u$ solves the GP equation
\begin{equation}
-\Delta u+(w\ast|u|^2)u=\mu\, u
 \label{eq:GP_local}
\end{equation}
in the sense of distributions for some constant $\mu\ge 0$.

\medskip
\noindent
{\rm (i)} Then $u$ is real analytic and all its derivatives are bounded. We have
\begin{equation}\label{eq:pointwise_bound_main}
\norm{u}_{L^\ii} \leq C\sqrt\mu \left(1+\mu^{\frac{d}{4}}\right),
\end{equation}
and, more generally, for any multi-index $\alpha$ of length $|\alpha|$, 
\begin{equation}\label{eq:pointwise_bound_derivatives_main}
\norm{\partial^\alpha u}_{L^\ii} \leq C^{|\alpha|}\mu^{\frac{|\alpha|+1}{2}}\left(1+\mu^{\frac{d}4}\right)^{1+|\alpha|}. 
\end{equation}

\medskip
\noindent
{\rm (ii)} If $u\ge 0$, then  we have the lower bound
\begin{equation}
\inf_{x\in \R^d} u(x)\geq \dps C^{-1}e^{-C\mu^{\frac{d+2}{4}}}\sqrt\mu.
\label{eq:pointwise_lower_bound}
\end{equation}

\medskip
\noindent
{\rm (iii)} If $u$ is an infinite ground state, then we have the averaged lower bound
\begin{equation}
\inf_{z\in \R^d}\int_{B(z,\ell)}|u|^2\geq \left(\frac{\mu}{C}-\frac{C}{\ell^{2}}\right)_+\ell^d,\quad \forall \ell \ge C.
 \label{eq:lower_bound}
\end{equation}
Here the constant $C > 0$ in \eqref{eq:pointwise_bound_main}, \eqref{eq:pointwise_bound_derivatives_main}, \eqref{eq:pointwise_lower_bound}, \eqref{eq:lower_bound} depends on $w$ and $d$ only.
\end{theorem}

We have a similar result for solutions of the GP equation in a bounded domain $\Omega$, with Dirichlet, Neumann, or non-homogeneous boundary conditions. We give the precise statement later in Section~\ref{sec:proof_bounds}. In most cases, the constant $C$ depends on the regularity of the domain $\Omega$.

That the upper bound~\eqref{eq:pointwise_bound_main} involves the power $\mu^{d/4}$ at large $\mu$ is well known for solutions to Schrödinger equations (see for instance~\cite[Lem.~3.1]{Davies-74}). On the other hand, the lower bound~\eqref{eq:lower_bound} was announced in the introduction and it means that, for $\mu>0$, infinite ground states have positive density, that is, are ``infinitely extended''. In particular they cannot tend to 0 at infinity.

The proof of Theorem~\ref{thm:uniform_bound} is rather involved and it occupies the whole of Section~\ref{sec:proof_bounds}. The case $w_2\geq0$ is easier, which is very common in statistical mechanics~\cite[Sec.~4.5]{Ruelle} (the proof in this case is quickly explained in~\eqref{eq:local_bd_w_positive_proof} below). This includes the two examples~\eqref{eq:example_w_cnst} and~\eqref{eq:example_w_VdW}. By following arguments in~\cite{Rebenko-98,PetReb-07} one can also handle quite easily the case where $\kappa$ in~\eqref{eq:w2_regular} is much smaller than $\eps$ in~\eqref{eq:decomp_w}. The difficulty is to treat any $w_2$ without any assumption on the size of $\kappa$. This was done first by Ruelle in~\cite{Ruelle-70} in the framework of statistical mechanics and our proof is inspired by his approach.

\subsection{Existence, sign and complex phase of infinite ground states}  \label{sec:main-result-2}

In this subsection, we first explain how to construct infinite ground states in $\R^d$ as the limit of suitable minimizers in bounded sets, for any $w$ and $\mu>0$. We introduce the minimal grand-canonical Dirichlet energy
\begin{equation}
\boxed{F_{\rm D}(\mu,\Omega):=\min_{u\in H^1_0(\Omega)}\cF_{\mu,\Omega}(u)}
\label{eq:def_J}
\end{equation}
for any regular bounded domain $\Omega$ and any $\mu\in\R$. We define the Neumann energy $F_{\rm N}(\mu,\Omega)$ by assuming $u\in H^1(\Omega)$ instead of $H^1_0(\Omega)$, hence $F_{\rm N}(\mu,\Omega)\leq F_{\rm D}(\mu,\Omega)$. Similar minimal energies defined by minimizing $\cE_\Omega$ with a mass constraint will be discussed later in Section~\ref{sec:main-result-3}. Recall that $F_{\rm N}(\mu,\Omega)$ is bounded-below by~\eqref{eq:F_bounded_below}.

The existence of minimizers for $F_{\rm D/N}(\mu,\Omega)$ follows from standard techniques in the calculus of variations. It is also well known that those minimizers must be positive inside $\Omega$, up to a constant phase factor. In fact, the non-negativity follows from the diamagnetic inequality $|\nabla u| \ge |\nabla |u||$, and the strict positivity follows from the Gross--Pitaevskii equation; see e.g. \cite[Thm. 11.8]{LieLos-01} for relevant arguments. 

The following says that these minimizers locally converge to an infinite ground state in the thermodynamic limit $\Omega_n\nearrow\R^d$.

\begin{corollary}[Existence of real positive infinite ground states]\label{cor:existence-1}
Let $d\geq1$ and $w$ satisfy Assumption~\ref{ass:w}. Let $\mu>0$.
Let $\{\Omega_n\}$ be a sequence of smooth bounded domains such that
$$B_{R_n}\subset \Omega_n$$
for some $R_n\to\ii$. Let $0\leq u_n\in H^1(\Omega_n)$ be any minimizer of $F_{\rm D/N}(\mu,\Omega_n)$.  Then, $u_n$ is bounded in $L^\ii(B_{R_n/2})$ and, after extraction of a subsequence, converges uniformly locally to an infinite ground state $0<u\in H^1_{\rm unif}(\R^d,\R)$ as in Definition \ref{def:infinite_GP_GS}. In particular, infinite ground states exist for all $\mu>0$.
\end{corollary}

The proof of Corollary~\ref{cor:existence-1} can be found in Section~\ref{sec:proof_existence}. It uses some uniform bounds on $u_n$ in the domain $\Omega_n$ which we have not stated in Theorem~\ref{thm:uniform_bound}.

An important feature of the ground state solutions obtained by Corollary~\ref{cor:existence-1} is that they are \emph{real-valued} and \emph{positive}. This is a natural consequence of the Dirichlet and Neumann boundary conditions imposed on $F_{\rm D}(\mu,\Omega)$ and $F_{\rm N}(\mu,\Omega)$, respectively. We show here that all infinite ground states must be positive in 1D, that is, the $U(1)$ symmetry can never be broken. This is not true in higher dimensions, however.

\begin{theorem}[Positivity of infinite ground states in low dimensions]\label{thm:real-valued}
Let $d\geq1$ and $w$ satisfy Assumption~\ref{ass:w}. Let $\mu>0$.

\medskip
\noindent
{\rm (i)} In dimension $d=1$, all the infinite ground states must be real-valued and strictly positive, up to a constant phase factor.

\medskip
\noindent
{\rm (ii)} In dimension $d=2$, all the \textbf{real-valued} infinite ground states must be strictly positive, up to a sign.
\end{theorem}

The proof of this result can be read in Section~\ref{sec:proof_real-valued}. Note that such positive solutions must satisfy the lower bound~\eqref{eq:pointwise_lower_bound}. It is an interesting question to determine if real-valued infinite ground states are also strictly positive in dimensions $d\geq3$. Our proof is very specific to the small dimensions $d\in\{1,2\}$.

When $\widehat{w}\geq0$ we expect that a \emph{real-valued} infinite ground state must necessarily be constant for all $\mu$. We are only able to prove this in 1D or for small values of $\mu$ (see Theorem~\ref{thm:phase_transitions-low-density} below for the latter case).

\begin{theorem}[Uniqueness in 1D for positive-definite interactions]\label{thm:uniqueness_1D}
Let $d=1$ and $w$ satisfy Assumption~\ref{ass:w} with $s>d+1=2$. 
Assume further that $\hat w\ge 0$. Then, for any $\mu>0$, the infinite ground states are exactly the constant functions
$$u(x)=e^{i\theta_0}\left(\frac\mu{\int_\R w}\right)^{\frac12},\qquad \theta_0\in\R.$$
\end{theorem}

The proof of Theorem~\ref{thm:uniqueness_1D} can be read in Section~\ref{sec:proof_uniqueness_1D}. The idea is to show that any infinite ground state must have a finite (properly renormalized) energy and to then use strict convexity with respect to $|u|^2$.

On the contrary, in dimensions $d\geq2$, we expect the existence of complex-valued ground states with a non-trivial phase, for instance describing point vortices in 2D and vortex lines in 3D. This is well known for the Ginzburg-Landau equation ($w=\delta$)~\cite{BetBreHel-94} and is certainly expected for all $w$ and all $\mu$. The generality of the assumption on $w$ poses serious technical problems, however, and we are only able to prove this for $\mu$ small enough.

\begin{theorem}[Existence of vortex-type infinite ground states in 2D]\label{thm:exist_vortex}
Let $d=2$ and $w$ satisfy Assumption~\ref{ass:w} with $s>d+1=3$. Suppose in addition that its Fourier transform satisfies
\begin{equation}
 2\widehat{w}(k)-k\cdot\nabla_k\widehat{w}(k)\geq c|\widehat{w}(k)|^2,\qquad\forall k\in\R^d
 \label{eq:assumption_Fourier_w_Pohozaev}
\end{equation}
for some $c>0$. Then, for $\mu$ small enough there exists an infinite ground state $u$ in $\R^2$ with a non trivial phase. More precisely, we have $|u(x)|>0$ for all $|x|\geq R_0$ and $u$ takes the form
\begin{equation}
u(x)=|u(x)|\frac{x_1+ix_2}{\sqrt{x_1^2+x_2^2}}e^{i\psi(x)},\qquad\forall |x|\geq R_0,
 \label{eq:vortex_GS}
\end{equation}
with a smooth (single-valued) $\psi$. In other words, $u$ has topological degree one at infinity.
\end{theorem}

The condition~\eqref{eq:assumption_Fourier_w_Pohozaev} is well known and appears naturally in the Pohozaev identity~\cite{GinVel-80,DeLaire-12}. Note that it implies $\widehat{w}\geq0$.\footnote{If $\widehat{w}$ took negative values then it would attain its minimum at some $k_0$ since $\widehat{w}(0)>0$ and $w\to0$ at infinity. But $\nabla \widehat{w}(k_0)=0$ and thus $\widehat{w}(k_0)\geq0$ from~\eqref{eq:assumption_Fourier_w_Pohozaev}, a contradiction.} The condition~\eqref{eq:assumption_Fourier_w_Pohozaev} is for instance valid if $\widehat w$ is \emph{radial non-increasing} (and hence non-negative) with $c=2/\widehat{w}(0)=2/\|\widehat w\|_{L^\ii}>0$, since then $k\cdot\nabla \widehat{w}\leq0$.

We can conclude from Theorem~\ref{thm:exist_vortex} that there exist infinite ground states $u$ for which the modulus $|u|$ is not a ground state. In fact, for the vortex~\eqref{eq:vortex_GS}, $|u|$ is not even a solution of the GP equation.

The proof of Theorem~\ref{thm:exist_vortex} can be read in Section~\ref{sec:proof_vortex}. The reason for the constraint that $\mu$ is small enough comes from our proof being based on techniques developed in the 90s in the context of the Ginzburg-Landau equation ($w=\delta$)~\cite{BetBreHel-94,BreMerRiv-94,Sandier-98b,Shafrir-94,Mironescu-96}, which is recovered in the low density limit $\mu\to0$. Note, in particular, that $w=\delta$ satisfies~\eqref{eq:assumption_Fourier_w_Pohozaev} since $\widehat{w}$ is constant hence $\nabla_k\widehat w=0$.

In the (2D) Ginzburg-Landau case it was proved in~\cite{Mironescu-96,Sandier-98} that there are exactly two infinite ground states, up to translations, complex conjugation and a global constant phase factor. We conjecture that the same property holds for a general positive-definite $w$, at least when $\mu\ll1$.

\subsection{Thermodynamic properties}  \label{sec:main-result-3}

In Section \ref{sec:main-result-2}, we have explained how to construct infinite ground state solutions as thermodynamic limits of minimizers in compact sets. Now we consider the reverse direction and show that every infinite ground state as in Definition \ref{def:infinite_GP_GS} exhibits \emph{universal} thermodynamic properties.

\subsubsection{Canonical and grand-canonical energies}
We have introduced in~\eqref{eq:def_J} the grand-canonical minimal energy $F_{\rm D/N}(\mu,\Omega)$ in a bounded domain $\Omega$ with a fixed chemical potential $\mu$. We can similarly introduce the \emph{minimal canonical Dirichlet energy}
\begin{equation}
\boxed{E_{\rm D}(\lambda,\Omega):=\min_{\substack{u\in H^1_0(\Omega)\\ \int_\Omega|u|^2=\lambda}}\cE_\Omega(u).}
 \label{eq:def_I}
\end{equation}
Recall from~\eqref{eq:cE_GP} that $\cE_\Omega=\cF_{0,\Omega}$ is the GP energy with $\mu=0$. The Neumann energy $E_{\rm N}(\lambda,\Omega)$ is defined similarly  with $u\in H^1(\Omega)$.

It is well known that the minimal energies per unit volume converge in the thermodynamic limit~\cite{Ruelle}. To be precise, let $\Omega_n\subset\R^d$ be a sequence of domains which converges to $\R^d$ in the sense of Fischer, that is, satisfies
\begin{equation}\label{eq:Fisher-0}
B_{R_n}\subset \Omega_n\subset B_{CR_n}
\end{equation}
and
\begin{equation}
\left|\left\{x\in \R^d\ :\ \rd(x,\partial\Omega_n)\leq \ell\right\}\right|\leq C R_n^{d-1}\ell,\qquad\forall \ell\leq R_n
\label{eq:Fisher}
\end{equation}
for some $C>0$ and $R_n\to\ii$. Then for all $\rho,\mu>0$, the following limits exist and are independent of the sequence $\{\Omega_n\}$ (as well as the constant $C$ in \eqref{eq:Fisher-0} and \eqref{eq:Fisher})
\begin{equation}
e(\rho)=\lim_{\substack{n\to\ii\\ \frac{\lambda_n}{|\Omega_n|}\to\rho}}\frac{E_{\rm D}(\lambda_n,\Omega_n)}{|\Omega_n|}=\lim_{\substack{n\to\ii\\ \frac{\lambda_n}{|\Omega_n|}\to\rho}}\frac{E_{\rm N}(\lambda_n,\Omega_n)}{|\Omega_n|},
\label{eq:def_e_lambda}
\end{equation}
\begin{equation}
f(\mu)=\lim_{n\to\ii}\frac{F_{\rm D}(\mu,\Omega_n)}{|\Omega_n|}=\lim_{n\to\ii}\frac{F_{\rm N}(\mu,\Omega_n)}{|\Omega_n|}.
\label{eq:def_f_mu}
\end{equation}
For the convenience of the reader, we quickly outline the argument in Appendix~\ref{app:existence-thermo}. Now we state some important properties of the functions $\rho\mapsto e(\rho)$ and $\mu\mapsto f(\mu)$.

\begin{theorem}[Thermodynamic limit]\label{thm:thermo}
Let $d\geq1$ and $w$ satisfy Assumption~\ref{ass:w}.
The function $\rho\mapsto e(\rho)$ is positive, increasing, $C^1$ and convex, whereas the function $\mu\mapsto f(\mu)$ is negative, decreasing, continuous and strictly concave. They are Legendre transforms to each other:
\begin{equation}
e(\rho)=\max_{\mu\geq0}\big\{f(\mu)+\mu\rho\big\},\qquad f(\mu)=\min_{\rho\geq0}\big\{e(\rho)-\mu\rho\big\}.
\label{eq:Legendre}
\end{equation}
For any $\rho>0$, the maximum 
in ~\eqref{eq:Legendre} is attained for $\mu(\rho)=e'(\rho)$. The function $\rho\mapsto \mu(\rho)$ is continuous non-decreasing and satisfies
\begin{equation}
C^{-1}\rho\leq \mu(\rho) \leq C\rho
\label{eq:bound_mu_rho}
\end{equation}
for a constant $C>0$. For any $\mu$, the minimum on the right of~\eqref{eq:Legendre} is attained for all the $\rho$'s such that $\mu(\rho)=\mu$, that is, $\rho$ in the interval $[\rho_-(\mu),\rho_+(\mu)]$ with $\rho_\pm(\mu)=-f'_\pm(\mu)$ (left and right derivative of $f$ at $\mu$). Finally, the function $\rho\mapsto e(\rho)/\rho$ is concave non-decreasing.
\end{theorem}

The proof is provided in Section~\ref{sec:proof_thermo}.

It is clear from the definition that $\mu\mapsto F_{\rm D/N}(\mu,\Omega)$ is a concave function as an infimum of linear functions. If $\widehat{w}\geq0$ one can show that $\lambda \mapsto E_{\rm D/N}(\lambda,\Omega)$ is convex. For general potentials $w$, $E_{\rm D/N}(\lambda,\Omega)$ need not be convex, however. By Theorem~\ref{thm:thermo} we see that it always becomes convex in the thermodynamic limit.

The most difficult part of the theorem is the strict concavity of $f$, which plays an important role since it allows us to show that $e$ is $C^1$, hence to define the map $\rho\mapsto\mu(\rho)$. More precisely, recall that since $e$ is convex it possesses left and right derivatives at any point, which coincide except possibly on a countable set. The theorem says that there can be no jump. The reason is that its Legendre transform $f$ is strictly concave. Recall that a jump in the derivative of a convex function corresponds to a linear part in its Legendre transform. In statistical mechanics, the strict concavity of the free energy is related to the positivity of the bulk compressibility of an infinite gas~\cite{Ginibre-67,Ruelle-70}.

We see no particular reason why $f$ would itself have to be $C^1$ everywhere. We thus have to work with the two functions $\rho_\pm(\mu)$ which correspond to the left and right derivatives of $f$ as defined in the theorem. From the concavity, we have of course $\rho_-(\mu)=\rho_+(\mu)$ everywhere except on a countable set. When $\rho_-(\mu)<\rho_+(\mu)$ we interpret this as a phase transition for the corresponding $\mu$, with two (pure) phases of different densities. One can then expect the existence of (mixed) infinite ground states, for instance describing an interface between the two pure phases.

Note that the energy per unit volume can then be seen as a function of $\mu$ via
\begin{equation}
 \tilde e_\pm(\mu):=e\big(\rho_\pm(\mu)\big).
\end{equation}

\subsubsection{Equivalence of ensembles}
In Corollary~\ref{cor:existence-1}, we have constructed infinite ground states using a grand-canonical thermodynamic limit, that is, fixing $\mu$. 
Here we prove that if we fix the density $\rho$ instead of $\mu$, then we obtain infinite ground states with the corresponding $\mu=\mu(\rho)$. 

\begin{theorem}[Limit of canonical minimizers]\label{thm:canonical}
Let $d\geq1$ and $w$ satisfy Assumption~\ref{ass:w}. Let $\rho>0$ and
let $\{\Omega_n\}$ be a sequence of smooth bounded domains such that $B(0,R_n)\subset \Omega_n$ for some $R_n\to\ii$.
Let $\lambda_n$ be such that $\lambda_n/|\Omega_n|\to\rho$. Let finally $u_n$ be any minimizer for $E_{\rm D/N}(\lambda_n,\Omega_n)$. Then, in the thermodynamic limit $n\to\ii$ and up to extraction of a subsequence, $u_n$ converges locally uniformly to an infinite ground state of chemical potential $\mu=\mu(\rho)$.
\end{theorem}

The theorem provides the \emph{equivalence of ensembles}, in the sense that we always obtain a \emph{grand-canonical} infinite ground state, even starting with the \emph{canonical} problem. The reader could have expected that limits of canonical minimizers should have an additional constraint of the form $\int_{\R^d}(|v|^2-|u|^2)=0$ in~\eqref{eq:local_min} or~\eqref{eq:min_u_v} but our result shows that there is no such constraint. Our proof is inspired by another famous work in statistical mechanics, by Dobru\v{s}in and Minlos~\cite{DobMin-67}, and the details can be found in Section~\ref{sec:proof_canonical}.

\subsubsection{Thermodynamic properties of infinite ground states}

We conclude this section by providing the announced universal thermodynamic properties of infinite ground states. In particular, we prove that they all have the same density $\rho$ and energy per unit volume $e(\rho)$. This is sometimes called ``equidistribution of mass and energy''~\cite{RotSer-15,PetRot-18}.

\begin{theorem}[Thermodynamic properties of infinite ground states]\label{thm:prop_GS}
Let $d\geq1$ and $w$ satisfy Assumption~\ref{ass:w}.
Let $u$ be an infinite ground state as in Definition~\ref{def:infinite_GP_GS}, for some $\mu>0$. Then we have
\begin{equation}
\lim_{R\to\ii}\sup_{\tau\in\R^d}\left|\frac{\cF_{\mu,B(\tau,R)}(u\1_{B(\tau,R)})}{|B_R|}-f(\mu)\right|=0
 \label{eq:local_free_energy}
\end{equation}
for its local free energy. Its mass satisfies
\begin{multline*}
\rho_-(\mu)\leq \liminf_{R\to\ii}\inf_{\tau\in \R^d}\frac1{|B_R|}\int_{B(\tau,R)}|u|^2\\
\leq\limsup_{R\to\ii}\sup_{\tau\in\R^d}\frac1{|B_R|}\int_{B(\tau,R)}|u|^2\leq\rho_+(\mu).
\end{multline*}
For its energy, we have
\begin{multline*}
\tilde e_-(\mu)\leq \liminf_{R\to\ii}\inf_{\tau\in\R^d}\frac{\cE_{B(\tau,R)}(u\1_{B(\tau,R)})}{|B_R|}\\
\leq \limsup_{R\to\ii}\sup_{\tau\in\R^d}\frac{\cE_{B(\tau,R)}(u\1_{B(\tau,R)})}{|B_R|}\leq \tilde e_+(\mu).
\end{multline*}
Finally, its average momentum per unit volume vanishes:
\begin{equation}
\lim_{R\to\ii} \sup_{\tau \in \R^d}|B_R|^{-1}\left|\int_{B(\tau,R)}\overline{u(y)}\nabla u(y)\,\dy\right|=0.
 \label{eq:momentum}
\end{equation}
\end{theorem}

The proof of the theorem, which is provided in Section~\ref{sec:prop_GS}, relies on the fact that $u$ is an exact minimizer of Problem~\eqref{eq:min_u_v} involving the external bath depending on itself. It is thus an approximate minimizer for the Dirichlet and Neumann problems, and must therefore satisfy similar properties as the exact minimizers of these problems.

If $f$ is $C^1$ at $\mu$, then we obtain
$$\lim_{R\to\ii}\sup_{\tau\in\R^d}\left|\frac1{|B_R|}\int_{B(\tau,R)}|u|^2-\rho\right|=\lim_{R\to\ii}\sup_{\tau\in\R^d}\left|\frac{\cE_{B(\tau,R)}(u\1_{B(\tau,R)})}{|B_R|}-e(\rho)\right|=0$$
where $\rho=\rho(\mu)=-f'(\mu)$.

Theorem~\ref{thm:prop_GS} says that any infinite ground state has to have the free energy per unit volume $f(\mu)$, whereas Corollary~\ref{cor:cnst} states that constant functions are always infinite ground states when $\widehat w\geq0$. Since it is easy to compute the free energy of the constant, we immediately obtain the following.

\begin{corollary}[Specific energies for positive-definite $w$]\label{cor:cnst_e_f}
Let $d\geq1$ and $w$ satisfy Assumption~\ref{ass:w}. If $w$ has a non-negative Fourier transform, $\widehat w\geq0$, then
\begin{equation}
f(\mu)=-\frac{\mu^2}{2 \int_{\R^d}w},\qquad e(\rho)=\frac{\rho^2}{2}\int_{\R^d}w,\qquad \mu(\rho)=\rho\int_{\R^d}w
\label{eq:f_e_cnst}
\end{equation}
for all $\mu,\rho>0$.
\end{corollary}

We see that there are no special phase transition visible at the level of the energies when $\widehat w\geq0$. It is not difficult to give a direct proof of the corollary, by looking at the Neumann problem in a bounded domain $\Omega$ and passing to the thermodynamic limit. Note, however, that the constant function is not an exact ground state in a ball since $\1_\Omega\ast w$ is not exactly constant over $\Omega$.

\subsection{Phase transition}  \label{sec:main-result-4}

Now we are ready to discuss the macroscopic breaking of translational symmetry and the possibility of having infinite ground states which are really far from a constant, in modulus.

We recall that the constant function
\begin{equation}
\boxed{u_{\rm cnst}(x)=\left(\frac{\mu}{\int_{\R^d}w}\right)^{\frac12}}
\label{eq:cnst}
\end{equation}
is interpreted as a \emph{fluid equilibrium state} of the GP equation. The latter has the free energy, density and energy per unit volume given by
\begin{equation}
f_{\rm cnst}(\mu)=-\frac{\mu^2}{2\int_{\R^d} w},\qquad \rho_{\rm cnst}(\mu)=\frac{\mu}{\int_{\R^d}w},\qquad e_{\rm cnst}(\rho)=\frac{\rho^2}{2}\int_{\R^d} w.
\label{eq:cnst_f_e}
\end{equation}
Recall that $\int_{\R^d}w>0$ under Assumption~\ref{ass:w}. Plugging the constant function in the Neumann problem and taking the thermodynamic limit, we obtain
$$f(\mu)\leq f_{\rm cnst}(\mu),\qquad e(\rho)\leq e_{\rm cnst}(\rho)$$
for all $\mu,\rho>0$. We will first investigate whether there is equality, before we ask whether the constant function $u_{\rm cnst}$ is actually an infinite ground state.

It turns out that there can be only one phase transition. Once $e(\rho)$ starts to be strictly below $e_{\rm cnst}(\rho)$, it stays below it for all larger $\rho$'s. The system can never go back to a fluid phase when the density is increased. To see this, we introduce the function
\begin{equation}
\phi(\rho):=\frac{e(\rho)}{\rho}-\frac{\rho}{2}\int_{\R^d}w
\label{eq:def_phi_critical_rho}
\end{equation}
which is concave and non-positive with $\phi(0)=0$, hence also non-increasing, by Theorem~\ref{thm:thermo}.

\begin{definition}[Critical density for the fluid-solid transition]\label{def:rho_c}
We define $\rho_c\in[0,\ii]$ to be the largest constant such that $\phi\equiv0$ on $[0,\rho_c)$. Correspondingly, we define
$\mu_c:=\rho_c\int_{\R^d}w$.
\end{definition}

Since $\phi<0$ on $(\rho_c,\ii)$, we deduce immediately that
\begin{equation}
e(\rho)-e_{\rm cnst}(\rho)\begin{cases}
=0&\text{for $\rho\leq\rho_c$,}\\
<0&\text{for $\rho>\rho_c$,}
\end{cases}\quad
f(\mu)-f_{\rm cnst}(\mu)\begin{cases}
=0&\text{for $\mu\leq\mu_c$,}\\
<0&\text{for $\mu>\mu_c$,}
\end{cases}
 \label{eq:f_not_cnst}
\end{equation}
after taking the Legendre transform.
By Theorem~\ref{thm:thermo} we know that the free energy per unit volume of an infinite ground state must be $f(\mu)$. For the constant function~\eqref{eq:cnst} this is $f_{\rm cnst}(\mu)$, hence the strict inequality in~\eqref{eq:f_not_cnst} implies that \emph{$u_{\rm cnst}$ is not an infinite ground state for $\mu>\mu_c$}.

By concavity we must have $\phi'<0$ on the interval $(\rho_c,\ii)$ and after differentiation, this implies
\begin{equation}
\mu(\rho)=e'(\rho)=\rho\phi'(\rho)+\frac{e(\rho)}{\rho}+\frac\rho2\int_{\R^d}w<\rho\int_{\R^d}w,\qquad\forall \rho>\rho_c,
\label{eq:strict_phi_super_critical}
\end{equation}
by~\eqref{eq:f_not_cnst}. Hence the chemical potential satisfies
\begin{equation}
\mu(\rho)-\rho\int_{\R^d}w\begin{cases}
=0&\text{for $\rho\leq\rho_c$,}\\
<0&\text{for $\rho>\rho_c$.}
\end{cases}
\label{eq:critical_mu}
\end{equation}
In other words, after the transition the density is larger than that of the fluid, for a given $\mu$.

Note that there are, in fact, two critical densities. In Definition~\ref{def:rho_c} we have introduced the \emph{left} critical density
$$\rho_c=\rho_-(\mu_c)=-f'_-(\mu_c)=\mu_c\left(\int_{\R^d}w\right)^{-1}.$$
The \emph{right} critical density
$$\rho'_c=\rho_+(\mu_c)=-f'_+(\mu_c)$$
is also the largest number such that $\rho\mapsto e(\rho)$ is linear on the interval $[\rho_c,\rho_c']$, precisely:
$$e(\rho)=\left(\rho-\frac{\rho_c}{2}\right)\rho_c\int_{\R^d}w,\qquad \forall \rho\in[\rho_c,\rho_c'].$$
When $\rho'_c>\rho_c$ this is called a \emph{first order phase transition} (the derivative of $f$ has a jump). In general we know that $f$ ceases to be a real-analytic function at $\mu_c$, otherwise it would coincide with $f_{\rm cnst}$ on the right of $\mu_c$. We say that the phase transition is of order $N$ if $f$ admits $N$ derivatives on the right of $\mu_c$ and all the derivatives $f^{(n)}$ are continuous at $\mu_c$ for $n<N$ whereas $f^{(N)}$ is not. A transition of infinite order is in principle possible.

Now that we know that there can be \emph{at most one} fluid-solid phase transition, we ask whether it exists or not. The easiest case is when $w$ is positive-definite, $\widehat w\geq0$. Then we already know from Corollary~\ref{cor:cnst_e_f} that $\mu_c=\rho_c=+\ii$, hence the system is always a fluid. We thus have to study the case where $\widehat w$ changes sign. In this case we will prove that the critical $\mu_c$ and $\rho_c$ are always \emph{finite and strictly positive}.

To determine whether the constant function could be an infinite ground state, a natural idea is to look at its linear (in)stability. Of course, linear instability implies global instability but the converse is not necessarily true. For the (real-valued) constant solution $u_{\rm cnst}$, linear stability amounts to the conditions
\begin{equation}
\begin{cases}
  -\Delta+w\ast u_{\rm cnst}^2-\mu=-\Delta \geq0,\\
  -\Delta+2\frac{\mu}{\int_{\R^d}w}w(x-y)\geq0,
  \end{cases}
  \label{eq:linearly_stable2}
\end{equation}
from~\eqref{eq:linearly_stable} in Lemma~\ref{lem:1st_2nd_order}. The first condition is always true and the second can be reformulated as
\begin{equation}
|k|^2+\frac{2\mu}{\widehat{w}(0)}\widehat{w}(k)\geq0,\qquad\forall k\in\R^d,
\label{eq:linearization_cnst_Fourier}
\end{equation}
after passing to the Fourier domain. This is obviously satisfied for all $\mu$ when $\widehat{w}\geq0$, which we already knew from Corollaries~\ref{cor:cnst} and~\ref{cor:cnst_e_f}. On the other hand, this is trivially wrong for a large enough $\mu$ when $\widehat{w}$ changes sign. This fact is well known and has appeared many times in the physical literature~\cite{Gross-57,KirNep-71,Nepomnyashchii-71,NepNep-71}. Similar results in the classical case can be read in~\cite{GatPen-69, Gates-72,ButLeb-05}. Our conclusion is that the constant function cannot be an infinite ground state for large $\mu$ if $\widehat{w}(k)$ changes sign.

The following theorem is our main result about the phase transition. It  gives explicit estimates for $\mu_c$ and more information about infinite ground states below and above $\mu_c$.

\begin{theorem}[Fluid-solid phase transition]\label{thm:phase_transitions}
Let $d\geq1$ and $w$ satisfy Assumption~\ref{ass:w}. Let $\mu_c$ be as in Definition \ref{def:rho_c}.

\smallskip
\noindent
{\rm (i)} We have $\mu_c=+\ii$ if and only if $\widehat{w}\geq0$. More precisely, if $\widehat{w}_-\neq0$, then
\begin{equation}
\mu_c\leq \min_{k\in\R^d}\frac{|k|^2\widehat{w}(0)}{2\widehat{w}_-(k)} < \infty.
\label{eq:estim_mu_c}
\end{equation}

\medskip
\noindent
{\rm (ii)} We always have $\mu_c>0$. More precisely, we have the explicit lower bound
\begin{equation}
\mu_c \ge \inf_{k\in \R^d} \frac{ |k|^2 \hat w(0)}{2^{\frac52} \Big( (1+\alpha) |\hat w(0) - \hat w(k)| -  \eps (2\pi)^{-\frac{d}2}  \big(1-  \frac{|k|^2r^2}6 \big)_+^2  \Big)_+}>0
\label{eq:lower_bd_mu_c}
\end{equation}
where $\eps,r$ are the constants appearing in Assumption~\ref{ass:w} and where $\alpha:=(\int_{\R^d}|w|)^{1/2} \max  \big(\eps^{-1/2},r\big).$

\medskip
\noindent
{\rm (iii)} If $\mu\le \mu_c$, then the constant function \eqref{eq:cnst} is an infinite ground state. Moreover, for any infinite ground state $u$ for $\mu<\mu_c$ and for almost every direction $\omega\in\bS^{d-1}$, there exist a constant $\theta_0\in\R$ and two sequences $s_n,r_n\to \ii$ such that
\begin{equation}
\lim_{n\to \infty} \norm{u-e^{i\theta_0}u_{\rm cnst}}_{L^\ii(B(s_n \omega,r_n))} =0.
\label{eq:min_max_const}
\end{equation}

\medskip
\noindent
{\rm (iv)} Any infinite ground state $u$ for $\mu>\mu_c$ satisfies
\begin{equation}
\liminf_{R\to\ii}\inf_{\tau\in \R^d}\left(\max_{B(\tau,R)}|u|^2-\min_{B(\tau,R)}|u|^2\right)^2\geq \frac{\rho}{\int_{\R^d}|w|}\left(\rho\int_{\R^d}w-\mu\right)>0
\label{eq:min_max_periodic}
\end{equation}
where $\rho=\rho_-(\mu)=-f'_-(\mu)$.
\end{theorem}

The proof is provided in Section~\ref{sec:fluid-solid transition}. On the right side of~\eqref{eq:estim_mu_c} and~\eqref{eq:lower_bd_mu_c}, we use the convention that the function equals $+\ii$ whenever the denominator vanishes.

The heuristic argument leading to (i) has been already explained. The upper bound \eqref{eq:estim_mu_c} essentially follows from the failure of the stability condition \eqref{eq:linearization_cnst_Fourier}. The positive lower bound in (ii) is important as it states that the system is always a fluid at low-enough density. We have provided the explicit bound~\eqref{eq:lower_bd_mu_c} for concreteness. It is not optimal at all and, in fact, our proof gives a slightly better bound which is harder to state. For $k$ large enough the function in the infimum behaves as
$$\frac{|k|^2}{4\sqrt2(1+\alpha)}\underset{|k|\to\ii}{\longrightarrow}+\ii$$
and this is why the infimum is attained and strictly positive.

The statement in (iii) that the constant function $u_{\rm cnst}$ is always an infinite ground state for $\mu\leq\mu_c$ confirms that the system is a fluid in this regime. It also implies that we could have equivalently defined $\mu_c$ as the largest $\mu$ for which the constant function is an infinite ground state, instead of only looking at the energy. Recall that the constant function is not expected to be the  only  ground state. There might exist other ones, for instance describing a vortex in 2D or a vortex line in 3D. The property~\eqref{eq:min_max_const} means that all infinite ground states (for $\mu<\mu_c$) must be ``very flat'' in some sufficiently large regions of space, in almost all possible directions (in dimension $d=1$, the limit~\eqref{eq:min_max_const} just holds for $\omega\in\{\pm1\}$). This is very natural for vortex points where typically $u\to\sqrt\rho e^{i\theta_0}$ at infinity in the direction where the polar angle is $\theta=\theta_0$. For vortex lines in 3D, the function also looks constant far away from the rotation axis. Our proof of (iii) in fact goes by first showing~\eqref{eq:min_max_const}, from which we can deduce by local convergence that the constant function $u_{\rm cnst}$ is an infinite ground state.

In (iv), the positivity of the right side of \eqref{eq:min_max_periodic} follows from~\eqref{eq:critical_mu}. Our estimate~\eqref{eq:min_max_periodic} shows that for $\mu>\mu_c$ the maximum and minimum of $|u|$ differ by a positive amount over any large enough ball. In other words an infinite ground state must oscillate everywhere in space. Note that in the proof we get a similar lower bound on the variance and kinetic energy in any large enough ball (see Lemma~\ref{lem:variance_kinetic}).

Theorem~\ref{thm:phase_transitions} says nothing about the behavior of infinite ground states at $\mu=\mu_c$, which depends on whether $f$ is $C^1$ at $\mu_c$ or not. If $f'_-(\mu_c)=f'_+(\mu_c)$ then we can prove that the property~\eqref{eq:min_max_const} holds for any ground states at $\mu=\mu_c$ as well (see Lemma~\ref{lem:variance_kinetic} below). If $f'$ has a jump at $\mu_c$, then there will be ground states satisfying~\eqref{eq:min_max_const} (e.g. the constant function $u_{\rm cnst}$) and others satisfying~\eqref{eq:min_max_periodic} (e.g. ground states obtained in the limit $\mu\to\mu_c^+$).

A very interesting problem is to find the precise shape of infinite ground states for $\mu>\mu_c$. Numerical simulations in~\cite{PomRic-93,JosPomRic-07b,SepJosRic-10,HenNatPoh-10,KunKat-12,MasJos-13,PreSerBru-18,PolBaiFerBla-24,Lewin-15} for different potentials indicate that they must be periodic in dimensions $d\in\{1,2,3\}$, a property which seems very hard to establish rigorously but would have an important meaning for the theory of supersolids. This is one more instance of the famous {\em crystallization conjecture}~\cite{BlaLew-15}. Convergence to a periodic function has been proved in the large density limit in~\cite{AftBlaJer-07,AftBlaJer-09} for the interaction potential $w=\1_{B_1}$ in dimension $d=1$. Some simple numerical simulations confirming the periodicity in dimension $d=1$ are provided in Figure~\ref{fig:numerics} below.

Another interesting question is to determine the nature of the phase transition. Should $\mu_c$ be equal to the right side of~\eqref{eq:estim_mu_c}, that is, exactly correspond to the value of $\mu$ at which the constant solution becomes unstable, it would be natural to expect a \emph{second-order phase transition} with oscillations developing progressively on top of the constant solution. This is what we seem to observe in the one-dimensional numerical simulations of Figure~\ref{fig:numerics} and what was predicted in the Physics literature as well~\cite{KirNep-71}. In dimension $d=2$, recent simulations in~\cite{MacMacCinPoh-13,PolBaiFerBla-24} rather predict the existence of a \emph{first order transition} where a periodic solution suddenly becomes more favorable than the constant, hence $f'$ has a jump at $\mu=\mu_c$.

We emphasize that there could exist phase transitions within the solid phase~\cite{ZhaMauPoh-19,RipBaiBla-23}. For instance, the phase diagram of Hartree-Fock Jellium (a model describing electrons interacting with the Coulomb potential) was thoroughly studied numerically in~\cite{BagDelBerHol-13,BagDelBerHol-14,BerDelDunHol-08,BerDelHolBag-11} and many transitions were found, including exotic ``incommensurate'' solid phases. The same could happen in our nonlinear model. The large density limit $\mu\to\ii$ is further studied in Section~\ref{sec:high-density} below.

\begin{remark}[Structure of the set of infinite ground states, pure phases]
In statistical mechanics, Gibbs states form a convex set, of which the extreme points have decaying correlations and are interpreted as ``pure phases''~\cite{Ruelle}. Due to the nonlinearity, the convexity does not hold in Gross-Pitaevskii theory where  infinite ground states rather form a kind of manifold. For instance, in the Ginzburg-Landau case $w=\delta_0$ the minimizers are exactly known in 2D and they take the form
\begin{equation}
u(x)=e^{i\theta_0}\quad\text{or}\quad e^{i\theta_0}f(x-x_0) \quad\text{or}\quad  e^{i\theta_0}\overline{f(x-x_0)},
 \label{eq:GL}
\end{equation}
for all $x_0\in\R^2$ and $\theta_0\in[0,2\pi)$, where $f$ describes one vortex at the origin, turning clockwise. Since $f$ converges to a constant at infinity in any fixed direction $\omega\in\bS^1$, we can compactify  and think that the three cases belong to the same ``manifold'' of solutions. With a general interaction $w$ we expect a similar picture in the fluid phase (as we have seen in Theorem~\ref{thm:exist_vortex}), except that several vortices could bind. In Ginzburg-Landau theory, only isolated vortices of topological degree one are stable.

In the solid phase $\mu>\mu_c$ an even more complicated structure should arise. First, vortices have recently been observed in experiments within the supersolid phase for dipolar gases~\cite{Casotti_etal-24}. In addition, there could exist infinite ground states describing an interface between two periodic solutions (one being shifted and/or rotated with respect to the other). Investigating the existence or non-existence of such ground states is an interesting open problem.

Finally, it would be interesting to find an appropriate concept of ``pure phases'' in Gross-Pitaevskii theory. One possible definition of a pure state would be to require that it arises ``at infinity'' from other infinite ground states, in directions $\omega\in\bS^{d-1}$ belonging to a set of positive measure on the sphere, modulo symmetries. With this definition, the statement \textnormal{(iii)} of Theorem~\ref{thm:phase_transitions} means that the constant function $u_{\rm cnst}$ is a pure phase in the fluid region. We conjecture it is the only one satisfying this property (modulo phase), so that the other ground states are only due to the breaking of $U(1)$ symmetry. A similar phenomenon is expected with periodic functions in the solid phase.
\end{remark}

\begin{figure}[t]
\centering
\setlength{\tabcolsep}{0pt}
\captionsetup{width=\linewidth}
\begin{tabular}{cc}
\includegraphics[width=6.4cm]{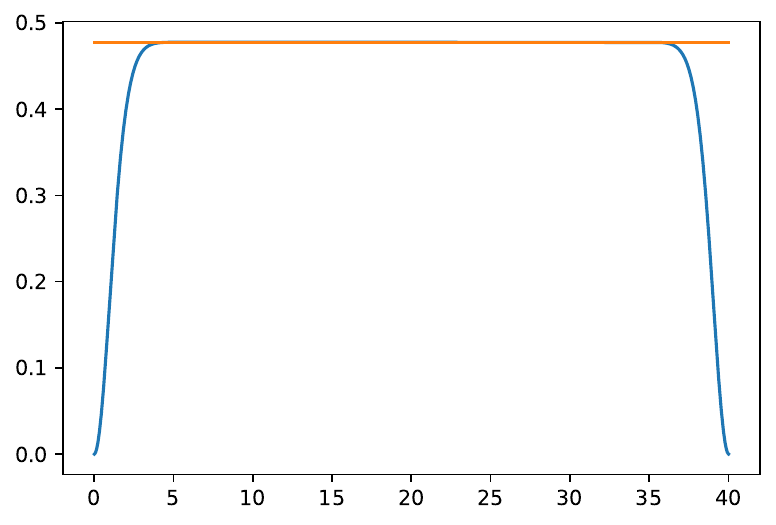}&\includegraphics[width=6.4cm]{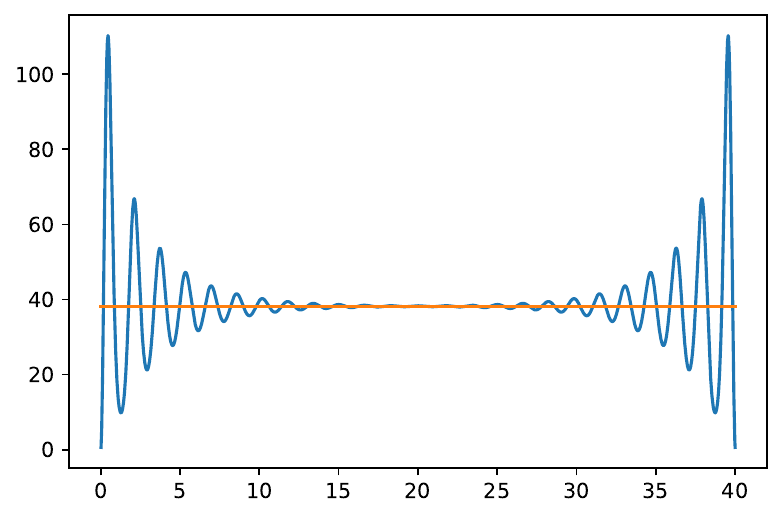}\\
\footnotesize$\mu=1$, $\rho\approx 0.47$ & \footnotesize $\mu=80$, $\rho\approx 39$\\
\includegraphics[width=6.4cm]{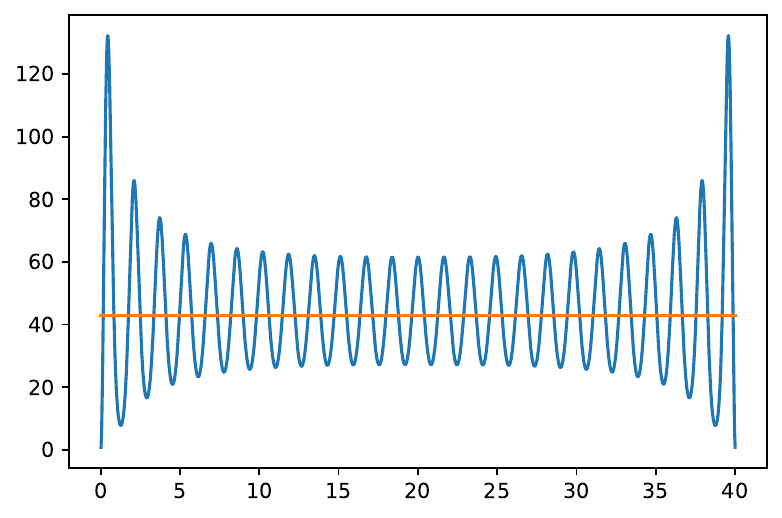}&\includegraphics[width=6.4cm]{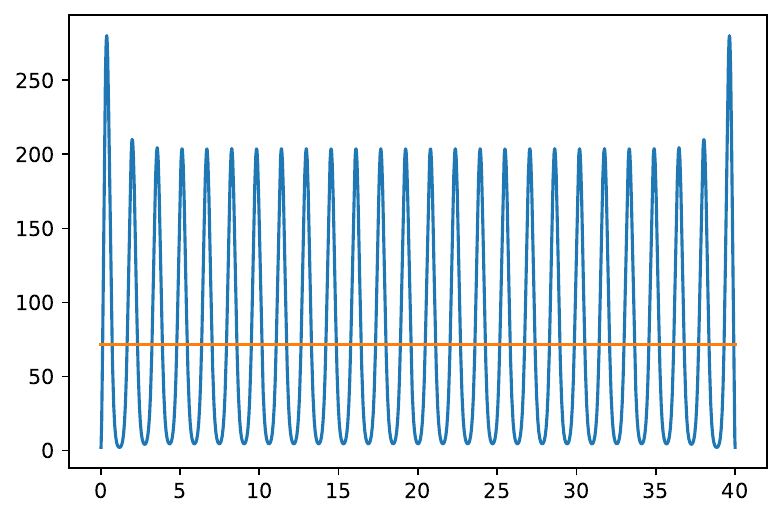}\\
\footnotesize$\mu=90$, $\rho\approx 44$&\footnotesize $\mu=150$, $\rho\approx 77$\\
\end{tabular}

\caption{\small Plot of $|u|^2$ for $u$ the solution to the Dirichlet problem in dimension $d=1$ with the interaction potential $w(x)=(1+x^6)^{-1}$ in the interval $\Omega=(0,40)$ and the mentioned values of $\mu$. Finite differences were used with 50 discretization points per unit length. The displayed horizontal line is the constant solution $u^2_{\rm cnst}=\mu/\int_\R w$ and the mentioned effective density is $\rho=\int_0^{40}|u|^2/40$. These computations confirm the occurrence of a phase transition to a periodic ground state somewhere between $\mu=80$ and $\mu=90$, with oscillations arising from the boundary. Of course, there is no well defined sharp transition in finite volume. Evaluating the right-hand side of~\eqref{eq:estim_mu_c} gives that the constant solution $u_{\rm cnst}$ becomes linearly unstable at $\mu\approx 86.5$, which might thus be the true value of $\mu_c$. This hints for a \emph{second-order phase transition}, with small oscillations developing on top of the constant function~\cite{KirNep-71}. Note that the lower bound in~\eqref{eq:mc-mu2} gives the rather loose estimate $\mu_c\geq 4.6$ for this potential $w$. Observe also that the period of the ground state $u$ does not vary very much with $\mu$, as predicted from Theorem~\ref{thm:high_density} below. There are respectively 25 and 26 peaks for $\mu=90$ and $\mu=150$, hence a period of about $1.6$. For comparison, the frequency at which the linear instability of $u_{\rm cnst}$ happens (the solution of the minimum in~\eqref{eq:estim_mu_c}) is $k_0\approx3.9$ which remarkably gives the same period $2\pi/k_0\approx 1.6$. Finally, for $\mu=150$ the effective density $\rho(\mu)\approx 77$ of the solid is higher than the density $\mu/\int_\R w\approx 71$ of the fluid, as predicted in~\eqref{eq:critical_mu}.
\label{fig:numerics}}
\end{figure}

Next we state two simpler lower bounds on $\mu_c$ under additional conditions on $w$, and compare our result with~\cite{AftBlaJer-09}.

\begin{proposition}[Improved lower bound on $\mu_c$]\label{prop:improved_mu_c}
Let $d\geq1$ and $w$ satisfy Assumption~\ref{ass:w}. Let $\mu_c$ be as in Definition \ref{def:rho_c}.

\noindent If $s>d+2$ and $w$ is radial, then we have
\begin{equation}\label{eq:mc-mu1-off-factor}
\mu_c \ge  \frac{\sqrt{5}-1}{2} \frac{d\int_{\R^d} w}{\int_{\R^d} |x|^2 |w|}>0.
\end{equation}

\smallskip

\noindent If $w\ge 0$, then
\begin{equation}\label{eq:mc-mu2}
\mu_c \ge \inf_{k\in \R^d} \frac{|k|^2\hat w(0)}{2(\hat w(0) - 2 \hat w(k))_+} >0.
\end{equation}
\end{proposition}

The proof of Proposition \ref{prop:improved_mu_c} is essentially a refinement of the proof of Theorem \ref{thm:phase_transitions} (ii). It is provided in Section \ref{sec:improved_mu_c}.

Note that if $\varphi$ is a real-valued radial function such that $|x|^2\varphi(x)\in L^1(\R^d)$, then
\begin{align} \label{eq:Taylor-expansion-fourier}
(2\pi)^{d/2}|\hat \varphi(0)-\hat \varphi(k)| &= \left| \int_{\R^d} (1-\cos(k\cdot x)) \varphi(x) \dx \right| \nn\\
&\le  \int_{\R^d} \frac{|k\cdot x|^2}{2} |\varphi(x)| \dx = \frac{|k|^2}{2d} \int_{\R^d} |x|^2 |\varphi(x)| \dx.
\end{align}
Consequently, if $s>d+2$, then
\begin{align} \label{eq:def-mu1}
\inf_{k\in \R^d} \frac{|k|^2\hat w(0)}{2(\hat w(0) - \hat w(k))_+} \ge \frac{d\int_{\R^d} w}{\int_{\R^d} |x|^2 |w|} =: \mu_1,
\end{align}
with equality when $w\ge 0$. In~\cite[Prop.~1.2]{AftBlaJer-09} it was shown that $\mu_c\ge \mu_1$ when $w\ge 0$ and $s>d+2$ (the result was written in a different setting but the proof applies to our case as well). Under these two assumptions on $w$ we can get the additional factor 2 in $\widehat{w}(0)-2\widehat{w}(k)$ so that~\eqref{eq:mc-mu2} is better than~\cite{AftBlaJer-09}. If we do not assume anything concerning the sign of $w$ but still require $s>d+2$, then we only lose the factor $(\sqrt5-1)/2$ in~\eqref{eq:mc-mu1-off-factor} compared with~\cite{AftBlaJer-09}.

\subsection{Low and high density limits}\label{sec:high-density}

Finally we give more information about the limits $\mu\to 0$ and $\mu\to \infty$.

\subsubsection{Low density limit} We have proved in Theorem \ref{thm:phase_transitions} that the constant function \eqref{eq:cnst} is an infinite ground state for all $\mu\le \mu_c$. In the limit $\mu\to 0$ we can prove uniqueness among real-valued functions.

\begin{theorem}[Low-density limit]\label{thm:phase_transitions-low-density} Let $w$ satisfy Assumption~\ref{ass:w}. Then  the constant function $u_{\rm cnst}=(\mu/\int_{\R^d}w)^{1/2}$ is the \textbf{unique real-valued} infinite ground state, up to a sign, for $\mu$ small enough.
\end{theorem}

The proof of Theorem \ref{thm:phase_transitions-low-density} is given in Section \ref{sec:proof_low_density}. It is based on the fact that  we obtain infinite ground states for the Ginzburg-Landau equation in the limit $\mu\to0$, after scaling. More precisely, if $u_\mu$ is a real-valued infinite ground state for some $\mu$, then the function
$$\widetilde u_\mu(x):=\left(\frac{\int_{\R^d} w}{\mu}\right)^{\frac12} u_\mu\left(\frac{x}{\sqrt\mu}\right)$$
is an infinite ground state for the rescaled potential
$$\widetilde w_\mu(x)=\frac{\mu^{-\frac{d}2}w(x/\sqrt\mu)}{\int_{\R^d}w}\underset{\mu\to0}\longrightarrow \delta_0$$
and solves the rescaled GP equation
$$-\Delta \tilde u_\mu+\big(\tilde w_\mu\ast\tilde u_\mu^2-1\big)\tilde u_\mu=0$$
For this reason, we first consider the special case of the delta interaction, where Theorem~\ref{thm:phase_transitions-low-density} is well known in dimensions $d\in\{1,2\}$ but does not seem to have been written in all space dimensions. Then the general case in Theorem \ref{thm:phase_transitions-low-density} is obtained by a quantitative comparison with the Ginzburg-Landau equation. Note that our proof is independent of Theorem \ref{thm:phase_transitions} and it gives an alternative (although non-quantitative) justification for the fact that $\mu_c>0$.

In dimension $d=1$ we know from Theorem~\ref{thm:real-valued} that all the infinite ground states are real-valued, hence Theorem \ref{thm:phase_transitions-low-density}  implies uniqueness at low density (up to a constant phase factor).

In dimension $d=2$, we expect that the vortex solution of degree one (see Theorem \ref{thm:exist_vortex}) is the only possible solution at low densities apart from the constant  function~\eqref{eq:cnst} (up to translations and complex conjugation). This fact is well-known for the Ginzburg-Landau equation, see \cite{YeZho-96,Mironescu-96,Sandier-98,ComMir-99,PacRiv-00,MilPis-10,FarMir-13,IgnNguSlaZar-20} for various uniqueness results, but we are not able to prove it under the generality of $w$ in Assumption~\ref{ass:w}.

\subsubsection{High density limit} In the high density limit $\mu\to+\ii$, we obtain the classical mean-field problem, that is, with the kinetic energy removed. The latter was studied for instance in~\cite{Suto-05,Suto-11}.

For simplicity in the whole subsection we assume that $w$ is a continuous bounded function satisfying Assumption~\ref{ass:w}. Then, in a bounded domain~$\Omega$, we define the lowest classical canonical and grand-canonical energies by
\begin{align}
E_{\rm cl}(\lambda,\Omega)&:=\min_{\substack{\nu\geq0\\ \nu(\Omega)=\lambda}}\frac12\iint_{\Omega^2}w(x-y)\,\rd\nu(x)\,\rd\nu(y),
 \label{eq:E_classical}\\
F_{\rm cl}(\mu,\Omega)&:=\min_{\nu\geq0}\left\{\frac12\iint_{\Omega^2}w(x-y)\,\rd\nu(x)\,\rd\nu(y)-\mu\,\nu(\Omega)\right\}.
 \label{eq:F_classical}
\end{align}
In other words we remove the kinetic energy in $E_{\rm D}(\lambda,\Omega)$ and $F_{\rm D}(\mu,\Omega)$ and replace $|u|^2$ by a non-negative finite Borel measure $\nu$. Multiplying $\nu$ by an appropriate constant, we see that
$$E_{\rm cl}(\lambda,\Omega)=\left(\frac\lambda{|\Omega|}\right)^2E_{\rm cl}(|\Omega|,\Omega),\qquad F_{\rm cl}(\mu,\Omega)=\mu^2F_{\rm cl}(1,\Omega).$$
Hence we can always work at unit density or unit chemical potential. In the thermodynamic limit, we find
\begin{equation}
\lim_{\substack{\Omega_n\nearrow\R^d\\ \frac{\lambda_n}{|\Omega_n|}\to\rho}}\frac{E_{\rm cl}(\lambda_n,\Omega_n)}{|\Omega_n|}=\rho^2e_{\rm cl},\qquad \lim_{\Omega_n\nearrow\R^d}\frac{F_{\rm cl}(\mu,\Omega_n)}{|\Omega_n|}=-\frac{\mu^2}{4e_{\rm cl}},
\label{eq:thermo_limit_class}
\end{equation}
for some constant $0<e_{\rm cl}\leq\int_{\R^d}w/2$. Here $\Omega_n$ is any sequence of domains satisfying the Fischer regularity conditions~\eqref{eq:Fisher-0} and \eqref{eq:Fisher}. The canonical and grand-canonical classical energies are thus exactly quadratic.

An appropriate modification of Theorem~\ref{thm:local_bound} below (with no kinetic energy) proves that there exists a constant $C$ depending only on $w$ and $d$ such that any minimizer $\nu_\Omega$ of $F_{\rm cl}(1,\Omega)$ satisfies $\nu_\Omega(Q)\leq C|Q|$ for any large enough cube $Q$.  This allows one to pass to the thermodynamic limit and obtain a classical infinite ground state.

\begin{definition}[Classical infinite ground state]
A \emph{classical infinite ground state} at chemical potential $\mu=1$ is a uniformly locally finite measure $\nu$ on $\R^d$, which minimizes the local free energy
\begin{multline*}
\nu'\mapsto \frac12\iint_{(B_R)^2}w(x-y)\rd\nu'(x)\,\rd\nu'(y)\\+\iint_{B_R\times(\R^d\setminus B_R)}w(x-y)\rd\nu'(x)\,\rd\nu(y) -\nu'(B_R)
\end{multline*}
for any $R>0$.
\end{definition}

Many of the results in this paper apply to the classical setting. For instance, by an equivalent of Theorem~\ref{thm:prop_GS} any classical infinite ground state must have the average density and (twice the) energy
$$\lim_{R\to\ii}\frac{\nu(B_R)}{|B_R|}=\lim_{R\to\ii}\frac{1}{|B_R|}\iint_{(B_R)^2}w(x-y)\,\rd\nu(x)\,\rd\nu(y)=\frac1{2e_{\rm cl}}.$$
In addition, any classical infinite ground state solves the implicit equation
\begin{equation}
 \nu\ast w\begin{cases}
\geq1&\text{on $\R^d$,}\\
=1&\text{$\nu$-almost surely.}
          \end{cases}
\end{equation}
Finally, we have $e_{\rm cl}=\int_{\R^d}w/2$ if and only if $\widehat{w}\geq0$.

\begin{theorem}[High-density limit]\label{thm:high_density}
Assume that $d\geq1$ and that $w$ is a bounded continuous function satisfying Assumption~\ref{ass:w}. Then, we have
\begin{equation}
 e(\rho) = \rho^2e_{\rm cl} + o(\rho^2)_{\rho\to\ii},\qquad f(\mu)=-\frac{\mu^2}{4e_{\rm cl}} + o(\mu^2)_{\mu\to\ii}
 \label{eq:high_density_energy}
\end{equation}

For any $\mu_n\to\ii$ and any associated positive infinite ground state $u_n>0$, we have after extraction of a subsequence
$$\frac{|u_n|^2}{\mu_n}\wto\nu$$
locally in the sense of measures, where $\nu$ is a classical infinite ground state.
\end{theorem}

The proof is provided in Section~\ref{sec:proof_high_density}. The statement means that, at high density, infinite ground states have a fixed profile multiplied by the large constant $\mu$. A similar effect occurs in systems of interacting classical particles, which tend to pile up into big clusters when the density is increased~\cite{VenNij-79-2,VenNij-79-1,NijRui-85b}. This phenomenon can also be understood using the mean-field classical model~\cite{Klein_etal-94,LikWatLow-98,LikLanWatLow-01,Likos-07,TorSti-08,BatStiTor-09,Suto-11b}.

The rest of the article is devoted to the proof of all the results collected in this section.

\section{Local and global bounds}\label{sec:proof_bounds}

In this section, we prove Theorem~\ref{thm:uniform_bound} and Corollary~\ref{cor:existence-1}. We need to divide the proofs into many intermediate steps. First we prove universal local bounds for $H^1_{\rm unif}(\Omega)$ solutions in any domain $\Omega$ (possibly the whole space), with either Dirichlet or Neumann boundary conditions. This is the content of Theorem~\ref{thm:local_bound} below. We also consider in Section~\ref{sec:proof_inhomogeneous} inhomogeneous conditions which will be useful for later purposes. The local average bounds can easily be upgraded into $L^\ii$ bounds using standard elliptic techniques (Corollary~\ref{cor:pointwise_bounds_I}), but the result is not as good as expected for $\mu\ll1$ in dimensions $d\geq4$. These dimensions require a different treatment and, in particular, a proof in Section~\ref{sec:proof_uniform_bound_low_density} that the only solution of the GP equation for $\mu=0$ is the trivial solution $u\equiv0$. In Section~\ref{sec:proof_derivatives} we use the $L^\ii$ bound on $u$ to get similar estimates on all its derivatives and in Section~\ref{sec:proof_Harnack} we use a Harnack-type argument to show pointwise lower bounds for positive solutions.
Finally, we use the local bounds in Section~\ref{sec:proof_existence} to deduce the existence of infinite ground states by means of a thermodynamic limit, thus proving Corollary~\ref{cor:existence-1}.

\subsection{Local bounds}

This section contains the first main estimates of the article, which will allow us to derive all the other results. It says that any solution to the GP equation in $H^1_{\rm unif}(\R^d)$ satisfies universal local bounds depending only on $w$, $\mu$ and the space dimension $d$.

\begin{theorem}[Local bounds]\label{thm:local_bound}
Let $d\geq 1$ and $w$ be a potential satisfying Assumption~\ref{ass:w}. Let $\Omega\subset\R^d$ be a (bounded or unbounded) smooth domain. Let $\mu>0$ and $u\in H^1_{\rm unif}(\Omega)$ be an arbitrary solution of the GP equation
\begin{equation}
-\Delta u+w\ast|u|^2u=\mu\, u
 \label{eq:GP_local2}
\end{equation}
in the sense of distributions, satisfying the Dirichlet $u_{|\partial\Omega}\equiv0$ or Neumann $\partial_\nu u_{|\partial\Omega}=0$ boundary condition if $\Omega\neq\R^d$, and having a finite local interaction in the sense of~\eqref{eq:finite_interation_u}. Then there exists a constant $C>0$, depending only on $d$ and $w$, such that
\begin{equation}
\int_{B(z,\ell)}|u|^2\leq C\left(\mu+\frac1{\ell^2}\right)\ell^d
\label{eq:local_bound_mass}
\end{equation}
and
\begin{equation}
 \int_{B(z,\ell)}|\nabla u|^2+\iint_{B(z,\ell)^2}|u(x)|^2|u(y)|^2w_+(x-y)\,\dx\,\dy\leq C\left(\mu+\frac1{\ell^2}\right)^2\ell^d
 \label{eq:local_bound_kinetic_interaction}
\end{equation}
for any $z\in\R^d$ and any $\ell\geq C$.
\end{theorem}

As announced, our proof is only based on local energy comparison techniques and it is inspired by Ruelle~\cite{Ruelle-70} (see also~\cite[App.~A]{JexLewMad-24}). The properties of $w$ are crucial here, and in particular the fact that $w$ is positive at the origin.

\begin{proof}[Proof of Theorem~\ref{thm:local_bound}] We split the proof into several steps.

\subsubsection*{Step 1. Estimates on the potential.}
We start with the proof that $w\ast|u|^2\in L^1_{\rm unif}(\R^d)$ when $u\in H^1_{\rm unif}(\R^d)$.

\begin{lemma}\label{lem:estim_potential}
For any $z\in\R^d$, any $R\geq \kappa$ and any $u\in L^2_{\rm unif}(\R^d)$, we have
\begin{equation}
w_+\ast|u|^2\leq w_+\ast(|u|^2\1_{B(z,2R)})+\frac{C}{R^{s}}\sup_{y\in \R^d}\int_{B(y,R)}|u|^2
\label{eq:estim_w_u_L2unif}
\end{equation}
on the ball $B(z,R)$, where $C$ only depends on $w$. Similarly, we have
\begin{equation}
w_-\ast|u|^2\leq \kappa (1+|\cdot|^s)^{-1}\ast|u|^2\leq C\left(\int_{B(z,2R)}|u|^2+\frac1{R^s}\sup_{y\in \R^d}\int_{B(y,R)}|u|^2\right)
 \label{eq:estim_w2_Linfty}
\end{equation}
on $B(z,R)$.
\end{lemma}

The lemma implies that
$$w\ast|u|^2=\underbrace{w_+\ast(|u|^2\1_{B(z,2R)})}_{=:w_{z,R}}+\widetilde{w}_{z,R},$$
on the ball $B(z,R)$, where $\|\widetilde{w}_{z,R}\|_{L^\infty(\R^d)}$ is bounded uniformly in $z$. The first potential $w_{z,R}$ is in $L^1(\R^d)$ and therefore $w\ast|u|^2$ belongs to $L^1_{\rm unif}(\R^d)$.
When $u$ has a finite interaction energy as in~\eqref{eq:finite_interation_u}, we have $w_{z,R} (|u|^2\1_{B(z,R)} ) \in L^1(\R^d)$. Thus we deduce that
$$u(w\ast|u|^2)\in L^1_{\rm loc}(\R^d).$$
In low dimensions, the situation is easier and we can use $u\in H^1_{\rm unif}(\R^d)$ instead of~\eqref{eq:finite_interation_u}. For instance, if $d\leq 4$ we have $u\in L^4_{\rm unif}(\R^d)$ by Sobolev and therefore $w\ast|u|^2\in L^2_{\rm unif}(\R^d)$. In dimension $d\leq 6$ we can use that $u\in L^{p}_{\rm unif}(\R^d)$, with $p=\ii$ in $d=1$, $p<\ii$ in $d=2$ and $p=2d/(d-2)$. This implies that $w\ast|u|^2\in L^{p/2}_{\rm unif}(\R^d)$ and hence $u(w\ast|u|^2)\in L^{p/3}$ by the Cauchy-Schwarz inequality, where $p\geq 3$ only in dimensions $d\leq6$.

\begin{proof}[Proof of Lemma~\ref{lem:estim_potential}]
We use that $w(x)=w_2(x)\leq \kappa|x|^{-s}$ for $|x|\geq R\geq \kappa\geq 2r$ by Assumption~\ref{ass:w}.
We write
\begin{equation}
w_+\ast |u|^2=w_+\ast (|u|^2\1_{B(z,2R)})+w_+\ast (|u|^2\1_{B(z,2R)^c}).
 \label{eq:decomp_w_+}
\end{equation}
For the second term, we use that there exists a constant $A$ (depending on $\kappa$) such that
$$\frac{\1(|y|\geq \kappa)}{|y|^s}\leq A \left(\frac{1(|\cdot|\geq \kappa)}{|\cdot|^s}\ast \1_{B_{\kappa/4}}\right)(y)$$
since the term on the right behaves like $|x|^{-s}$ at infinity. By scaling this gives
$$\frac{\1(|y|\geq R)}{|y|^s}\leq A\left(\frac\kappa R\right)^{d} \left(\frac{1(|\cdot|\geq R)}{|\cdot|^s}\ast \1_{B_{R/4}}\right)(y)$$
and hence
$$\int_{|x-y|\geq R} \frac{|u(y)|^2}{|x-y|^s}\,\dy\leq A\left(\frac\kappa R\right)^{d}\left(\sup_{y\in \R^d}\int_{B(y,R/4)}|u|^2\right)\int_{|y|\geq R}\frac{\dy}{|y|^s}$$
as we wanted. The proof for the bounded potentials $w_-$ and $(1+|x|^s)^{-1}$ is similar, using that the first term on the right side of~\eqref{eq:decomp_w_+} is now bounded.
\end{proof}

\subsubsection*{Step 2. A localization method}

Let us integrate the GP equation~\eqref{eq:GP_local} against $\chi^2\overline{u}$ for some $\chi\in W^{1,\ii}_c(\R^d,\R)$. We obtain
\begin{equation}
-\int_{\Omega}\chi^2\overline{u}\Delta u+D(\chi^2|u|^2,|u|^2)=\mu\int_{\Omega}\chi^2|u|^2,
 \label{eq:GP_localized_chi_pre}
\end{equation}
with the notation
\begin{equation}
\boxed{D(\rho_1,\rho_2)=\iint_{\R^d\times\R^d}\rho_1(x)\rho_2(y)w(x-y)\,\dx\,\dy.}
 \label{eq:def_D_notation}
\end{equation}
Since the second and third terms in~\eqref{eq:GP_localized_chi_pre} are real, we can freely take the real part of the first term. For any $u\in H^1(\Omega)$ satisfying either the Dirichlet or Neumann boundary condition, we have the IMS formula
\begin{equation}
-\Re\int_{\Omega}\chi^2\overline{u}\Delta u=\int_{\Omega}|\nabla (\chi u)|^2-\int_{\Omega}|\nabla \chi|^2|u|^2.
 \label{eq:IMS}
\end{equation}
Therefore we arrive at
\begin{equation}
\int_{\Omega}|\nabla (\chi u)|^2+D(\chi^2|u|^2,|u|^2)=\mu\int_{\Omega}\chi^2|u|^2+\int_{\Omega}|\nabla \chi|^2|u|^2,
\label{eq:GP_localized_chi}
\end{equation}
for all $\chi\in W^{1,\ii}_c(\R^d,\R)$. Let $z\in\R^d$ and $\ell_0\geq r$. Our goal is to estimate the mass $\int_{B(z,\ell_0)}|u|^2$ uniformly in $z$.

\medskip

\paragraph{\it Positive interaction} When $w_2\geq0$, the argument is rather easy. In this case, using~\eqref{eq:decomp_w} we find from~\eqref{eq:GP_localized_chi}
$$\eps\int\left(\chi^2|u|^2\ast\delta_r\right)^2\leq \mu\int_{\Omega}\chi^2|u|^2+\int_{\Omega}|\nabla \chi|^2|u|^2.$$
We choose $\chi$ so that $\chi\equiv1$ on $B(z,\ell_0)$ and $\chi\equiv0$ outside of $B(z,2\ell_0)$, with $|\nabla \chi|\leq 1/\ell_0$, and find
\begin{align}
\int_{B(z,\ell_0)}|u|^2\leq\int_{\R^d}\chi^2|u|^2&=\int_{\R^d}(\chi^2|u|^2)\ast\delta_r\nn\\
&\leq |B_1|^{\frac12}(3\ell_0)^{\frac{d}2}\left(\int_{\R^d}\left(\chi^2|u|^2\ast\delta_r\right)^2\right)^{\frac12}\nn\\
&\leq \frac{|B_1|^{\frac12}(3\ell_0)^{\frac{d}2}}{\sqrt\eps}\left(\mu+\ell_0^{-2}\right)^{\frac12}\left(\int_{B(z,2\ell_0)}|u|^2\right)^{\frac12}.
\label{eq:local_bd_w_positive_proof}
\end{align}
In the second line we have used the Cauchy-Schwarz inequality and the fact that the convolution $\chi^2|u|^2\ast\delta_r$ has its support in the ball $B(z,2\ell_0+r)\subset B(z,3\ell_0)$.
Next we call $M:=\sup_{z'\in\R^d}\int_{B(z',\ell_0)}|u|^2$ and $J$ the minimal number of balls of radius 1 that can be used to cover a ball of radius 2 ($J$ depends only on the dimension). By scaling we can cover $B(z,2\ell_0)$ by $J$ balls of radius $\ell_0$ and~\eqref{eq:local_bd_w_positive_proof} provides
$$M\leq \frac{|B_1|^{\frac12}(3\ell_0)^{\frac{d}2}}{\sqrt\eps}\left(\mu+\ell_0^{-2}\right)^{\frac12}\sqrt{JM}$$
or equivalently
\begin{equation}
\sup_{z\in\R^d}\int_{B(z,\ell_0)}|u|^2=M\leq \frac{J|B_1|3^d}{\eps}\ell_0^d\left(\mu+\ell_0^{-2}\right).
 \label{eq:local_bd_w_positive}
\end{equation}
Inserting back in~\eqref{eq:GP_localized_chi} we get the claimed inequality~\eqref{eq:local_bound_kinetic_interaction} on the kinetic and interaction energies.

\medskip

\paragraph{\it General interaction}
The previous simple proof cannot be easily extended to a general interaction $w$. More information on the way that the mass is distributed outside of the ball $B(z,\ell_0)$ is needed. We follow a localization method due to Ruelle in~\cite{Ruelle-70}. We fix an $\alpha\in(0,1]$ and define the exponentially growing sequence
$$\ell_j=\ell_0(1+\alpha)^j.$$
We will need to take $\alpha$ small enough and $\ell_0$ large enough. The explicit conditions will be given later. We denote by $\cB_j:=B(z,\ell_j)$ the ball of radius $\ell_j$ and by $V_j:=|\cB_j|=\ell_j^d|B_1|$ its volume. It will be important that we estimate the local \emph{smeared squared mass} defined as
\begin{equation}
M_j:=\int_{\R^d}\rho_j(x)^2\,\dx,\qquad \rho_j:=\big(|u|^2\1_{\cB_j}\big)\ast\delta_{r},
 \label{eq:def_M_j_rho_j}
\end{equation}
for all $j\geq0$, with $\delta_r$ the function in~\eqref{eq:def_delta_r}. Except in the case $r=0$ (contact interaction), it is unclear if $u\in L^4_{\rm loc}(\R^d)$ in high dimensions, hence the need of a smearing. Note that $M_j$ depends on the chosen center $z\in\R^d$ and that $\rho_j$ has its support in the ball $B(z,\ell_j+r)$.

Next, we apply~\eqref{eq:GP_localized_chi} to a smeared version of $\1_{\cB_j}$. We choose
$$\chi_j(x)=\begin{cases}
1&\text{for $|x-z|\leq \ell_j$,}\\
\frac{\ell_{j+1}-|x-z|}{\ell_{j+1}-\ell_j}&\text{for $\ell_j\leq |x-z|\leq \ell_{j+1}$,}\\
0&\text{for $|x-z|\geq \ell_{j+1}$,}
\end{cases}$$
and obtain
\begin{equation}
\int_{\Omega}|\nabla (\chi_j u)|^2+D(\chi_j^2|u|^2,|u|^2)\leq \left(\mu+\frac1{\alpha^2\ell_0^2}\right)\int_{\cB_{j+1}}|u|^2
\label{eq:estim_mu_begin}
\end{equation}
since $\ell_{j+1}-\ell_j=\alpha\ell_j\geq\alpha\ell_0$.
If $r>0$, we can use our assumptions on $w$ and bound
\begin{align}
D(\chi_j^2|u|^2,|u|^2)&=D(\chi_j^2|u|^2,\chi_j^2|u|^2)+D(\chi_j^2|u|^2,(1-\chi_j^2)|u|^2)\nn\\
&\geq \eps\int_{\R^d}\left(\chi_j^2|u|^2\ast\delta_r\right)^2-D_2(\chi_j^2|u|^2,(1-\chi_j^2)|u|^2)\nn\\
&\geq \eps M_j-D_2(\1_{\cB_{j+1}}|u|^2,\1_{\cB^c_j}|u|^2),
\label{eq:estim_interaction_below}
\end{align}
where
$$D_2(\rho_1,\rho_2):=\kappa\iint_{\R^d\times\R^d}\frac{\rho_1(x)\rho_2(y)}{1+|x-y|^s}\,\dx\,\dy.$$
If $r=0$, we write instead
\begin{align*}
D(\chi_j^2|u|^2,|u|^2)&=\eps \int\chi_j^2|u|^4+\iint\chi_j(x)^2|u(x)|^2|u(y)|^2w_2(x-y)\,\dx\,\dy\\
&\geq \eps \int\chi_j^4|u|^4-D_2(\chi_j^2|u|^2,(1-\chi_j^2)|u|^2)\\
&\geq \eps M_j-D_2(\1_{\cB_{j+1}}|u|^2,\1_{\cB^c_j}|u|^2),
\end{align*}
which is the same as~\eqref{eq:estim_interaction_below}. Many of the following estimates are easier and sharper for $r=0$, but we do not distinguish this case for shortness. We end up with
\begin{equation}
\int_{\cB_j}|\nabla u|^2+\eps M_j\leq \left(\mu+\frac1{\alpha^2\ell_0^2}\right)\int_{\cB_{j+1}}|u|^2+D_2(\1_{\cB_{j+1}}|u|^2,\1_{\cB^c_j}|u|^2)
 \label{eq:estim_M_j}
\end{equation}
for all $j\geq0$.

Since $u\in L^2_{\rm unif}(\R^d)$ and $(1+|\cdot|^s)^{-1}\ast|u|^2\in L^\ii(\R^d)$ by Lemma~\ref{lem:estim_potential}, the right side of~\eqref{eq:estim_M_j} can be estimated by a constant times $V_{j+1}=(1+\alpha)V_j\leq 2V_j$. Therefore we know that $M_j\leq CV_j$ for some constant $C$, which however depends on the unknown value of $\|u\|_{L^2_{\rm unif}}$. Our goal is to prove a universal upper bound depending on $w$ and $\mu$ but not on $u$. The problem is that the inequality~\eqref{eq:estim_M_j} furnishes a control on $M_j$ which depends on $u$ outside of $\cB_j$, that is, on all the other $M_k$'s. We will therefore be forced to control all the $M_j$'s. The idea is to first allow for a small deviation from a volume term. We thus define
$$\psi_j:=\psi_0(1+\alpha)^{\eta j}$$
for some small $0<\eta\leq1$ such that $s>d+\eta$, for instance
$$\eta:=\min\left(1,\frac{s-d}2\right),$$
and some $\psi_0$ (to be defined later). We then prove that necessarily $M_j\leq V_j\psi_j$ for all $j\geq0$, when $\psi_0,\ell_0$ are chosen large enough and $\alpha$ is small enough. Once this sub-optimal inequality is shown, we infer in particular $M_0\leq\psi_0V_0$ uniformly in $z$ and can then infer by tiling that $M_j$ is in fact controlled by the volume $V_j$.

The proof goes by contradiction, assuming that $M_j>V_j\psi_j$ for at least one $j\geq0$. Let $q$ be the largest integer such that
\begin{equation}
 M_q> V_q\psi_q.
 \label{eq:prop_q}
\end{equation}
We have seen that $M_j\leq CV_j$ for some $C$ depending on $u$ and we have $\psi_j\to\ii$. This implies that $q$ is always finite. 
We then have by definition 
\begin{equation}
M_j\leq V_j\psi_j,\qquad \forall j>q.
\label{eq:def_q}
\end{equation}
In other words, the mass is too large for $j=q$ but good for $j>q$. The idea is then to use~\eqref{eq:estim_M_j} for $j=q$ and control the outside using~\eqref{eq:def_q} to arrive at a contradiction. The following is our main technical estimate.

\begin{lemma}[Error estimate]\label{lem:error_estim}
There exists an $0<\alpha<1$ (depending only on $d$ and $w$) such that for all $\ell_0$ large enough (depending on $d,w,\alpha$), we have for the last term in~\eqref{eq:estim_interaction_below}
$$D_2(\1_{\cB_{q+1}}|u|^2,\1_{\cB^c_q}|u|^2)\leq \frac{\eps}{2}M_q$$
whenever $q$ defined by \eqref{eq:def_q} exists, that is, $M_j>V_j\psi_j$ for one $j\geq0$.
\end{lemma}

It is important that the conditions on $\alpha$ and $\ell_0$ do not depend on $\psi_0$ (but of course the value of $q$ does when it exists). It is because the error term $D_2(\1_{\cB_{q+1}}|u|^2,\1_{\cB^c_q}|u|^2)$ is quartic in $u$ that we have to estimate a (smeared) squared mass and not the local mass itself.

Before providing the proof of Lemma~\ref{lem:error_estim}, we explain how to use it in order to get a contradiction when $\psi_0$ is large enough. Inserting the estimate of Lemma \eqref{lem:error_estim}  into ~\eqref{eq:estim_mu_begin} for $j=q$, we obtain
\begin{equation}
\frac{\eps}{2}M_q\leq\left(\mu+\frac1{\alpha^2\ell_0^2}\right)\int_{\cB_{q+1}}|u|^2.
\label{eq:estim_u2}
\end{equation}
We write using the Cauchy-Schwarz inequality
$$\int_{\cB_{q+1}}|u|^2=\int\rho_{q+1}=\int\rho_{q+1}\ast\delta_r\leq \sqrt{M_{q+1}}|\cB_{q+1}+B_r|^{\frac12}\leq 3^{\frac{d}2}\sqrt{M_{q+1}V_q}$$
since $\ell_{q+1}=(1+\alpha)\ell_q\leq 2\ell_q$ and $r\leq\ell_0\leq \ell_q$, hence $\ell_{q+1}+r\leq 3\ell_q$. Here recall \eqref{eq:def_M_j_rho_j} for the definition of $\rho_{q+1}$. Now we have by definition of $q$
$$M_{q+1}\leq V_{q+1}\psi_{q+1}=(1+\alpha)^{d+\eta}V_q\psi_q\leq (1+\alpha)^{d+\eta}M_q\leq 2^{d+1}M_q$$
hence we arrive at
$$\int_{\cB_{q+1}}|u|^2\leq 3^{\frac{d}2}2^{\frac{d+1}2}\frac{M_q}{\sqrt{\psi_q}}\leq 3^{\frac{d}2}2^{\frac{d+1}2}\frac{M_q}{\sqrt{\psi_0}}.$$
Inserting in~\eqref{eq:estim_u2}, this gives
$$\psi_0\leq\frac{3^d2^{d+3}}{\eps^2}\left(\mu+\frac1{\alpha^2\ell_0^2}\right)^2.$$
If we choose $\psi_0$ to be twice the right side we arrive at a contradiction, namely we conclude as we wanted that $M_j\leq V_j\psi_j$ for all $j\geq0$.

For $j=0$, this gives us by the same argument as above
\begin{multline}
\int_{\cB_0}|u|^2=\int_{\cB_0+B_r}\rho_0
\leq \sqrt{M_0}|\cB_0+B_r|^{\frac12}\leq 2^{\frac{d}2}\sqrt{\psi_0}|\cB_0|\\
\leq\frac{3^{\frac{d}2}2^{d+2}|B_1|}{\eps}\left(\mu+\frac1{\alpha^2\ell_0^2}\right)\ell_0^d
\leq\frac{3^{\frac{d}2}2^{d+2}|B_1|}{\eps\alpha^2}\left(\mu+\frac1{\ell_0^2}\right)\ell_0^d.
 \label{eq:M0}
\end{multline}
This is the desired estimate on the mass in~\eqref{eq:local_bound_mass}. The constant only depends on $d,\eps$ and $\alpha^2$, but the conditions on $\ell_0$ depend on the other variables, as we will see in the proof of Lemma~\ref{lem:error_estim}.

\subsubsection*{Step 3. Proof of Lemma~\ref{lem:error_estim}}
We define
\begin{equation}
K:=\max_{x\in\R^d}\frac{(1+|x|^s)^{-1}}{(1+|\cdot|^s)^{-1}\ast\delta_r\ast\delta_r}.
\label{eq:def_K}
\end{equation}
We are using here that $(1+|\cdot|^s)^{-1}\ast\delta_r\ast\delta_r$ behaves like $|x|^{-s}$ at infinity, so that the function in the maximum is continuous, positive and vanishing at infinity. For $r=0$, we simply have $K=1$. We obtain
\begin{equation}
D_2(\nu_1,\nu_2)\leq KD_2(\nu_1\ast\delta_r,\nu_2\ast\delta_r),
 \label{eq:D2_smeared}
\end{equation}
for every $0\leq \nu_1,\nu_2\in L^1(\R^d)$.
Then, defining
\begin{equation}
S:=\kappa\,K\int_{\R^d}\frac{\dx}{1+|x|^s},
\label{eq:def_S}
\end{equation}
we obtain by Young's inequality
\begin{equation}
D_2(\nu_1,\nu_2)\leq S\left(\int_{\R^d}\big(\nu_1\ast\delta_r\big)^2\right)^{\frac12}\left(\int_{\R^d}\big(\nu_2\ast\delta_r\big)^2\right)^{\frac12},
\label{eq:D2_estim}
\end{equation}
for every $0\leq \nu_1,\nu_2\in L^1(\R^d)$. We split our error term as
\begin{multline}
D_2(\1_{\cB_{q+1}}|u|^2,\1_{\cB^c_q}|u|^2)=D_2(\1_{\cB_{q+1}}|u|^2,\1_{\cB_{q+2}\setminus\cB_q}|u|^2)\\
+D_2(\1_{\cB_{q+1}}|u|^2,\1_{\cB_{q+2}^c}|u|^2).
\label{eq:split_error}
\end{multline}
For the first term on the right of~\eqref{eq:split_error}, we use~\eqref{eq:D2_estim} and \eqref{eq:def_M_j_rho_j} and obtain
$$D_2(\1_{\cB_{q+1}}|u|^2,\1_{\cB_{q+2}\setminus\cB_q}|u|^2)\leq S\sqrt{M_{q+1}}\left(\int(\rho_{q+2}-\rho_q)^2\right)^{\frac12}$$
Using $(a-b)^2\leq a^2-b^2$ for $a\geq b \ge 0$, we obtain from the definition of $q$
\begin{multline*}
\int(\rho_{q+2}-\rho_q)^2\leq M_{q+2}-M_q\leq V_{q+2}\psi_{q+2}-M_q\\
=(1+\alpha)^{2d+2\eta}V_q\psi_q-M_q\leq\big((1+\alpha)^{2d+2\eta}-1\big)M_q\leq\big((1+\alpha)^{2d+2}-1\big)M_q.
\end{multline*}
Using finally that
$$M_{q+1}\leq V_{q+1}\psi_{q+1}= (1+\alpha)^{d+\eta}V_q\psi_q\leq 2^{d+1} M_q,$$
we obtain our final bound on the first term in~\eqref{eq:split_error}
$$D_2(\1_{\cB_{q+1}}|u|^2,\1_{\cB_{q+2}\setminus\cB_q}|u|^2)\leq S2^{\frac{d+1}2}\left((1+\alpha)^{2d+2}-1\right)^{\frac12}M_q.$$
At this point we can fix the value of $\alpha$. We want the above term to be controlled by the term $\eps M_q$ in~\eqref{eq:estim_interaction_below} and choose $\alpha\leq1$ so that
\begin{equation}
S2^{\frac{d+1}2}\left((1+\alpha)^{2d+2}-1\right)^{\frac12}\leq\frac{\eps}4.
 \label{eq:alpha_cons}
\end{equation}
In other words, we pick
\begin{equation}
\alpha:=\min\left(1\,,\, \left(1+\frac{\eps^2}{2^{d+5}S^2}\right)^{\frac1{2d+2}}-1\right)
\label{eq:formula_alpha}
\end{equation}
where we recall that $S$ is defined in~\eqref{eq:def_S} with $K$ as in~\eqref{eq:def_K}. This value of $\alpha$ appears in the final bound~\eqref{eq:M0}. We still have to explain how large $\ell_0$ has to be. As a conclusion, we have shown the bound
\begin{equation}
D_2(\1_{\cB_{q+1}}|u|^2,\1_{\cB_{q+2}\setminus\cB_q}|u|^2)\leq \frac{\eps}{4}M_q
\label{eq:final_split1}
\end{equation}
on the first term of~\eqref{eq:split_error}.

Next we consider the second term in~\eqref{eq:split_error}, which we estimate using \eqref{eq:D2_smeared}. This time we have
\begin{align*}
&D_2(\1_{\cB_{q+1}}|u|^2,\1_{\cB_{q+2}^c}|u|^2)\\
&\quad\leq \kappa K\iint\frac{\rho_{q+1}(x)(\1_{\cB^c_{q+2}}|u|^2\ast\delta_r)(y)}{|x-y|^s}\dx\,\dy\\
&\quad\leq \frac{\kappa K}2\int_{\cB_{q+1}+B_r}\int_{(\cB_{q+2}+B_r)^c}\frac{\rho_{q+1}(x)^2}{|x-y|^s}\dx\,\dy\\
&\quad \qquad+\frac{\kappa K}{2}\int_{\cB_{q+1}+B_r}\int_{(\cB_{q+2}+B_r)^c}\frac{(\1_{\cB^c_{q+2}}|u|^2\ast\delta_r)(y)^2}{|x-y|^s}\dx\,\dy.
\end{align*}
For the first term, we use that
$$|x-y|\geq \ell_{q+2}-\ell_{q+1}-2r=\alpha\ell_{q+1}-2r\geq \alpha\ell_0-2r$$
and obtain after scaling
$$\int_{\cB_{q+1}+B_r}\int_{(\cB_{q+2}+B_r)^c}\frac{\rho_{q+1}(x)^2}{|x-y|^s}\dx\,\dy\leq \frac{M_{q+1}L}{(\alpha\ell_0-2r)^{s-d}}\leq \frac{2^{d+1}M_{q}L}{(\alpha\ell_0-2r)^{s-d}}$$
with
$$L:=\int_{|x|\geq 1}\frac{\dx}{|x|^s}=\frac{|\bS^{d-1}|}{s-d}.$$
We can choose $\ell_0>2r/\alpha$ large enough so that
\begin{equation}
\frac{2^{d}\kappa KL}{(\alpha\ell_0-2r)^{s-d}}\leq \frac{\eps}{8}.
\label{eq:ell0_cons_1}
\end{equation}
Finally, for the last term we use that
\begin{align*}
(\1_{\cB^c_{q+2}}|u|^2\ast\delta_r)^2&=\left(\sum_{j\geq q+2}\1_{\cB_{j+1}\setminus\cB_j}|u|^2\ast\delta_r\right)^2\\
&\leq 2\sum_{j\geq q+2}\left(\1_{\cB_{j+1}\setminus\cB_j}|u|^2\ast\delta_r\right)^2\leq 2\sum_{j\geq q+2}\rho_{j+1}^2-\rho_j^2
\end{align*}
since in the square at most two functions have a non-intersecting support at every point of space. Recall that $\ell_0>2r/\alpha$, which implies $\ell_{j+1}-r>\ell_j+r$. We find
\begin{align*}
&\int_{\cB_{q+1}+B_r}\int_{(\cB_{q+2}+B_r)^c}\frac{(\1_{\cB^c_{q+2}}|u|^2\ast\delta_r)(y)^2}{|x-y|^s}\dx\,\dy\\
&\quad \leq 2|\cB_{q+1}+B_r|\sum_{j\geq q+2}\frac{M_{j+1}-M_j}{(\ell_j-\ell_{q+1}-2r)^s}\\
&\quad\leq 2|\cB_{q+1}+B_r|\sum_{j\geq q+3}M_j\left(\frac{1}{(\ell_{j-1}-\ell_{q+1}-2r)^s}-\frac{1}{(\ell_{j}-\ell_{q+1}-2r)^s}\right).
\end{align*}
Using $M_j\leq\psi_jV_j$ as well as $|\cB_{q+1}+B_r|\leq 2^{d}V_{q+1}\leq 2^{2d}V_q$, we find after factorizing by $\ell_{q+1}$
\begin{align}
&\int_{\cB_{q+1}+B_r}\int_{(\cB_{q+2}+B_r)^c}\frac{(\1_{\cB^c_{q+2}}|u|^2\ast\delta_r)(y)^2}{|x-y|^s}\dx\,\dy\nn\\
&\quad\leq \frac{2^{2d+1}V_q\psi_0}{\ell_{q+1}^{s-d-\eta}}\sum_{j\geq 2}(1+\alpha)^{(d+\eta)j}\Bigg(\frac{1}{\left((1+\alpha)^{j-1}-1-\frac{2r}{\ell_{q+1}}\right)^s}\nn\\
&\qquad \qquad\qquad\qquad\qquad\qquad\qquad\qquad\qquad\qquad    -\frac{1}{\left((1+\alpha)^{j}-1-\frac{2r}{\ell_{q+1}}\right)^s}\Bigg)\nn\\
&\quad\leq \frac{2^{2d+1}M_q}{\ell_{0}^{s-d-\eta}}\sum_{j\geq 2}(1+\alpha)^{(d+\eta)j}\Bigg(\frac{1}{\left((1+\alpha)^{j-1}-1-\frac{2r}{\ell_0}\right)^s}-\frac{1}{\left((1+\alpha)^{j}-1\right)^s}\Bigg).\label{eq:ell0_cons_2}
\end{align}
We choose $\ell_0$ large enough so that the right side of~\eqref{eq:ell0_cons_2} is $\leq \frac{\eps}{4\kappa K}M_q$. Recall that $\eta=(s-d)/2$, hence the previous series is convergent. As a conclusion, we have shown that the second term in~\eqref{eq:split_error} can be controlled by
\begin{equation}
D_2(\1_{\cB_{q+1}}|u|^2,\1_{\cB_{q+2}^c}|u|^2)\leq\frac{\eps}{4}M_q.
\label{eq:final_split2}
\end{equation}
Inserting all this in~\eqref{eq:estim_interaction_below} concludes the proof of Lemma~\ref{lem:error_estim}.\qed

\subsubsection*{Step 4. Estimates on kinetic and interaction energies}
Let us deal with the gradient term and the local interaction. Taking $j=0$ in~\eqref{eq:estim_interaction_below} and~\eqref{eq:estim_M_j} and decomposing $|u|^2=\1_{\cB_2}|u|^2+\1_{\cB_2^c}|u|^2$, we find
\begin{multline}
\int_{\cB_0}|\nabla u|^2+D(\chi_j^2|u|^2,\chi_j^2|u|^2)\leq \left(\mu+\frac1{\alpha^2\ell_0^2}\right)\int_{\cB_{1}}|u|^2+D_2(\1_{\cB_{1}}|u|^2,\1_{\cB_2}|u|^2)\\+D_2(\1_{\cB_{1}}|u|^2,\1_{\cB_2^c}|u|^2).
\label{eq:estim_kin_int}
\end{multline}
Instead of involving $M_0$, we write
\begin{multline*}
D(\chi_j^2|u|^2,\chi_j^2|u|^2)\geq \iint_{\cB_0^2}|u(x)|^2|u(y)|^2w_+(x-y)\,\dx\,\dy\\
-D_2(\1_{\cB_1}|u|^2,\1_{\cB_1}|u|^2)
\end{multline*}
where we have used that $w_-\leq (w_2)_-\leq \kappa(1+|x|^s)^{-1}$.
As before we have
$$\int_{\cB_{1}}|u|^2\leq 2^{\frac{3d+1}2} V_0\sqrt{\psi_0}$$
and
\begin{align*}
D_2(\1_{\cB_1}|u|^2,\1_{\cB_1}|u|^2)&\leq D_2(\1_{\cB_{1}}|u|^2,\1_{\cB_2}|u|^2)\\
&\leq S\sqrt{M_1M_2}\leq S\sqrt{V_1\psi_1V_2\psi_2}\leq S2^{\frac{3d}2}\psi_0V_0.
\end{align*}
For the last term in~\eqref{eq:estim_kin_int}, we use~\eqref{eq:estim_w2_Linfty} in Lemma~\ref{lem:estim_potential} for $R=\ell_1$. This provides
\begin{equation*}
D_2(\1_{\cB_{1}}|u|^2,\1_{\cB_2^c}|u|^2)\leq \frac{C}{\ell_0^{s}}\left(\int_{\cB_1}|u|^2\right)\sup_{z\in\R^d}\left(\int_{B(z,\ell_0)}|u|^2\right)\leq C\psi_0\ell_0^{2d-s}.
\end{equation*}
Hence we have shown that
$$\int_{\cB_0}|\nabla u|^2+\iint_{\cB_0^2}|u(x)|^2|u(y)|^2w_+(x-y)\,\dx\,\dy\leq C\left(\mu+\frac1{\alpha^2\ell_0^2}\right)^2\ell_0^d.$$
This concludes the proof of~\eqref{eq:local_bound_kinetic_interaction}, hence of Theorem~\ref{thm:local_bound}.
\end{proof}

Our first corollary is a lower bound on the local mass for an infinite ground state, which was stated in~\eqref{eq:lower_bound} of Theorem~\ref{thm:uniform_bound}.

\begin{corollary}[Lower bound on the local mass for ground states]\label{cor:lower_bd_mass}
Let $d\geq1$ and $w$ satisfy Assumption~\ref{ass:w}. Let $\mu>0$. We assume that
\begin{itemize}
 \item either $\Omega=\R^d$ and $u$ is an infinite ground state,
 \item or $\Omega$ is a regular bounded domain and $u$ is a minimizer of the grand-canonical GP energy $\cF_{\mu,\Omega}$ in~\eqref{eq:cF_GP}, with either Dirichlet or Neumann boundary condition.
\end{itemize}
Then there exists a constant $C$ (depending only on $d$ and $w$) such that
$$\int_{B(z,\ell)}|u|^2\geq C^{-1}\left(\mu-\frac{C^2}{\ell^{2}}\right)_+\ell^d,\quad \forall \ell \ge C$$
for any ball $B(z,\ell)\subset\Omega$. In particular, if $\Omega$ contains a sufficiently large
ball, then $u$ is not identically zero. 
\end{corollary}

\begin{proof}
From Lemma~\ref{lem:1st_2nd_order} if $\Omega=\R^d$ and standard methods if $\Omega$ is bounded, we know from the second-order variation that
$$\int_{\Omega}|\nabla h|^2+\int_\Omega (|u|^2\ast w-\mu)|h|^2+2\iint_{\Omega\times\Omega}\Re(\overline{u}h)(x)\Re(\overline{u}h)(y)w(x-y)\,\dx\,\dy\geq0$$
for every $h\in H^1(\Omega)\cap L^\ii(\Omega)$ having compact support. Writing
$$2|\Re(\overline{u}h)(x)\Re(\overline{u}h)(y)|\leq |u(x)|^2|h(y)|^2+|u(y)|^2|h(x)|^2$$
we obtain
\begin{equation}
\mu\int_{\Omega}|h|^2\leq \int_{\Omega}|\nabla h|^2+3\int_\Omega (|u|^2\ast |w|)|h|^2.
 \label{eq:estim_lower_mass_pre}
\end{equation}
We assume that $h$ is supported in the ball $B(z,\ell/2)$. Using~\eqref{eq:estim_w_u_L2unif} and the local bound~\eqref{eq:local_bound_mass}, we obtain
\begin{multline*}
\int_\Omega (|u|^2\ast |w|)|h|^2\leq  \left(\int_{\R^d} |w|\right)\left(\int_{B(z,\ell)} |u|^2\right)\|h\|^2_{L^\ii}\\
+\frac{C}{\ell^{s-d}}\left(\int_{\R^d}|h|^2\right)\left(\mu+\frac1{\ell^2}\right).
\end{multline*}
At this point we choose $h=2^{d/2}\ell^{-d/2}\phi_1(2(x-z)/\ell)$ to be the first Dirichlet eigenfunction of the ball $B(z,\ell/2)\subset B(z,\ell)\subset \Omega$ (here $\phi_1$ is the normalized first eigenfunction of the unit ball). Inserting in~\eqref{eq:estim_lower_mass_pre} we obtain
\begin{equation*}
\mu\leq \frac{4\lambda_1(B_1)}{\ell^2}+\frac{C}{\ell^d}\int_{B(z,\ell)} |u|^2+\frac{C}{\ell^{s-d}}\left(\mu+\frac1{\ell^2}\right).
\end{equation*}
This gives the claimed inequality for $\ell$ large enough.
\end{proof}

Using standard elliptic techniques, the average bounds from Theorem~\ref{thm:local_bound} provide pointwise bounds. The average bound~\eqref{eq:local_bound_mass} suggests that $|u(x)|\leq C\sqrt{\mu}$ for small $\mu$. 
This is achieved in dimensions $d\in\{1,2,3\}$ in Corollary \ref{cor:pointwise_bounds_I} below. The case of higher dimensions requires a different argument, which we provide later in Section~\ref{sec:proof_uniform_bound_low_density}.

For simplicity, we only state the corresponding bounds for the Dirichlet boundary condition and make later a remark about the Neumann case.

\begin{corollary}[Uniform bounds I]\label{cor:pointwise_bounds_I}
Let $d\geq 1$ and $w$ be a potential satisfying Assumption~\ref{ass:w}. Let $\Omega\subset\R^d$ be a (bounded or unbounded) smooth domain. Let $\mu>0$ and $u\in H^1_{\rm unif}(\Omega)$ be an arbitrary solution of the GP equation~\eqref{eq:GP_local} in the sense of distributions, satisfying the Dirichlet boundary condition $u_{|\partial\Omega}\equiv0$ if $\Omega\neq\R^d$ and having a finite local interaction in the sense of~\eqref{eq:finite_interation_u}. Then we have
\begin{equation}
\|u\|_{L^\ii(\Omega)}\leq \begin{cases}
C\sqrt\mu \left(1+\mu^{\frac{d}{4}}\right)&\text{if $d\leq 3$,}\\
C\left(1+\mu^{\frac{d+2}{4}}\right)&\text{if $d\geq4$.}
\end{cases}
\label{eq:pointwise_bound}
\end{equation}
In particular, \eqref{eq:pointwise_bound_main} holds if $\mu$ is large or $d\le 3$. 
\end{corollary}

\begin{proof}
In dimension $d=1$, we just use our local bounds~\eqref{eq:local_bound_mass} and~\eqref{eq:local_bound_kinetic_interaction} and the fact that $H^1_{\rm unif}$ is embedded into $L^\ii$. Precisely, using $(|u|^2)'=2\Re(\bar u u')$, we obtain over any interval $I\subset\R$ of length $2\ell$
$$\max_I|u(x)|^2\leq \frac1{2\ell}\int_I|u|^2+2\left(\int_I |u|^2\right)^{\frac12}\left(\int_I|u'|^2\right)^{\frac12}.$$
Inserting our local bounds from Theorem~\ref{thm:local_bound} gives
$$\max_I|u(x)|^2\leq C\left(\mu+\ell^{-2}\right)+C\ell\left(\mu+\ell^{-2}\right)^{\frac32}\leq C
\left(1+\sqrt{\mu}\ell\right)\left(\mu+\ell^{-2}\right).$$
Taking $\ell=\max(1/\sqrt\mu,C)$ for a sufficiently large constant $C$ so as to be able to apply Theorem~\ref{thm:local_bound}, we obtain~\eqref{eq:pointwise_bound} in $d=1$.

Next we turn to larger dimensions. By Kato's inequality (see, e.g., ~\cite[Thm.~X.27]{ReeSim2}), we have
\begin{equation}
 -\Delta |u|\leq (\mu-w\ast|u|^2)|u|\leq (\mu+w_-\ast|u|^2)|u|.
 \label{eq:Kato}
 \end{equation}
We now recall the inequality
\begin{equation}
\norm{v}^2_{L^\ii(B(z,\ell))}\leq C\left(1+\ell^2\norm{V}_{L^q(B(z,2\ell))}^{\frac{2q}{2q-d}}\right)^{\frac{d}2}\left(\frac1{\ell^d}\int_{B(z,2\ell)}v^2\right)
 \label{eq:Harnack}
\end{equation}
for any Dirichlet non-negative subsolution $-\Delta v\leq Vv$ on $\Omega$, with $0\leq V\in L^q(B(z,2\ell))$ and $d/2<q\leq\ii$. We only need the Dirichlet boundary condition if $B(z,2\ell)\cap\partial\Omega\neq\emptyset$. The proof is for instance in~\cite{Trudinger-67,Trudinger-73,Agmon} and uses the Moser iteration technique~\cite{Moser-60,Moser-61}. The constant $C$ only depends on $q$ and $d$. If $V=V_1+V_2$ with $V_1\in L^{q}_{\rm loc}(\Omega)$ and $V_2\in L^\ii_{\rm loc}(\Omega)$, one can replace $\norm{V}_{L^q}^{\frac{2q}{2q-d}}$ by $\norm{V_1}_{L^q}^{\frac{2q}{2q-d}}+\norm{V_2}_{L^\ii}$ in~\eqref{eq:Harnack}. In the rest of the proof we assume for simplicity of notation that the center of the ball is $z=0$. Next we split
$$w_-\ast |u|^2=w_-\ast(|u|^2\1_{B_{4\ell}})+w_-\ast(|u|^2\1_{B^c_{4\ell}})$$
on $B_{2\ell}$ and estimate the last term as in Lemma~\ref{lem:estim_potential}. Using our local bound~\eqref{eq:local_bound_mass} we deduce that
\begin{equation}
\norm{u}_{L^\ii(B_\ell)}^2\leq C\left(1+\ell^2\mu+\ell^2\norm{w_-\ast(|u|^2\1_{B_{4\ell}})}_{L^q(B_{2\ell})}^{\frac{2q}{2q-d}}\right)^{\frac{d}2}\left(\mu+ \ell^{-2}\right).
\label{eq:bound_u_Harnack}
\end{equation}
Like in Lemma~\ref{lem:estim_potential} we can use
\begin{equation}
\norm{w_-\ast (|u|^2\1_{B_{4\ell}})}_{L^\ii}\leq \|w_-\|_{L^\ii}\int_{B_{4\ell}}|u|^2\leq C\ell^d(\mu+\ell^{-2}).
 \label{eq:estim_potential_Harnack}
\end{equation}
Thus, after taking $\ell=C$ a large enough constant and $q=+\ii$, we obtain $\|u\|_{L^\infty}\le C (1+\mu)^{\frac {d +2 } 4}$ in all dimensions $d\geq2$. In particular, we have proved the stated bound~\eqref{eq:pointwise_bound} for $d\ge 4$ with all $\mu>0$ and for $d\in \{2,3\}$ with $\mu\geq C$. 

It remains to handle small $\mu$'s in dimensions $d\in\{2,3\}$. Like in 1D, we take $\ell=1/\sqrt{\mu}$ and bound, this time,
\begin{equation}
\norm{w_-\ast(|u|^2\1_{B_{4\ell}})}_{L^{q}}\leq \|w_-\|_{L^1}\norm{u}_{L^{2q}(B_{4\ell})}^2.
 \label{eq:estim_potential_Harnack_2}
\end{equation}
The Sobolev embedding $H^1(B_{4\ell})\subset L^{2q}(B_{4\ell})$ and our local bounds from Theorem~\ref{thm:local_bound} give
\begin{align*}
\|u\|^2_{L^{2q}(B_{4\ell})}&\leq C\left(\ell^{\frac{d}{q}-d+2}\int_{B_{4\ell}}|\nabla u|^2+\ell^{\frac{d}{q}-d}\int_{B_{4\ell}}|u|^2\right)\\
&\leq C\ell^{\frac{d}{q}}\mu=C\ell^{\frac{d-2q}{q}}
\end{align*}
for all $1\leq q<\ii$ if $d=2$ and for all $1\leq q\leq \frac{d}{d-2}$ if $d\geq3$. Recall that we also need $q>d/2$ in order to apply the estimate~\eqref{eq:bound_u_Harnack}. This is only possible if $\frac{d}{d-2}>\frac{d}2$, which means $d\leq3$. The potential term in~\eqref{eq:bound_u_Harnack} is bounded by
$$\ell^2\norm{w_-\ast|u|^2\1_{B_{4\ell}}}_{L^q(B_{2\ell})}^{\frac{2q}{2q-d}}\leq C\ell^2\ell^{\frac{d-2q}{2q}\frac{4q}{2q-d}}=C.$$
Inserting in~\eqref{eq:bound_u_Harnack} provides the bound $|u|\leq C\sqrt\mu$ for $\mu$ small enough.
\end{proof}

\begin{remark}[Neumann case]\label{rmk:Neumann_pointwise}
For Neumann, we apply~\eqref{eq:bound_u_Harnack} in a ball $B_\ell$ so that $B_{4\ell}\subset \Omega$. We therefore get the same upper bound~\eqref{eq:pointwise_bound} but only at a distance $\max(C,\mu^{-1/2})$ to the boundary $\partial\Omega$. Close to the boundary the bound~\eqref{eq:Harnack} is still valid, but with a constant depending on the regularity of the domain $\Omega$~\cite{Trudinger-67,Trudinger-73}, hence we also get a bound on $u$ close to the boundary which depends on $\Omega$.
\end{remark}

\subsection{Local bounds for inhomogeneous boundary conditions}\label{sec:proof_inhomogeneous}

Let $\Omega$ be a smooth open set (bounded or unbounded) and let $g\in C^{1}_b(\partial\Omega)$. Let $u\in H^1_{\rm unif}(\Omega)$ be a solution of
\begin{equation}
\begin{cases}
-\Delta u+w\ast|u|^2\,u=\mu\,u,&\text{in $\Omega$,}\\
u=g&\text{on $\partial\Omega$.}
\end{cases}
\label{eq:GP_BC}
\end{equation}
We also assume that $u$ has a finite local interaction, as in~\eqref{eq:finite_interation_u}.
We will provide local bounds on $u$ that are similar to Theorem~\ref{thm:local_bound} and depend on~$g$. To state our result, it is easier to give ourselves any extension $G\in C^1_b(\overline \Omega)$ of $g$ to the whole domain $\Omega$. The function $G$ can be defined in various ways so as to have
$$\norm{G}_{L^\ii(\Omega)}+\norm{\nabla G}_{L^\ii(\Omega)}\leq C_\Omega\left(\norm{g}_{L^\ii(\partial\Omega)}+\norm{\nabla g}_{L^\ii(\partial\Omega)}\right)$$
but the constant $C_\Omega$ typically depends on the regularity properties of $\Omega$. Our estimate in terms of $G$ will, on the contrary, be completely independent of $\Omega$.

The following is proved by following the method we used for Theorem~\ref{thm:local_bound} and we thus state it as a corollary.

\begin{corollary}[Local bounds for inhomogeneous boundary conditions] \label{cor:inh-bdr}
There exists a constant $C>0$, depending only on $d$ and $w$, such that
\begin{equation}
\int_{B(z,\ell)}|u|^2\leq C\left(\mu+\frac1{\ell^2}+\|G\|_{L^\ii(\Omega)}^2+\|\nabla G\|_{L^\ii(\Omega)}\right)\ell^d
\label{eq:local_bound_mass_BC}
\end{equation}
and
\begin{multline}
 \int_{B(z,\ell)}|\nabla u|^2+\iint_{B(z,\ell)^2}|u(x)|^2|u(y)|^2w_+(x-y)\,\dx\,\dy\\
 \leq C\left(\mu+\frac1{\ell^2}+\|G\|_{L^\ii(\Omega)}^2+\|\nabla G\|_{L^\ii(\Omega)}\right)^2\ell^d
 \label{eq:local_bound_kinetic_interaction_BC}
\end{multline}
for any $z\in\R^d$ and any $\ell\geq C$.
\end{corollary}

\begin{proof}
We integrate~\eqref{eq:GP_BC} against $\chi^2(\overline{u}-\overline{G})$ for some $\chi\in C^\ii_c(\R^d,[0,1])$ and take the real part. Since $u-G$ vanishes at the boundary, we have no boundary term and obtain
\begin{multline*}
\int_\Omega |\nabla (\chi u)|^2-\int_\Omega|\nabla \chi|^2|u|^2+D\big(|u|^2,\chi^2|u|^2-\chi^2\Re(\overline u G)\big)\\
=\mu\Re\int_\Omega\chi^2\overline{u}(u-G)+\Re\int_\Omega\nabla\overline u\cdot \nabla (\chi^2G).
\end{multline*}
Our goal is to estimate the terms involving $G$ to reduce the proof to that of Theorem~\ref{thm:local_bound}. We have
\begin{equation}
\mu\Re\int_\Omega\chi^2\overline{u}(u-G)\leq \frac{3\mu}2\int_\Omega\chi^2|u|^2+\frac{\mu}2\int_\Omega\chi^2|G|^2
 \label{eq:estim_mu_inhomog}
\end{equation}
and
\begin{align*}
&\int_\Omega\nabla\overline u\cdot \nabla (\chi^2G)\\
&\quad=\int_\Omega\nabla\overline u\cdot \big(2\chi G\nabla\chi+\chi^2\nabla G\big)\\
&\quad=\int_\Omega\big(\nabla(\chi\overline u)-\overline u\nabla \chi\big)\cdot \big(2 G\nabla\chi+\chi\nabla G\big)\\
&\quad\leq \frac12\int_\Omega|\nabla\chi u|^2+\frac12\int_\Omega |u|^2|\nabla \chi|^2+4\int_\Omega |\nabla\chi|^2|G|^2+\int_\Omega \chi^2|\nabla G|^2.
\end{align*}
Thus we obtain
\begin{multline*}
\frac12 \int_\Omega |\nabla (\chi u)|^2+D\big(|u|^2,\chi^2|u|^2-\chi^2\Re(\overline u G)\big)\\
\leq \frac{3\mu}2\int_\Omega\chi^2|u|^2+\frac32\int_\Omega|\nabla \chi|^2|u|^2
+4\|G\|_\ii^2\int_\Omega |\nabla\chi|^2+\|\nabla G\|_\ii^2\int_\Omega \chi^2+\frac\mu2\|G\|_\ii^2\int_\Omega\chi^2.
\end{multline*}
We now estimate the interaction term
\begin{multline*}
D\big(|u|^2,\chi^2|u|^2-\chi^2\Re(\overline u G)\big)\\
=\eps\int (\delta_r\ast|u|^2)\delta_r\ast\big(\chi^2|u|^2-\chi^2\Re(\overline u G)\big)+D_{w_2}\big(|u|^2,\chi^2|u|^2-\chi^2\Re(\overline u G)\big).
\end{multline*}
For the first term, we use $\Re(\bar u G)\leq |u|^2/4+|G|^2$ and obtain
\begin{align*}
&\int (\delta_r\ast|u|^2)\delta_r\ast\big(\chi^2|u|^2-\chi^2\Re(\overline u G)\big)\\
\qquad &\geq\int (\delta_r\ast|u|^2)\delta_r\ast\left(\frac34\chi^2|u|^2-\chi^2|G|^2\right)\\
\qquad &\geq\frac34\int (\delta_r\ast\chi^2|u|^2)^2-\|G\|_\ii^2\int_\Omega \chi^2|u|^2-\|G\|_\ii^2\int_\Omega(1-\chi^2)|u|^2(\delta_r\ast\delta_r\ast\chi^2)
\end{align*}
The stability assumption on $w_2$ implies
$$D_{w_2}(\chi^2|u-G/4|^2,\chi^2|u-G/4|^2)\geq0$$
which, after expanding, provides a lower bound in the form
$$D_{w_2}\big(\chi^2|u|^2,\chi^2|u|^2-\chi^2\Re(\overline u G)\big)\geq -C\|G\|_\ii^2\int_\Omega \chi^2|u|^2-C\norm{G}_\ii^4\int_\Omega \chi^2$$
where $C$ is proportional to $|w_2|(\R^d)$. On the other hand, since $|u|^2-\Re(\overline u G)+|G|^2\geq0$ we can use the pointwise lower bound on $w_2$ and obtain
\begin{align*}
&D_{w_2}\big((1-\chi^2)|u|^2,\chi^2|u|^2-\chi^2\Re(\overline u G)+\chi^2|G|^2\big)\\
&\qquad\geq -D_{2}\big((1-\chi^2)|u|^2,\chi^2|u|^2-\chi^2\Re(\overline u G)+\chi^2|G|^2\big)\\
&\qquad\geq -D_{2}\big((1-\chi^2)|u|^2,5\chi^2|u|^2/4+2\chi^2|G|^2\big).
\end{align*}
As a conclusion, we have proved a bound in the form
\begin{multline*}
\frac12 \int_\Omega |\nabla (\chi u)|^2+\frac{3\eps}4\int (\delta_r\ast\chi^2|u|^2)^2\\
\leq  \frac54 D_2((1-\chi^2)|u|^2,\chi^2|u|^2)+2D_2((1-\chi^2)|u|^2,\chi^2|G|^2)\\
+D_{|w_2|}((1-\chi^2)|u|^2,\chi^2|G|^2)+\eps\|G\|_\ii^2\int_\Omega(1-\chi^2)|u|^2(\delta_r\ast\delta_r\ast\chi^2)\\
+\frac{3\mu+C\|G\|_\ii^2}2\int_\Omega\chi^2|u|^2+\frac32\int_\Omega|\nabla \chi|^2|u|^2\\
+4\|G\|_\ii^2\int_\Omega |\nabla\chi|^2+\|\nabla G\|_\ii^2\int_\Omega \chi^2+\frac{\mu\|G\|_\ii^2+C\|G\|_\ii^4}2\int_\Omega\chi^2.
\end{multline*}
We apply it to the same $\chi$ as in the proof of Theorem~\ref{thm:local_bound} and obtain (using $\ell_0>2r$)
\begin{align*}
\frac12 \int_{\cB_j} |\nabla u|^2+\frac{3\eps}4M_j\leq &\frac32 D_2(\1_{\cB_{j+1}}|u|^2,\1_{\cB_j^c}|u|^2)+3 D_2(\1_{\cB_{j+1}}|G|^2,\1_{\cB_{j+2}^c}|u|^2)\\
&+C\left(\mu+\|G\|_\ii^2+\frac{1}{\ell_0^2}\right)\int_{\cB_{j+2}}|u|^2\\
&+C\left(\ell_0^{-2}\|G\|_\ii^2+\|\nabla G\|_\ii^2+\mu \|G\|_\ii^2+\|G\|_\ii^4\right)V_{j+1}.
\end{align*}
The second term is new but when $j=q>0$ it can be treated as in~\eqref{eq:ell0_cons_2} in the proof of Lemma~\ref{lem:error_estim}. Choosing $\ell_0$ large enough, we obtain in the case that $q>0$
$$\frac32 D_2(\1_{\cB_{j+1}}|u|^2,\1_{\cB_j^c}|u|^2)+3 D_2(\1_{\cB_{j+1}}|G|^2,\1_{\cB_{j+2}^c}|u|^2)\leq \frac{\eps}{2}M_q,$$
which implies
$$\frac{\eps}{4}\leq C\frac{\mu+\|G\|_\ii^2+\ell_0^{-2}}{\sqrt{\psi_0}}+C\frac{\ell_0^{-4}+\|\nabla G\|_\ii^2+\mu^2+\|G\|_\ii^4}{\psi_0}.$$
This is a contradiction if we choose
$$\psi_0=C'\left(\mu+\|G\|_\ii^2+\ell_0^{-2}+\|\nabla G\|_\ii\right)^2$$
for a large enough constant $C'$. We therefore deduce that $q=0$ and can conclude similarly as for Theorem~\ref{thm:local_bound}.
\end{proof}

The local bounds provide uniform bounds away from the boundary by arguing as in Corollary~\ref{cor:pointwise_bounds_I}. Close to the boundary, the bounds again depend on $g$ and the regularity of $\Omega$. We do not state these bounds here, since we will only need them in dimension $d=2$ for $\Omega$ a disk and for a very special~$g$, see Lemma~\ref{lem:uniform_bound_u_L} below.

\subsection{Triviality of GP solutions for $\mu=0$}\label{sec:mu_0}

In dimensions $d\geq4$ our uniform bound~\eqref{eq:pointwise_bound} is not as good as we would like for small $\mu$. 
In order to prove the desired bound $|u|\leq C\sqrt\mu$, our argument is to argue by contradiction and study how would the solution behave in the limit $\mu\to0$ in case this estimate fails. For simplicity we only consider the case of the whole space $\Omega=\R^d$.

The first step in the proof is to show that, under our assumption on $w$, the GP equation admits no solution at all when $\mu=0$. The local bounds in Theorem~\ref{thm:local_bound} are perfectly valid for $\mu=0$. In dimension $d=1$ we find immediately $u\equiv0$ after letting $\ell\to\ii$ in~\eqref{eq:local_bound_mass}, since $\int_{B(z,1)}|u|^2\leq \int_{B(z,\ell)}|u|^2\leq C\ell^{d-2}=C\ell^{-1}\to0$. In dimensions $d\in\{2,3\}$ the same argument works with the kinetic energy~\eqref{eq:local_bound_kinetic_interaction}. In dimensions $d \geq 4$ the situation is more complicated and we need the following.

\begin{proposition}[Triviality of solutions for $\mu=0$]\label{prop:mu_0}
Let $d\geq 1$ and $w$ be a potential satisfying Assumption~\ref{ass:w}. If $\mu=0$, then the unique solution to the GP equation~\eqref{eq:GP_local} in $H^1_{\rm unif}(\R^d)$ of finite local interaction energy is $u\equiv0$.
\end{proposition}

\begin{proof} Our goal is to show that $u\in L^2(\R^d)$, and then derive $u\equiv 0$ using the GP equation and the integrability. Let $B_R=B(0,R)\subset \R^d$ and $S_R=\partial B(0,R)$ with $R>0$ large. Let us compute
\begin{align} \label{eq:BRu-BRcu-compact-0}
\frac{d}{dR} \left( R^{1-d}\int_{S_R} |u|^2 \right) = R^{1-d} \int_{B_R} \Delta (|u|^2) \ge 2 R^{1-d}  \int_{B_R}  |u|^2 (w*|u|^2)
\end{align}
where we used the GP equation ~\eqref{eq:GP_local} to get
$$
\Delta (|u|^2) = 2 \overline{u} \Delta u + 2|\nabla u|^2 = 2 |u|^2 (w*|u|^2) + 2|\nabla u|^2 \ge 2 |u|^2 (w*|u|^2).
$$
To estimate further the right-hand side of \eqref{eq:BRu-BRcu-compact-0}, let us recall that $w=\eps \delta_r * \delta_r + w_2$ with $w_2$ satisfying \eqref{eq:w2_regular} for some $s>d$.   Then using the upper bound $\int_{B_r} |u|^2 \le Cr^{d-2}$ from \eqref{eq:local_bound_mass},  we have for every $L>0$ large enough,
  \begin{align*}
  & \|\1_{B_R} ((w_2)_-* (\1_{B_{R+L}^c} |u|^2))\|_{L^\infty} \nn\\
  & \le C \int_{|x|\ge L} \frac{|u(x)|^2}{|x|^s} \dx = C \sum_{k=1}^N  \int_{2^{k+1}L \ge |x|\ge 2^k L} \frac{|u(x)|^2}{|x|^s} \dx\nn\\
   &\le C \sum_{k=1}^N   \frac{(2^{k+1}L)^{d-2}}{(2^k L)^s} \le C L ^{d-s-2}.
\end{align*}
Combining with the stability of $w_2$ and recalling the definition of $K$ from \eqref{eq:def_K}, we obtain
  \begin{align*}  
&\int_{\R^d} \1_{B_R} |u|^2 (w_2*|u|^2) \ge \int_{\R^d} (\1_{B_R} |u|^2) (w_2* (\1_{B_{R}^c}|u|^2) ) \nn\\
&\quad = \int_{\R^d} (\1_{B_R} |u|^2) (w_2* (\1_{B_{R+L}\backslash B_R}|u|^2) ) +  \int_{\R^d} (\1_{B_R} |u|^2) (w_2* (\1_{B_{R+L}^c} |u|^2) )   \nn\\
&\quad \ge -  K  \int_{\R^d} ((\1_{B_R} |u|^2) *\delta_r)  ( (1+|x|^s)^{-1} * \delta_r* (\1_{B_{R+L}\backslash B_R}|u|^2) )  \nn\\
& \qquad - \|\1_{B_R} ((w_2)_-* (\1_{B_{R+L}^c} |u|^2))\|_{L^\infty}   \int_{B_R} |u|^2 \nn\\
&\quad\ge - \frac{\eps}{2} \|(\1_{B_R} |u|^2) *\delta_r \|_{L^2}^2 - \frac{C}{\eps} \|  \delta_r* (\1_{B_{R+L}\backslash B_R}|u|^2)\|_{L^2}^2 - CL ^{d-s-2} \int_{B_R} |u|^2.
\end{align*}
Thus we deduce from \eqref{eq:BRu-BRcu-compact-0} that
\begin{align}\label{eq:SR-BR-integral-pre-general}
&\frac{d}{dR} \left( R^{1-d}\int_{S_R} |u|^2 \right) \ge - C L ^{d-s-2}  R^{1-d}  \int_{B_R} |u|^2 \\
&\qquad \qquad \qquad +  R^{1-d} \left( \eps \| (\1_{B_R} |u|^2) * \delta_r \|_{L^2}^2 - C\eps^{-1} \| (\1_{B_{R+L}\backslash B_R}|u|^2)* \delta_r \|_{L^2}^2 \right) \nn.
\end{align}
Next, let $\eta>0$ be sufficiently small such that $(1-\eta) (d-s-2) \le -2-\eta$ and choose
\begin{align}\label{eq:R0-L}
  R_0= L^{1/(1-\eta)}.
\end{align}
Using again $\int_{B_r} |u|^2 \le Cr^{d-2}$, we can find $R_1 \in [R_0^{d+1},R_0^{d+2}]$ such that
\begin{align}\label{eq:Rn-choice}
R_1^{1-d}\int_{S_{R_1}} |u|^2 \le C R_1^{-2} \le CR_0^{-2(d+1)}.
\end{align}
Then integrating  \eqref{eq:SR-BR-integral-pre-general} over $R\in [R_0,R_1]$, we get
\begin{align}\label{eq:SR-BR-integral}
& C R_0^{-2(d+1)} - R_0^{1-d}\int_{S_{R_0}} |u|^2 \ge - C R_0^{-2-\eta}  \int_{R_0}^{R_1} R^{1-d} \left( \int_{B_R} |u|^2 \right)  \rd R \\
&\quad + \int_{R_0}^{R_1}  R^{1-d} \left( \eps \| (\1_{B_R} |u|^2) * \delta_r \|_{L^2}^2 - C\eps^{-1} \| (\1_{B_{R+L}\backslash B_R}|u|^2)* \delta_r \|_{L^2}^2 \right) \rd R, \nn
\end{align}
where we have used \eqref{eq:Rn-choice} for the first term on the left-hand side and used \eqref{eq:R0-L} for the first term on the right-hand side.

To estimate further, let us take $M\in \mathbb{N}$ such that
$MR_0 \le R_1 \le (M+1)R_0$.
Then we can decompose
\begin{align*}
\int_{R_0}^{R_1}  R^{1-d} \| (\1_{B_R} |u|^2) * \delta_r \|_{L^2}^2\, \rd R & \ge \sum_{k=1}^{M-1} \int_{kR_0}^{(k+1)R_0}  R^{1-d} \| (\1_{B_R} |u|^2) * \delta_r \|_{L^2}^2\, \rd R \nn\\
&\ge  \sum_{k=1}^{M-1} R_0  ((k+1)R)^{1-d} \| (\1_{B_{kR_0}} |u|^2) * \delta_r \|_{L^2}^2 \nn\\
&\ge C^{-1} R_0 \sum_{k=1}^{M-1}    (kR)^{1-d} \| (\1_{B_{kR_0}} |u|^2) * \delta_r \|_{L^2}^2.
 \end{align*}
On the other hand, we have
\begin{align*}
&\int_{R_0}^{R_1}  R^{1-d}   \| (\1_{B_{R+L}\backslash B_R}|u|^2)* \delta_r \|_{L^2}^2 \, \rd R \nn\\
&\qquad \le \sum_{k=1}^M \sum_{0\le \ell \le R_0/L}  \int_{kR_0+\ell L}^{kR_0+ (\ell +1)L}  R^{1-d}   \| (\1_{B_{R+L}\backslash B_R}|u|^2)* \delta_r \|_{L^2}^2 \, \rd R \nn\\
&\qquad \le \sum_{k=1}^M \sum_{0\le \ell \le R_0/L} L (kR_0)^{1-d}   \| (\1_{B_{kR_0+(\ell+2) L}\backslash B_{kR_0+\ell L} }|u|^2)* \delta_r \|_{L^2}^2 \nn\\
&\qquad \le \sum_{k=1}^M L (kR_0)^{1-d}   \bigg\| \sum_{0\le \ell \le R_0/L}  (\1_{B_{kR_0+(\ell+2) L}\backslash B_{kR_0+\ell L} }|u|^2)* \delta_r \bigg\|_{L^2}^2.
\end{align*}
Using the pointwise estimate
 $$
\sum_{0\le \ell \le R_0/L} \1_{B_{kR_0+(\ell+2) L}\backslash B_{kR_0+\ell L} } \le 3  \1_{B_{(k+1)R_0+2L}} \le 3 \1_{B_{(k+2)R_0}}
$$
we can continue and obtain
\begin{align*}
&\int_{R_0}^{R_1}  R^{1-d}   \| (\1_{B_{R+L}\backslash B_R}|u|^2)* \delta_r \|_{L^2}^2 \, \rd R \nn\\
&\le C L \sum_{k=1}^M   ((k+2)R_0)^{1-d}   \| (\1_{B_{(k+2)R_0} }|u|^2)* \delta_r \|_{L^2}^2 \nn\\
&\le CL \sum_{k=1}^{M-1}   (kR_0)^{1-d}   \| (\1_{B_{kR_0} }|u|^2)* \delta_r \|_{L^2}^2 - CL R_1^{1-d} \| (\1_{B_{2R_1} }|u|^2)* \delta_r \|_{L^2}^2.
\end{align*}
Putting these bounds together, we can estimate the second term on the right-hand side of \eqref{eq:SR-BR-integral} as
\begin{align} \label{eq:SR-BR-integral-secondterm-1}
&\int_{R_0}^{R_1}  R^{1-d} \left( \eps \| (\1_{B_R} |u|^2) * \delta_r \|_{L^2}^2 - C\eps^{-1} \| (\1_{B_{R+L}\backslash B_R}|u|^2)* \delta_r \|_{L^2}^2 \right)\, \rd R\nn\\
&\ge \left( C^{-1} \eps R_0 - C\eps^{-1}L\right) \sum_{k=1}^{M-1}     (kR_0)^{1-d} \| (\1_{B_{kR}} |u|^2) * \delta_r \|_{L^2}^2 \nn\\
&\quad - C\eps^{-1}L R_1^{1-d} \| (\1_{B_{2R_1} }|u|^2)* \delta_r \|_{L^2}^2  \ge - C \eps^{-1} R_0^{-d}
\end{align}
for $\eps>0$ small but fixed. Here we have used $C^{-1} \eps R_0 - C\eps^{-1}L \ge 0$ as $L=R_0^{1-\eta}\ll R_0$ for $R_0$ large  and
$$
L R_1^{1-d} \| (\1_{B_{2R_1} }|u|^2)* \delta_r \|_{L^2}^2 \le C L R_1^{-1} \le C R_0^{-d}
$$
as $\int_{B_r} |u|^2 \le C r^{d-2}$. Inserting \eqref{eq:SR-BR-integral-secondterm-1} in \eqref{eq:SR-BR-integral} and using also $R_1\le R_0^{d+2}$ we conclude that
\begin{align}\label{eq:SR-BR-integral-final-general-v0}
R_0^{1-d}\int_{S_{R_0}} |u|^2 \le CR_0^{-d} + C R_0^{-2-\eta}  \int_{R_0}^{R_0^{d+2}} R^{1-d} \left( \int_{B_R} |u|^2 \right)  \,\rd R.
\end{align}
Since our proof of \eqref{eq:SR-BR-integral-final-general-v0} works equally well with $B_{R}=B (0,R)$ replaced by $z+B_R=B (z,R)$, we also have
\begin{align}\label{eq:SR-BR-integral-final-general}
R_0^{1-d}\int_{z+S_{R_0}} |u|^2 \le CR_0^{-d} + C R_0^{-2-\eta}  \int_{R_0}^{R_0^{d+2}} R^{1-d} \left( \int_{B(z,R)} |u|^2 \right)  \,\rd R
\end{align}
for all $z\in \R^d$ and $R_0$ large (the constant $C$ is independent of $z$ and $R_0$).

Finally, let us conclude that $u\in L^2(\R^d)$ using \eqref{eq:SR-BR-integral-final-general} and a bootstrap argument. Assume that we have proved for some constant $\alpha \ge 0$, the bound
\begin{equation} \label{eq:bootstrap-SR-BR}
R^{-d}\int_{B(z,R)} |u|^2 \le C R^{-d}+ C R^{-2-\alpha}
\end{equation}
for all $z\in \R^d$ and $R$ large (we know that \eqref{eq:bootstrap-SR-BR} holds for $\alpha=0$ by \eqref{eq:local_bound_mass}). Then
inserting \eqref{eq:bootstrap-SR-BR} in \eqref{eq:SR-BR-integral-final-general} we get
$$
R_0^{1-d}\int_{z+S_{R_0}} |u|^2 \le CR_0^{-d} + C R_0^{-2-\alpha-\eta/2}
$$
for all $z\in \R^d$ and $R_0$ large. Averaging over $z$ and relabeling $R_0$ by $R$ we obtain
\begin{equation} \label{eq:bootstrap-SR-BR2}
R^{-d}\int_{B(z,R)} |u|^2 \le  C R^{-d} + C R^{-2-\alpha-\eta/2}.
\end{equation}
for all $z\in \R^d$ and $R$ large. This allows us to increase $\alpha$ in \eqref{eq:bootstrap-SR-BR} to get
\begin{equation} \label{eq:bootstrap-SR-BR3}
R^{-d}\int_{B(z,R)} |u|^2 \le C R^{-d}
\end{equation}
for all $z\in \R^d$ and $R$ large, namely $u\in L^2(\R^d)$. We conclude that  $u\equiv 0$ by integrating the GP equation ~\eqref{eq:GP_local} over $\R^d$ against $u$.
\end{proof}

\subsection{Uniform bound at low density}\label{sec:proof_uniform_bound_low_density}

We are now able to complete the proof of our uniform bounds with the following

\begin{corollary}[Uniform bounds II]\label{cor:pointwise_bounds_II}
Let $d\geq 4$ and $w$ be a potential satisfying Assumption~\ref{ass:w}. Let $\mu>0$ and $u\in H^1_{\rm unif}(\R^d)$ be an arbitrary solution of the GP equation~\eqref{eq:GP_local}, having a finite local interaction in the sense of~\eqref{eq:finite_interation_u}. Then we have
\begin{equation}
\|u\|_{L^\ii(\R^d)}\leq C\sqrt\mu \left(1+\mu^{\frac{d}{4}}\right)
\label{eq:pointwise_bound_d4}
\end{equation}
where $C$ only depends on $d$ and $w$.
\end{corollary}

\begin{proof}
By Corollary~\ref {cor:pointwise_bounds_I} we only have to consider the case of small $\mu$. Let us argue by contradiction and assume that there exists a sequence $\mu_n\to0$ and a GP solution $u_n$ satisfying $\|u_n\|_{L^\ii(\R^d)}/\sqrt{\mu_n}\to\ii$. Recall that we already know that $u_n\in L^\ii$ with the bound~\eqref{eq:pointwise_bound}. Upon changing the origin of space, we can assume that
$$\frac{\|u_n\|_{L^\ii(\R^d)}}{2}\leq |u_n(0)|\leq \|u_n\|_{L^\ii(\R^d)}.$$
If $\|u_n\|_{L^\ii}$ does not tend to 0, then using our uniform bounds we can pass to the limit and we find a nontrivial solution in $H^1_{\rm unif}(\R^d)$ of the GP equation with $\mu=0$. These do not exist by Proposition~\ref{prop:mu_0} hence we conclude that necessarily $\|u_n\|_{L^\ii}\to0$. Now, let us introduce $\eps_n:=\|u_n\|_{L^\ii}$ and
$$\tilde u_n(x)=\eps_n^{-1}u_n(x/\eps_n),$$
which is so that $\|\tilde u_n\|_{L^\ii}=1$ and  $ |\tilde u_n(0)|\geq1/2$. After changing variables we see that
\begin{equation}
-\Delta \tilde u_n+\left(\tilde w_n\ast|\tilde u_n|^2-\frac{\mu_n}{\eps_n^2}\right)\tilde u_n=0
\label{eq:rescaled_mu_0}
\end{equation}
with $\tilde w_n(x):=\eps_n^{-d}w(x/\eps_n)\wto \delta_0$. Next we recall the Landau-Kolmogorov inequality
\begin{equation}
 \|\nabla f\|^2_{L^\ii(\R^d)}\leq C\|\Delta f\|_{L^\ii(\R^d)}\|f\|_{L^\ii(\R^d)}.
 \label{eq:Landau-Kolmogorov}
\end{equation}
The latter can be proved by writing, for instance,
\begin{equation}
f=(-\Delta+M^2)^{-1}(-\Delta f+M^2f)= (2\pi)^{-\frac{d}2}Y_M\ast (-\Delta f+M^2f)
\label{eq:Landau-Kolmogorov_proof}
\end{equation}
where $Y_M(x)=M^{d-2}Y_1(Mx)$ is the Yukawa potential, that is, $\widehat{Y_M}(k)=(|k|^2+M^2)^{-1}$. Taking the gradient and using that $\nabla Y_1\in L^1(\R^d)$ yields
 $$\|\nabla f\|_{L^\ii(\R^d)}\leq (2\pi)^{-\frac{d}2}\norm{\nabla Y_1}_{L^1(\R^d)}\left(M^{-1}\|\Delta f\|_{L^\ii(\R^d)}+M\|f\|_{L^\ii(\R^d)}\right)$$
which gives~\eqref{eq:Landau-Kolmogorov} after optimizing over $M$. Using~\eqref{eq:Landau-Kolmogorov} and Equation~\eqref{eq:rescaled_mu_0}, we obtain
 $$ \|\nabla \tilde u_n\|_{L^\ii}\leq C\left(\frac{\mu_n}{\eps_n^2}+\int_{\R^d} |w|\right)^{\frac12}.$$
Recall that $\mu_n/\eps_n^2\to0$ by assumption. Thus $\tilde u_n$ is uniformly Lipschitz and we can find a uniform local limit $u$ by Ascoli, after extraction. Passing to the limit $n\to\ii$ in~\eqref{eq:rescaled_mu_0} we obtain a non-trivial solution of the Ginzburg-Landau equation $\Delta u=|u|^2u\int_{\R^d}w$. But these do not exist by Proposition~\ref{prop:mu_0} again, this time for $w=(\int_{\R^d} w)\delta$ (the argument also follows easily from the maximum principle). Hence we have arrived at a contradiction and conclude, as we wanted, that~\eqref{eq:pointwise_bound_d4} holds (however with an unknown constant).
\end{proof}

\subsection{Bounds on higher derivatives}\label{sec:proof_derivatives}

We can now deduce bounds on the derivatives using the uniform bounds on $u$ and the GP equation.

\begin{corollary}[Bounds on higher derivatives]\label{cor:derivatives}
Let $d\geq 1$ and $w$ be a potential satisfying Assumption~\ref{ass:w}. Let $\mu>0$ and $u\in H^1_{\rm unif}(\R^d)$ be an arbitrary solution of the GP equation~\eqref{eq:GP_local} in $\R^d$, having a finite local interaction in the sense of~\eqref{eq:finite_interation_u}. Then $u$ is real-analytic and all its derivatives are bounded with
\begin{equation}
\norm{\partial^\alpha u}_{L^\ii}\leq C^{|\alpha|}\mu^{\frac{|\alpha|+1}{2}}(1+\mu^{\frac{d}4})^{1+|\alpha|}
\label{eq:uniform_bound_derivatives}
\end{equation}
for some $C$ depending only on $w$ and $d$.
\end{corollary}

\begin{proof}
By the Landau-Kolmogorov inequality~\eqref{eq:Landau-Kolmogorov} and the GP equation, we have
$$ \norm{\nabla u}_{L^\ii}\leq C\left(\mu+\|u\|_{L^\ii}^2\int_{\R^d}|w|\right)^{\frac12}\|u\|_{L^\ii}\leq C\mu(1+\mu^{\frac{d}4})^2.$$
We can then iterate the argument, using that
$$ \norm{\partial_{j_1}\cdots\partial_{j_N} u}_{L^\ii}\leq C\norm{\partial_{j_1}\cdots\partial_{j_{N-1}}(2\mu-w\ast |u|^2)u }_{L^\ii}^{\frac12}\norm{\partial_{j_1}\cdots\partial_{j_{N-1}} u}^{\frac12}_{L^\ii}$$
by~\eqref{eq:Landau-Kolmogorov} and we obtain~\eqref{eq:uniform_bound_derivatives} by induction.
\end{proof}

\begin{remark}
If $w\in L^\ii(\R^d)$, we can use the local bound in Theorem~\ref{thm:local_bound} and Lemma~\ref{lem:estim_potential} to infer $\|w\ast|u|^2\|_{L^\ii}\leq C\mu$ and thus get the better estimate
\begin{equation}
\norm{\nabla u}_{L^\ii}\leq C\mu(1+\mu^{\frac{d}4}).
\label{eq:uniform_bound_gradient_better}
\end{equation}
\end{remark}

In a bounded domain it is possible to estimate higher derivatives as well, but the constants depend on the regularity of the domain. For instance, let us consider a solution $u$ of the GP equation satisfying the Dirichlet boundary condition in a ball $B_R$. Then we get
\begin{equation}
\norm{\nabla u}_{L^\ii(B_R)}\leq C\mu(1+\mu^{\frac{d}4})^2,
\label{eq:Dirichlet_estim_derivatives}
\end{equation}
after using the equivalent of the Landau-Kolmogorov inequality in a finite ball, stated later in Lemma~\ref{lem:Landau-Kolomogorov_ball} in Appendix~\ref{app:Landau-Kolmogorov}.

\subsection{Lower bound for positive solutions}\label{sec:proof_Harnack}

It remains to show the lower bound~\eqref{eq:pointwise_lower_bound} for positive solutions, which will conclude the proof of Theorem~\ref{thm:uniform_bound}.

\begin{corollary}[Harnack's inequality for positive solutions in $\R^d$]\label{cor:lower_bound}
Let $d\geq 1$ and $w$ be a potential satisfying Assumption~\ref{ass:w}. Let $\mu>0$ and $u\in H^1_{\rm unif}(\R^d)$ be an arbitrary \textbf{non-negative} solution of the GP equation~\eqref{eq:GP_local} in $\R^d$, having a finite local interaction in the sense of~\eqref{eq:finite_interation_u} and satisfying the lower bound~\eqref{eq:lower_bound} on its local mass. Then we have the lower bound
\begin{equation}
u(x)\geq C^{-1}e^{-C\mu^{\frac{d+2}{4}}}\sqrt\mu.
\label{eq:pointwise_lower_bound2}
\end{equation}
\end{corollary}

\begin{proof}
We rewrite the GP equation in the form
$$(-\Delta+M^2)u=(M^2+\mu-w\ast |u|^2)u$$
and pick
$$M= \left(\int_{\R^d}w_+\right)^{\frac12}\|u\|_\ii\leq C\sqrt\mu \left(1+\mu^{\frac{d}{4}}\right)$$
to ensure $M^2+\mu-w\ast |u|^2\geq 0$. Then we have $(-\Delta+M^2)u\geq0$
and we can use Harnack's inequality in the form of~\cite[Thm.~9.9]{LieLos-01}, leading to
$$u(x)\geq \frac{e^{-cM\ell}}{|B_{\ell}|}\int_{B(x,\ell)}u\geq \frac{e^{-cM\ell}}{\|u\|_{L^\ii}|B_{\ell}|}\int_{B(x,\ell)}u^2\geq \frac{e^{-cM\ell}}{CM}(\mu-\ell^{-2})_+.$$
We obtain~\eqref{eq:pointwise_lower_bound} after choosing as usual $\ell=\max(C,2/\sqrt\mu)$.
\end{proof}
The proof of Theorem~\ref{thm:uniform_bound} is now complete.

\subsection{Existence of positive ground states: Proof of Corollary~\ref{cor:existence-1}}\label{sec:proof_existence}
Let us now prove Corollary~\ref{cor:existence-1} concerning the existence of positive infinite ground states through a thermodynamic limit. We consider a positive Dirichlet or Neumann minimizer $u_n$ in a growing domain $\Omega_n\supset B(0,R_n)$ as in the statement. Since it is a minimizer it solves the GP equation in $\Omega_n$. From the universal bounds in Theorem~\ref{thm:local_bound}, $u_n$ is  bounded in $H^1_{\rm unif}(\Omega)$. In the Dirichlet case, $u_n$ is even bounded in $L^\ii(\Omega_n)$ by Corollary~\ref{cor:pointwise_bounds_I}. This implies that $\Delta u_n$ is bounded and then $\nabla u_n$ is bounded far away from the boundary, by the Landau-Kolmogorov inequality~\eqref{eq:Landau-Kolmogorov} applied to $\chi u_n$. The same holds in the Neumann case, far away from the boundary. Therefore, after extraction of a subsequence, we have $u_n\to u$ uniformly locally.
Using the same decomposition as in Lemma~\ref{lem:estim_potential}, we see that $w\ast u_n^2\to w\ast u^2$ locally uniformly. This allows us to pass to the limit in the GP equation, in the sense of distributions. Similarly, we know that $u_n$ satisfies~\eqref{eq:local_min} for all $h\in H^1(\R^d)$ having compact support and finite interaction energy, as soon as its support is contained in $\Omega_n$. We can pass to the limit in this relation and obtain~\eqref{eq:local_min}. This concludes the proof of Corollary~\ref{cor:existence-1}.\qed

\section{Thermodynamic properties of infinite ground states}\label{sec:thermo_limit}

\subsection{Free energy per unit volume: Proof of Theorem \ref{thm:thermo}}\label{sec:proof_thermo}
In this section we provide the proof of Theorem~\ref{thm:thermo} concerning the properties of the thermodynamic functions $f(\mu)$ and $e(\rho)$. The main difficulty is to show that $f$ is strictly concave.

\subsubsection*{Step 1: Convexity of $e$.}  Let $\rho_1,\rho_2>0$ and $t\in(0,1)$. We split a large cube $C_n$ as
$$C_n=(0,n)^d=(0,tn)\times (0,n)^{d-1}\cup[tn,n)\times (0,n)^{d-1}=:A_n\cup B_n.$$
In $A_n$ we place a minimizer $a_n\geq0$ for $E_{\rm D}(\rho_1|A_n|,A_n)$ and in $B_n$ we place a minimizer $b_n\geq0$ for $E_{\rm D}(\rho_2|B_n|,B_n)$. We then take as trial state the function $u_n:=a_n+b_n\in H^1_0(C_n)$, which has the mass
$$\int u_n^2=\rho_1|A_n|+\rho_2|B_n|=(t\rho_1+(1-t)\rho_2)|C_n|.$$
This gives
\begin{multline*}
E_{\rm D}\big((t\rho_1+(1-t)\rho_2)|C_n|,C_n\big)\leq E_{\rm D}(\rho_1|A_n|,A_n)+E_{\rm D}(\rho_2|B_n|,B_n)\\
+\int_{A_n}\int_{B_n}a_n(x)^2b_n(y)^2w(x-y)\,\dx\,\dy.
\end{multline*}
From Corollary~\ref{cor:pointwise_bounds_I} we know that $a_n$ and $b_n$ are uniformly bounded.
Next we distinguish whether $x\in A_n$ is at a distance $\leq \delta$ to $B_n$ or not, and obtain
\begin{equation*}
\iint_{A_n\times B_n}a_n(x)^2b_n(y)^2|w(x-y)|\,\dx\,\dy
\leq Cn^d\left(\int_{|z|\geq \delta}|w|\,\rd z+C\frac\delta{n}\right)=o(n^d)
\end{equation*}
if we choose $1\ll\delta\ll n$. After dividing by $|C_n|=n^d$ and passing to the limit we obtain
$$e\big(t\rho_1+(1-t)\rho_2\big)\leq te(\rho_1)+(1-t)e(\rho_2).$$
We have used here that $A_n$ and $B_n$ satisfy the Fischer regularity conditions~\eqref{eq:Fisher-0} and \eqref{eq:Fisher}, up to a translation.

\subsubsection*{Step 2: Proof of the Legendre relations.}
Let again $C_n=(0,n)^d$ be the cube of side length $n$ and call $u_n$ the Dirichlet ground state in $C_n$ of mass $\int_{C_n}u_n^2=\lambda_n=\rho|C_n|$. For any $\mu>0$, we have
$$E_{\rm D}(\lambda_n,C_n)=\cE_{C_n}(u_n)=\cF_{\mu,C_n}(u_n)+\mu\lambda_n\geq F_{\rm D}(\mu,C_n)+\mu\lambda_n.$$
Dividing by $|C_n|$ and passing to the thermodynamic limit provides
\begin{equation}
e(\rho)\geq\sup_{\mu>0}\{f(\mu)+\mu\rho\}=:f^*(\rho).
\label{eq:e_lower_f}
\end{equation}
Let now $\mu>0$ and $v_n$ be a minimizer for $F_{\rm D}(\mu,C_n)$, that is,
$$F_{\rm D}(\mu,C_n)=\cE_{C_n}(v_n)-\mu\lambda_n,\qquad\lambda_n:=\int_{C_n}v_n^2.$$
From the local bounds in Theorem~\ref{thm:local_bound} we know that
$$c^{-1}|C_n|\leq\lambda_n\leq c|C_n|$$
and therefore we can assume after extraction of a subsequence that $|C_n|^{-1}\lambda_n\to\bar\rho$ for some $\bar\rho>0$. After passing to the thermodynamic limit, this gives
\begin{equation}
f(\mu)\geq e(\bar\rho)-\mu\bar\rho\geq \inf_{\rho>0}\{e(\rho)-\mu\rho\}=:e_*(\mu).
\label{eq:f_lower_e}
\end{equation}
Together with~\eqref{eq:e_lower_f}, we have proved that $e\geq f^*\geq (e_*)^*$. Since $e$ is convex, it is the Legendre transform of itself, that is, $(e_*)^*=e$. The series of inequalities then implies $e=f^*$ and thus $e_*=f$, as claimed.

In particular, we deduce, as mentioned after the statement, that $e$ and $f$ are $C^1$, except possibly on a countable set where their left and right derivatives do not coincide. A jump in the derivative of $e$ corresponds to an interval where $f$ is linear and conversely. In order to show that $e$ is $C^1$, we thus have to show that $f$ is strictly concave. We do this in Step 4.

\subsubsection*{Step 3: Optimizers have the asymptotic dual density and multiplier.}
We will need the following.

\begin{lemma}\label{lem:dual_mass}
Let $\Omega_n$ be a sequence of domains satisfying~\eqref{eq:Fisher-0} and \eqref{eq:Fisher}. Let $u_n$ be a sequence of minimizers for $E_{\rm D/N}(\lambda_n,\Omega_n)$, with $\lambda_n|\Omega_n|^{-1}\to \rho$. Then the corresponding Lagrange multiplier $\mu_n$ satisfies
\begin{equation}
e'(\rho)_-\leq \liminf_{n\to\ii}\mu_n\leq \limsup_{\to\ii}\mu_n\leq e'(\rho)_+.
 \label{eq:asymp_mu_n}
\end{equation}
Similarly, let $v_n$ be a sequence of quasi-minimizers for $F_{\rm D/N}(\mu,\Omega_n)$, in the sense that
$$\cF_{\mu,\Omega_n}(v_n)=F_{\rm D/N}(\mu,\Omega_n)+o(|\Omega_n|).$$
Then its mass satisfies
$$-f'(\mu)_-\leq \liminf_{n\to\ii}\frac1{|\Omega_n|}\int_{\Omega_n}|v_n|^2\leq\limsup_{n\to\ii}\frac1{|\Omega_n|}\int_{\Omega_n}|v_n|^2\leq-f'(\mu)_+.$$
\end{lemma}

\begin{remark}
Later we will prove that $e$ is $C^1$, in which case~\eqref{eq:asymp_mu_n} simply becomes $\mu_n\to\mu(\rho):=e'(\rho)$.
\end{remark}

\begin{proof}
We start with $v_n$ and denote $\lambda_n:=\int_{\Omega_n}v_n^2$. From the local bounds in Theorem~\ref{thm:local_bound} we know that $c^{-1}\leq|\Omega_n|^{-1}\lambda_n\leq c$. Upon extraction of a subsequence, we assume $\lambda_n|\Omega_n|^{-1}\to\bar\rho>0$. We have from the thermodynamic limit proved in Step 1
$$f(\mu)=\lim_{n\to\ii}\frac{\cE_{\Omega_n}(v_n)-\mu\lambda_n}{|\Omega_n|}\geq e(\bar\rho)-\mu\bar\rho.$$
However $f(\mu)=\min_{\rho>0}\{e(\rho)-\mu\rho\}$, as we have seen, and therefore we conclude that $\bar\rho$ realizes the minimum. Equivalently, $\bar\rho\in [-f'(\mu)_-,-f'(\mu)_+]$ as desired. We also deduce that $|\Omega_n|^{-1}\cE_{\Omega_n}(v_n)\to e(\bar\rho)$.

Let us now consider $u_n$ and assume similarly that $\mu_n\to\bar\mu$. We then use $(1+t)^{1/2}u_n$ as a trial state for $E_{\rm D/N}((1+t)\lambda_n,\Omega_n)$, for some fixed $t$. We obtain, using the Gross-Pitaevskii equation for the linear term in $t$,
\begin{align}
E_{\rm D/N}((1+t)\lambda_n,\Omega_n)&\leq \cE_{\Omega_n}(\sqrt{1+t}u_n)\nn\\
&=E_{\rm D/N}(\lambda_n,\Omega_n)+t\mu_n\lambda_n\nn\\
&\qquad+\frac{t^2}{2}\iint_{\Omega_n^2}w(x-y)|u_n(x)|^2|u_n(y)|^2\,\dx\,\dy.\label{eq:perturb_can}
\end{align}
Passing to the thermodynamic limit provides
$$e(\rho(1+t))-e(\rho)\leq t\bar\mu\rho+O(t^2).$$
Taking $t\to0^+$ and then $t\to 0^-$ provides the claimed bounds.
\end{proof}

\subsubsection*{Step 4: Strict concavity of $f$.}
Let us work in the cube $C_n=(0,n)^d$ and call $u_n\geq0$ a solution of the Neumann grand-canonical problem $F_{\rm N}(\mu,C_n)$. After extraction of a subsequence we know from Lemma~\ref{lem:dual_mass} that $|C_n|^{-1}\lambda_n\to \bar\rho\in [-f'(\mu)_-,-f'(\mu)_+]$ for $\lambda_n=\int_{C_n}u_n^2$. We now perturb $u_n$ by adding a constant in the form $u_n+\eta c\1_{C_n}$ for two  small constants $\eta,c\in(-1,1)$ and use it as a trial state for the chemical potential $\mu+\eta$. Using the equation for $u_n$, all the linear terms in $\eta$ vanish except the one coming from $\mu+\eta$ and we obtain
\begin{align*}
&F_{\rm N}(\mu+\eta,C_n)\\
&\leq\int_{C_n}|\nabla u_n|^2+\frac12\iint_{C_n^2}w(x-y)|u_n(x)+c\eta |^2|u_n(y)+c\eta |^2\,\dx\,\dy\\
&\qquad-(\mu+\eta)\int_{C_n}|u_n+c\eta |^2\\
&=F_{\rm N}(\mu,{C_n})-\eta\lambda_n+\eta^2\bigg(2c^2\int_{C_n} u_n(w\ast u_n)+c^2\int_{C_n} w\ast |u_n|^2\\
&\qquad -\mu c^2|{C_n}|-2\mu c\int_{C_n}u_n\bigg)+\eta^3c^3\int_{C_n}w\ast u_n+\eta^4c^4\iint_{C_n^2}w(x-y)\,\dx\,\dy\\
&\leq F_{\rm N}(\mu,{C_n})-\eta\lambda_n+\eta^2\bigg(3c^2\left(\int_{\R^d} w_+\right)\int_{C_n}|u_n|^2-2c\int_{C_n} u_n\bigg)\\
&\qquad+\eta^3c^3\left(\int_{\R^d} w_+\int_{C_n}u_n+\eta c|C_n|\right).
\end{align*}
From our local bounds we have $\int_{C_n}u_n^2\leq \kappa|C_n|$ and
$$\int_{C_n}u_n\leq |C_n|^{\frac12}\left(\int_{C_n}u_n^2\right)^{\frac12}\leq \kappa |C_n|.$$
There exists also a lower bound, which can be proved using a Harnack inequality inside $C_n$ like in Corollary~\ref{cor:lower_bound}.
Since we have not considered a bounded set there, we provide a different argument. Let $2<p< 2^*$ be any subcritical exponent. By Hölder's inequality, we have
$$\int_{B(z,R)}u_n^2\leq \left(\int_{B(z,R)}u_n\right)^{\frac{p-2}{p-1}}\norm{u_n}_{L^p(B(z,R))}^{\frac{p}{p-1}}\leq \left(\int_{B(z,R)}u\right)^{\frac{p-2}{p-1}}\norm{u_n}_{H^1(B(z,R))}^{\frac{p}{p-1}}.$$
Our local upper bounds on $\int_{B(z,R)}u_n^2+|\nabla u_n|^2$ and our lower bounds on $\int_{B(z,R)}u_n^2$ imply a uniform positive lower bound on $\int_{B(z,R)}u_n$, when $B(z,R)\subset{C_n}$. Therefore, after choosing $c$ small enough and passing to the thermodynamic limit, we find
$$f(\mu+\eta)\leq f(\mu)-\eta\bar\rho -\alpha\eta^2+O(\eta^3)$$
for some $\alpha>0$ depending on $\mu$. In principle $\bar\rho$ is not uniquely defined for a $\mu$. But it is always uniquely defined in places where $f(\mu)$ is linear. The above estimate therefore shows that $f$ cannot be linear at all. Hence it is strictly concave and $e$ is $C^1$. For any $\rho$ there is a unique $\mu(\rho)=e'(\rho)$ so that $e(\rho)=f(\mu(\rho))+\mu(\rho)\,\rho$. It satisfies the bound~\eqref{eq:bound_mu_rho} by Theorem~\ref{thm:local_bound}.

\subsubsection*{Step 5: Concavity of $\rho\mapsto e(\rho)/\rho$.}
After letting $u=\sqrt{\lambda}v$, we can also write the Dirichlet problem in~\eqref{eq:def_I} as
$$E_{\rm D}(\lambda,\Omega):=\lambda\,\min_{\substack{v\in H^1_0(\Omega)\\ \int_\Omega|v|^2=1}}\left\{\int_\Omega|\nabla v|^2+\frac\lambda2\iint_{\Omega\times\Omega}|v(x)|^2|v(y)|^2w(x-y)\dx\,\dy\right\}.$$
This proves that $E_{\rm D}(\lambda,\Omega)/\lambda$ is a concave non-decreasing function of $\lambda$ and thus, after passing to the thermodynamic limit, that $\rho\mapsto e(\rho)/\rho$ is concave non-decreasing.
This concludes the proof of Theorem~\ref{thm:thermo}.
\qed

\subsection{Convergence of canonical minimizers: Proof of Theorem \ref{thm:canonical}}\label{sec:proof_canonical}

We provide the proof that canonical minimizers converge in the thermodynamic limit to a (grand-canonical) infinite ground state of the corresponding $\mu(\rho)$. The proof is based on the following important lemma, inspired by Dobru\v{s}in and Minlos~\cite{DobMin-67}.

\begin{lemma}[First order perturbations]\label{lem:DobMin}
Let $\Omega_n$ be a sequence satisfying the Fischer regularity assumptions~\eqref{eq:Fisher-0} and \eqref{eq:Fisher}. Consider two sequences $\lambda_n$ and $a_n$ such that
$$\frac{\lambda_n}{|\Omega_n|}\to\rho>0,\qquad a_n\to a\in\R.$$
Then we have
$$\lim_{n\to\ii}\Big(E_{\rm D/N}(\lambda_n+a_n,\Omega_n)-E_{\rm D/N}(\lambda_n,\Omega_n)\Big)=\mu(\rho)\,a.$$
\end{lemma}

\begin{proof}[Proof of Lemma~\ref{lem:DobMin}]
Let $u_n$ be a minimizer for $E_{\rm D/N}(\lambda_n,\Omega_n)$, which solves the GP equation
$$(-\Delta+w\ast |u_n|^2)u_n=\mu_n\,u_n.$$
We know by Lemma~\ref{lem:dual_mass} that $\mu_n\to\mu(\rho)>0$. We estimate $E_{\rm D/N}(\lambda_n+a_n,\Omega_n)$ using $(1+a_n/\lambda_n)^{1/2}u_n$ as trial state and obtain from~\eqref{eq:perturb_can}
\begin{align*}
&E_{\rm D/N}(\lambda_n+a_n,\Omega_n)\\
&\qquad \leq E_{\rm D/N}(\lambda_n,\Omega_n)+\mu_n a_n
+\frac{a_n^2}{2\lambda_n^2}\iint_{\Omega_n^2}w(x-y)|u_n(x)|^2|u_n(y)|^2\,\dx\,\dy\\
&\qquad \leq E_{\rm D/N}(\lambda_n,\Omega_n)+\mu_n a_n
+\frac{C}{\lambda_n^2}|\Omega_n|.
\end{align*}
In the last estimate we have used the local bounds on $u_n$. Since the last term goes to zero, we have proved that
$$\limsup_{n\to\ii}\Big(E_{\rm D/N}(\lambda_n+a_n,\Omega_n)- E_{\rm D/N}(\lambda_n,\Omega_n)\Big)\leq \mu(\rho) a.$$
To obtain the reverse inequality on the liminf, we replace $\lambda_n$ by $\lambda_n'=\lambda_n+a_n$ and $a_n$ by $a_n'=-a_n$.
\end{proof}

\begin{proof}[Proof of Theorem~\ref{thm:canonical}]
Let $u_n$ be as in the statement. From the local bounds in Theorem~\ref{thm:local_bound} we can pass to the limit exactly as in the proof of Theorem~\ref{cor:existence-1}. We also know from Lemma~\ref{lem:dual_mass} that $\mu_n\to\mu(\rho)$. The only difficulty is to prove that~\eqref{eq:local_min} holds for all $v$ and not only for those such that $\int(|v|^2-|u|^2)=0$ as is required in $\Omega_n$. Let thus $h$ be any function in $H^1(\R^d)$ having compact support and finite interaction energy. With
$$a_n:=\int_{\Omega_n}(2u_n\Re(h)+|h|^2)\to \int_{\R^d}(2u_\ii\Re(h)+|h|^2)=:a,$$
we write
$$\cE_{\Omega_n}(u_n+h)-\cE_{\Omega_n}(u_n)\geq E_{\rm D/N}(\lambda_n+a_n,\Omega_n)-E_{\rm D/N}(\lambda_n,\Omega_n)=\mu(\rho)a+o(1)$$
by Lemma~\ref{lem:DobMin}. Passing to the limit provides exactly the ground state characterization~\eqref{eq:local_min}.
\end{proof}

\subsection{Properties of ground states: Proof of Theorem \ref{thm:prop_GS}}\label{sec:prop_GS}
Let $u$ be any infinite ground state. For simplicity  we work in the ball $B_R=B(0,R)$, that is, take $\tau=0$. The estimates below are completely independent of the center of the considered ball. We recall that $u$ solves the minimization problem~\eqref{eq:min_u_v} locally in $B_R$, that is, solves the minimization problem
\begin{equation}
F_{u}(\mu,B_R):=\min_{v_{|\partial B_R}=u_{|\partial B_R}}\cF_{\mu,B_R,u}(v)
 \label{eq:F_DLR}
\end{equation}
with
\begin{equation}
\cF_{\mu,B_R,u}(v)=\cF_{\mu,B_R}(v)+\int_{B_R}|v|^2 (|u|^2\1_{\R^d\setminus B_R})\ast w.
 \label{eq:local_energy_bis}
\end{equation}
This takes the same form as $F_{\rm D/N}(\mu,B_R)$ except for the additional potential induced by $|u|^2\1_{\R^d\setminus B_R}$ and the boundary condition that $v=u$ on $\partial B_R$. Using that $u\in L^\ii(\R^d)$ by Theorem~\ref{thm:uniform_bound} we see that
$$\int_{B_R}|u|^2 (|u|^2\1_{\R^d\setminus B_R})\ast w=o(R^d),$$
which then implies
$$F_{u}(\mu,B_R)=\cF_{\mu,B_R,u}(u)=\cF_{\mu,B_R}(u)+o(R^d)\geq F_{\rm N}(\mu,B_R)+o(R^d).$$
To get the reversed bound, we use the trial state $v+\eta u$ where $v$ is a Dirichlet minimizer for $F_{\rm D}(\mu,B_{R})$ and $\eta\in C^\ii$ is a localization function which equals $0$ in $B_{R-1}$ and $1$ outside of $B_R$. Using the local bounds for both $v$ and $u$ we obtain
$$F_{u}(\mu,B_R)\leq F_{\rm D}(\mu,B_R)+o(R^d).$$
Since the thermodynamic limit is the same for Dirichlet and Neumann, this proves that $F_{u}(\mu,B_R)=f(\mu)|B_R|+o(R^d)$ and therefore that $\cF_{\mu,B_R}(u)=f(\mu)|B_R|+o(R^d)$, as was stated in~\eqref{eq:local_free_energy}. Also, we obtain $u\1_{B_R}$ is a sequence of approximate minimizers for the Neumann problem $F_{\rm N}(\mu,B_R)$ and the claims of the statement concerning the mass and GP energy follow from Lemma~\ref{lem:dual_mass}.

It remains to prove that the momentum per unit volume vanishes as in~\eqref{eq:momentum}. The idea is to introduce a momentum in the $j$th direction using
$$\int_{B_R}|\nabla (e^{it x_j}u)|^2=\int_{B_R}|\nabla u|^2+t^2\int_{B_R}|u|^2-2t\Im\left(\int_{B_R}\overline{u}\partial_ju\right).$$
We obtain
\begin{align*}
F_N(\mu,B_R)&=\cF_{\mu,B_R}(u)+o(R^d)\\
&=\cF_{\mu+t^2,B_R}(e^{it x_j}u)+2t\Im\left(\int_{B_R}\overline{u}\partial_ju\right)+o(R^d)\\
&\geq F_{\rm N}(\mu+t^2,B_R)+2t\Im\left(\int_{B_R}\overline{u}\partial_ju\right)+o(R^d).
\end{align*}
After changing $t$ into $-t$, this gives
$$\left|\Im\left(\int_{B_R}\overline{u}\partial_ju\right)\right|\leq \frac{F_{\rm N}(\mu,B_R)-F_{\rm N}(\mu+t^2,B_R)}{|t|}+o(R^d/|t|).$$
After taking first $R\to\ii$ and then $t\to0$ we obtain the result for the imaginary part. An integration by parts shows that the real part equals
$$\Re\left(\int_{B_R}\overline{u}\partial_ju\right)=\frac12\int_{B_R}\partial_j|u|^2=\frac12\int_{\partial B_R}|u|^2\frac{x_j}{R}= O(R^{d-1})$$
where the last bound is because $u\in L^\ii(\R^d)$. \qed

\section{Phase transition}\label{sec:phase_transitions}

In this section we study the fluid-solid transition in Theorem~\ref{thm:phase_transitions}. We first discuss the upper and lower bounds of $\mu_c$, then derive variance and kinetic energy estimates for infinite ground states, before we can conclude the proof of Theorem~\ref{thm:phase_transitions}. The improved bounds in Proposition \ref{prop:improved_mu_c} are explained at the end of this section.

\subsection{Upper bound on $\mu_c$} Let us prove Theorem~\ref{thm:phase_transitions} (i).

We first assume that $\widehat{w}\geq0$. Then we have already seen in Corollary~\ref{cor:cnst} that the constant function is an infinite ground state at all chemical potential $\mu>0$. By Theorem~\ref{thm:prop_GS} we know that any infinite ground state must have the free energy per unit volume $f(\mu)$ and thus conclude that $f(\mu)=-\mu^2/(2\int_{\R^d}w)$. The formulas for $e(\rho)$ and $\mu(\rho)=e'(\rho)$ follow from the Legendre relations in Theorem~\ref{thm:thermo}.

Let us now assume that $\widehat w$ takes negative values. We know that the constant function becomes linearly unstable (that is~\eqref{eq:linearization_cnst_Fourier} fails at least for one $k$) when $\mu$ satisfies the condition in~\eqref{eq:estim_mu_c}, namely $\mu<\min_{k\in\R^d}\frac{|k|^2\widehat{w}(0)}{2\widehat{w}_-(k)}$. Hence the constant function~\eqref{eq:cnst} is not an infinite ground state. In principle, it could however still happen that the free energy equals $f_{\rm cnst}(\mu)$, so we have to use the linear instability to show the claimed strict inequalities in~\eqref{eq:f_not_cnst}. We go back to finite volume and consider the trial state
$$u(x)=\sqrt{\rho(1+t\cos(k_0\cdot x))},\qquad \rho=\frac{\mu}{\int_{\R^d}w},$$
for $t$ small enough. Plugging it in the Neumann problem and taking the thermodynamic limit gives after a calculation
\begin{align*}
f(\mu)&\leq-\frac{\mu^2}{2\int_{\R^d} w}+t^2\left(\frac{|k_0|^2\rho}{8\pi}\int_0^{2\pi}\frac{\sin^2(t)}{1+t\cos(t)}\rd t+\frac{\rho^2}{4}\widehat{w}(k_0)\right)\\
&\leq -\frac{\mu^2}{2\int_{\R^d} w}+\frac{t^2\rho}8\left(\frac{|k_0|^2}{1-t}+\frac{2\mu\widehat{w}(k_0)}{\widehat{w}(0)}\right).
\end{align*}
The correction term is negative for $t$ small enough, when $\mu$ is as in the statement. This concludes the proof of (i) in Theorem~\ref{thm:phase_transitions}. \qed

\subsection{Lower bound on $\mu_c$}
Now we prove  Theorem~\ref{thm:phase_transitions} (ii).  By Lemma~\ref{lem:1st_2nd_order} (iii), the constant function $u_{\rm cnst}=(\mu/\int_{\R^d} w)^{1/2}$ is an infinite ground state if
\begin{equation} \label{eq:const-GS-condition}
\int_{\R^d} |\nabla h|^2 + \frac{1}{2} \iint (2u h(x) + h^2(x)) (2u h(y) + h^2(y)) w(x-y) \dx\, \dy \ge 0
\end{equation}
for every real-valued function $h\in H^1(\R^d)$ having compact support and satisfying $h\ge -u$. Let us decompose $h= h_+ - h_-$, $0\le h_- \le u$
and
$$
f= 2u h+ h^2 = f_+ - f_-, \quad f_+=  2u h_++ h_+^2, \quad f_-= 2u h_- - h_-^2.
$$
Using  $0\le h_- \le u$ we find that
\begin{align} \label{eq:f-vs.h-}
0  \le f_- \le u^2,\quad |\nabla f_-|\le 2 u |\nabla h_-|.
\end{align}
The interaction term in \eqref{eq:const-GS-condition} can be split into three parts as
\begin{align} \label{eq:const-GS-condition-interaction}
\frac{1}{2} \iint f(x) &  f(y) w(x-y) \dx\, \dy = \frac{1}{2} \iint f_+(x) f_+(y) w(x-y) \dx\, \dy\\
&+ \frac{1}{2} \iint f_-(x) f_-(y) w(x-y) \dx\, \dy -  \iint f_-(x) f_+(y) w(x-y) \dx\, \dy. \nn
\end{align}
The diagonal parts in \eqref{eq:const-GS-condition-interaction} are non-negative due to the stability property of $w$.
In order to estimate the off-diagonal integral, we further decompose $f_+$ into two parts that, loosely speaking, correspond to $h\geq 2\alpha u$ or $0\leq h<2\alpha u$, for some large constant $\alpha$ to be determined later. We cannot cut sharply due to the kinetic energy, however, and instead write
$$
f_+ = g_1 + g_2^2, \quad g_1= 2u h_+ + \frac{4\alpha^2 u^2 h_+^2}{4\alpha^2 u^2 +h_+^2}, \quad g_2= \frac{h_+^2}{\sqrt{4\alpha^2 u^2 +h_+^2}}
$$
as well as
\begin{align}\label{eq:const-GS-condition-off-diagonal}
\iint f_-(x) f_+(y) &w(x-y) \dx\, \dy = \iint f_-(x)  g_1(y)  w(x-y)  \dx\, \dy \nn \\
&+  \iint f_-(x)  g_2(y)  w(x-y)  \dx\, \dy.
\end{align}
For the part involving $g_1$, we use the orthogonality of $h_-$ and $g_1$ to estimate
\begin{align}\label{eq:const-GS-condition-off-diagonal-g1}
&\left| \iint f_-(x)  g_1(y)  w(x-y)  \dx\, \dy \right| = (2\pi)^{\frac{d}2} \left| \int \overline{\hat f_-(k)} \hat g_1(k) ( \hat w(0) - \hat w(k) ) \rd k \right| \nn\\
& \le (2\pi)^{\frac{d}2} \int  |\hat f_-(k)|  |\hat g_1(k) | |\hat w(0) - \hat w(k)|  \rd k .
\end{align}
For the part in \eqref{eq:const-GS-condition-off-diagonal} involving $g_2$, we use $|f_-|\le u^2$ and bound
\begin{align}\label{eq:const-GS-condition-off-diagonal-g22}
\left| \iint f_-(x)  g^2_2(y)  w(x-y)  \dx\, \dy \right| \le  u^2 \|w\|_{L^1} \int |\hat g_2(k)|^2 {\rm d}k.
\end{align}

Now let us consider more carefully the diagonal parts in \eqref{eq:const-GS-condition-interaction} and use them to control the off-diagonal term, using the superstability property of $w$. Recalling that $w=\eps \delta_r*\delta_r +w_2$ with $w_2$ stable, we have
\begin{align} \label{eq:const-GS-condition-diagonal-1}
\frac{1}{2} \iint f_{-}(x) f_{-}(y) w(x-y) \dx\, \dy &\ge \frac{\eps}{2} \int |f_{-}* \delta_r|^2 \nn\\
& =  \frac{\eps}{2} (2\pi)^d \int |\hat f_{-} (k)|^2 |\hat \delta_r(k)|^2 {\rm d}k.
\end{align}
Moreover, note that $g_1 \le 2 u (1+\alpha) h_+$ and $g_2 \le \frac{h_+^2}{2 u\alpha}$, so that
$$
f_+ = 2u h_+ + h_+^2 \ge \frac{1}{ 1 + \alpha} g_1 + 2u\alpha g_2.
$$
Therefore,
\begin{align} \label{eq:const-GS-condition-diagonal-2}
&\frac{1}{2} \iint f_{+}(x) f_{+}(y) w(x-y) \dx\, \dy \ge \frac{\eps}{2} \int |f_{+}* \delta_r|^2 \\
&\qquad \ge  \frac{\eps}{2(1 +  \alpha)^2}  \int |g_1* \delta_r|^2 + 2 \eps u^2 \alpha^2   \int |g_2* \delta_r|^2\nn\\
&\qquad = \frac{\eps(2\pi)^d}{2 (1 +  \alpha)^2}   \int |\hat g_1(k)|^2 |\hat \delta_r(k)|^2 {\rm d} k  + 2  \eps u^2 \alpha^2  (2\pi)^d   \int |\hat g_2(k)|^2 |\hat \delta_r(k)|^2   {\rm d} k .\nn
\end{align}
Finally, we consider the kinetic term
$$
\int_{\R^d} |\nabla h|^2 = \int |\nabla h_+|^2 + \int |\nabla h_-|^2.
$$
Since $|\nabla f_-|\le 2 u |\nabla h_-|$, we have
\begin{align} \label{eq:const-GS-condition-kinetic-1}
\int |\nabla h_-|^2 \ge \frac{1}{4u^2} \int |\nabla f_-|^2 = \frac{1}{4u^2} \int |k|^2 |\hat f_-(k)|^2 \rd k.
\end{align}
Moreover, a calculation gives
$$
|\nabla g_1| = \left| 2u \nabla h_+ + \frac{8\alpha^2 u^2 h_+ \nabla h_+ }{4\alpha^2 u^2 +h_+^2} -   \frac{8\alpha^2 u^2 h_+^3  \nabla h_+ }{(4\alpha^2 u^2 +h_+^2)^2} \right| \le 4u (1+\alpha) |\nabla h_+|
$$
and
$$
|\nabla g_2|= \left| \frac{2h_+ \nabla h_+}{\sqrt{4\alpha^2 u^2 +h_+^2}} - \frac{h_+^3 \nabla h_+}{(4\alpha^2 u^2+h_+^2)^{3/2}}  \right| \le \frac{|\nabla h_+|}{u \alpha} ,
$$
so that
\begin{align} \label{eq:const-GS-condition-kinetic-2}
\int |\nabla h_+|^2  \ge \frac{1}{32 u^2(1+\alpha)^2} \int |k|^2 |\hat g_1(k)|^2 + \frac{\alpha^2 u^2}{2} \int |k|^2 |\hat g_2(k)|^2.
\end{align}

In summary, collecting the estimates \eqref{eq:const-GS-condition-off-diagonal-g1}, \eqref{eq:const-GS-condition-off-diagonal-g22}, \eqref{eq:const-GS-condition-diagonal-1},  \eqref{eq:const-GS-condition-diagonal-2} and \eqref{eq:const-GS-condition-kinetic-1} we have proved
\begin{align}\label{eq:const-GS-condition-final-1}
&\int_{\R^d} |\nabla h|^2 + \frac{1}{2} \iint (2u h(x) + h^2(x)) (2u h(y) + h^2(y)) w(x-y) \dx\, \dy \nn\\
&\ge   \int_{\R^d} \Bigg\{ \left( \frac{1}{4u^2} |k|^2 + \frac{\eps}{2} (2\pi)^d |\hat \delta_r(k)|^2\right)   |\hat f_-(k)|^2  \nn\\
&\qquad  \qquad  + \frac{1}{(1+\alpha)^2} \Big( \frac{1}{32 u^2}  |k|^2 +\frac{\eps}{2} (2\pi)^{d} |\hat \delta_r(k)|^2 \Big)  |\hat g_1(k)|^2  \nn\\
&\qquad \qquad - (2\pi)^{\frac{d}2} |\hat f_-(k)|  |\hat g_1(k)|  |\hat w(0) - \hat w(k)| \Bigg\} {\rm d}k \nn\\
&\quad + u^2 \int_{\R^d} \left (  \frac{\alpha^2}{2}  |k|^2 + 2 \alpha^2  \eps (2\pi)^{d}  |\hat \delta_r(k)|^2 - \|w\|_{L^1}  \right) |\hat g_2(k)|^2 {\rm d}k.
\end{align}
It remains to choose $\alpha>0$ to make all relevant integrands non-negative when $\mu$ is small enough. Although we could keep $\widehat{\delta_r}(k)$ everywhere and get a tighter bound, we instead provide a simpler and more explicit estimate. Using \eqref{eq:Taylor-expansion-fourier} we have
\begin{align} \label{eq:deltar-fourier}
|1-(2\pi)^{d/2} \hat \delta_r(k)| \le \frac{|k|^2}{2} \frac{1}{d} \int |x|^2 \delta_r(x) \dx = \frac{|k|^2 r^2}{2 (d+2)} \le  \frac{|k|^2 r^2}{6} .
\end{align}
Hence,
\begin{align*}
|k|^2 + 4 \eps (2\pi)^{d} |\hat \delta_r(k)| ^2 \ge |k|^2 + 4 \eps \left[ 1-  \frac{|k|^2r^2}{6}  \right]_+^2
\ge 2\min (\eps, r^{-2}).
\end{align*}
(the non-optimal factor $2$ was chosen to simplify the final bound). Thus we can choose
$$\alpha^2=\frac{\|w\|_{L^1}}{\min(\eps,r^{-2})}=\|w\|_{L^1} \max (\eps^{-1},r^{2})$$
such that the last integrand involving $|g_2(k)|$ in \eqref{eq:const-GS-condition-final-1} is non-negative. Next, by the Cauchy--Schwarz inequality and \eqref{eq:deltar-fourier}, we find
\begin{align*}
&\left( \frac{|k|^2}{4u^2} + \frac{\eps}{2} (2\pi)^d |\hat \delta_r(k)|^2\right)   |\hat f_-(k)|^2 + \left( \frac{|k|^2}{32 u^2}   +\frac{\eps}{2} (2\pi)^{d} |\hat \delta_r(k)|^2 \right) \frac{ |\hat g_1(k)|^2}{(1+\alpha)^2} \nn\\
&\qquad\ge  \left( \frac{|k|^2}{4\sqrt{2} u^2}  + \eps (2\pi)^d |\hat \delta_r(k)|^2\right)  \frac{|\hat f_-(k)|  |\hat g_1(k)|}{1+\alpha} \nn\\
&\qquad\ge  \left( \frac{|k|^2}{4\sqrt{2} u^2} + \eps \Big[ 1-  \frac{|k|^2r^2}{6}  \Big]_+^2 \right)\frac{|\hat f_-(k)|  |\hat g_1(k)|}{1+\alpha}.
\end{align*}
Therefore, the integrand involving $|g_1(k)|$ and $|f_-(k)|$ in \eqref{eq:const-GS-condition-final-1} is non-negative if
$(2\pi)^{d/2} \hat w(0) u^2=\mu\le \mu_0$ with
$$
\mu_0 = \inf_{k\in \R^d} \frac{\hat w(0) |k|^2}{4\sqrt{2} \Big( (1+\alpha) |\hat w(0) - \hat w(k)| -  \eps (2\pi)^{-\frac{d}2}  \big[ 1- |k|^2r^2/6 \big]_+^2  \Big)_+}>0.
$$
In this case we have the desired inequality \eqref{eq:const-GS-condition}, namely the constant function $u$ is an infinite ground state when $\mu\le \mu_0$.\qed

\subsection{Variance and kinetic energy estimates for infinite ground states}\label{sec:variance-kinetic}
In order to prove the last part of Theorem~\ref{thm:phase_transitions}, we start with the following

\begin{lemma}[Variance and kinetic energy]\label{lem:variance_kinetic}
Let $u$ be any infinite ground state of chemical potential $\mu$. If $\mu<\mu_c$ (or $\mu=\mu_c$ and $f$ is $C^1$ at $\mu_c$), we have for $\rho=\mu/\int_{\R^d}w$
\begin{equation}
\lim_{R\to\ii}\sup_{\tau\in \R^d}\frac1{|B_R|}\int_{B(\tau,R)}\left(|u|^2-\rho\right)^2
=\lim_{R\to\ii}\sup_{\tau\in \R^d}\frac1{|B_R|}\int_{B(\tau,R)}|\nabla u|^2=0.
\label{eq:variance_fluid}
\end{equation}
If $\mu>\mu_c$, we have for $\rho:=\rho(\mu)_-$
\begin{multline}
\liminf_{R\to\ii}\inf_{\tau\in \R^d}\frac1{|B_R|}\int_{B(\tau,R)}\left(|u|^2-\frac1{|B_R|}\int_{B(\tau,R)}|u|^2\right)^2\\
\geq \frac{\rho}{\int_{\R^d}|w|}\left(\rho\int_{\R^d}w-\mu\right)>0
\label{eq:variance}
\end{multline}
and
\begin{equation}
\liminf_{R\to\ii}\inf_{\tau\in \R^d}\frac1{|B_R|}\int_{B(\tau,R)}|\nabla u|^2 \ge 2f(\mu)+\mu\rho >0.
 \label{eq:kinetic_solid}
\end{equation}
\end{lemma}

\begin{proof}[Proof of Lemma~\ref{lem:variance_kinetic}]
Let $u$ be an infinite ground state for some fixed $\mu>0$. Consider any $\tau_n\in\R^d$ and $R_n\to\ii$ such that
the following limits exist and are finite:
$$\fint_{B(\tau_n,R_n)}|u|^2\to\rho,\qquad \fint_{B(\tau_n,R_n)}|\nabla u|^2\to K,\qquad \fint_{B(\tau_n,R_n)}|u|^2(|u|^2\ast w)\to I,$$
with $\fint_\Omega:=|\Omega|^{-1}\int_\Omega$ denoting the average. We obtain
\begin{multline}
\lim_{n\to\ii}\frac1{|B_{R_n}|}\iint_{B(\tau_n,R_n)^2}(|u(x)|^2-\rho)(|u(y)|^2-\rho)w(x-y)\,\dx\,\dy\\
=I-\rho^2\int_{\R^d} w.
\label{eq:lower_bd_for_variance}
\end{multline}
By Theorem~\ref{thm:prop_GS}, we have $\rho\in[\rho_-(\mu),\rho_+(\mu)]$ and
$$K+\frac{I}{2}=f(\mu)+\mu\rho=e(\rho).$$
Integrating the GP equation against $u$ in the same ball $B(\tau_n,R_n)$ and using that the boundary term is a $O(R_n^{d-1})$ due to the uniform bounds on $u$ and $\nabla u$, we find $K+I-\mu\rho=0$ and therefore we deduce that
$$I=-2e(\rho)+2\mu \rho,\qquad K=2e(\rho)-\mu\rho=2e(\rho)-e'(\rho)\rho=2f(\mu)+\mu\rho.$$
If $\mu<\mu_c$ (or $\mu=\mu_c$ and $\rho_-(\mu_c)=\rho_+(\mu_c)=\rho_c$) then $\rho=\mu/\int_{\R^d}w$ and $e'(\rho)=\mu=\rho\int_{\R^d}w$ and thus we obtain $K=0$ as claimed. If $\mu>\mu_c$ we deduce that $K>0$ after using the following

\begin{lemma}\label{lem:e_rho}
For $\rho>\rho_c=-f'_-(\mu_c)$, we have $2e(\rho)>e'(\rho)\rho$.
\end{lemma}

\begin{proof}[Proof of Lemma~\ref{lem:e_rho}]
Recall that the function $\phi(\rho):=e(\rho)/\rho-\rho\int_{\R^d}w/2$ defined previously in~\eqref{eq:def_phi_critical_rho} is concave. It satisfies $\phi(0)=0$ and therefore, by concavity, $\phi'(\rho)<{\phi(\rho)}/{\rho}$ for $\rho>\rho_c$. This is exactly the claimed inequality.
\end{proof}

Next we turn to the variance. If $\mu>\mu_c$, we have
\begin{multline*}
\lim_{n\to\ii}\frac1{|B_{R_n}|}\iint_{B(\tau_n,R_n)^2}(|u(x)|^2-\rho)(|u(y)|^2-\rho)w(x-y)\,\dx\,\dy\\
= -2\left(e(\rho)-\mu \rho+\frac{\rho^2}{2}\int_{\R^d}w\right)<-\rho\left(\rho\int_{\R^d}w-\mu\right).
\end{multline*}
In the last inequality we have used Lemma~\ref{lem:e_rho} again.
We use Young's inequality in the form
\begin{multline}
\iint_{B(\tau_n,R_n)^2}(|u(x)|^2-\rho)(|u(y)|^2-\rho)w(x-y)\,\dx\,\dy\\
\geq-\int_{\R^d}|w|\int_{B(\tau_n,R_n)}(|u|^2-\rho)^2
\end{multline}
and deduce that
$$\liminf_{n\to\ii}\fint_{B(\tau_n,R_n)}(|u|^2-\rho)^2\geq \frac{\rho}{\int_{\R^d}|w|}\left(\rho\int_{\R^d}w-\mu\right)$$
as we wanted.

It remains to prove that the variance goes to 0 when $\mu<\mu_c$ (or $\mu=\mu_c$ and $f$ is $C^1$ at $\mu_c$). In this case we have $\mu=\rho\int_{\R^d}w$ and the limit in~\eqref{eq:lower_bd_for_variance} is zero. The uniform bounds allow us to replace the sharp cutoff by a smooth one so that
\begin{multline*}
\lim_{n\to\ii}\frac1{|B_{R_n}|}\int_{\R^d}|\nabla (\chi_nu)|^2=\lim_{n\to\ii}\frac1{|B_{R_n}|}\int_{\R^d}|\nabla (\chi_nF)|^2\\
=\lim_{n\to\ii}\frac{D(F\chi_n,F\chi_n)}{|B_{R_n}|}=0,
\end{multline*}
where $F=|u|^2-\rho$ and $\chi_n\in C^\ii_c(B(\tau_n,R_n))$ is such that $\chi\equiv1$ on $B(\tau_n,R_n)$ and $\chi\equiv0$ outside of $B(\tau_n,R_n+1)$, with $|\nabla\chi|\leq C$. The second limit is because $\nabla F=2\Re(\overline u\nabla u)$ and $u\in L^\ii(\R^d)$. Passing to the Fourier domain gives
$$\lim_{n\to\ii}\frac1{|B_{R_n}|}\int_{\R^d}|k|^2|\widehat{\chi_nF}(k)|^2\,\rd k=\lim_{n\to\ii}\frac1{|B_{R_n}|}\int_{\R^d}\widehat{w}(k)|\widehat{\chi_nF}(k)|^2\,\rd k=0.$$
Since $\widehat{w}(0)>0$, we have $C|k|^2+\widehat{w}(k)\geq C^{-1}$ for $C$ large enough and thus obtain the desired $L^2$ convergence
$$\lim_{n\to\ii}\frac1{|B_{R_n}|}\int_{\R^d}|\chi_nF|^2=\lim_{n\to\ii}\frac1{|B_{R_n}|}\int_{B(\tau_n,R_n)}|F|^2=0$$
as was claimed.
\end{proof}

\subsection{End of the proof of Theorem~\ref{thm:phase_transitions}}\label{sec:fluid-solid transition}

Now we are ready to conclude the proof of Theorem~\ref{thm:phase_transitions}. Let $u$ be any infinite ground state.

We start with (iv) and note that over any ball $B$ we have
$$\left||u|^2-\frac1{|B|}\int_B|u|^2\right|\leq \max_B|u|^2-\min_B|u|^2$$
so that the estimate ~\eqref{eq:variance} implies immediately the stated bound~\eqref{eq:min_max_periodic} in (iv) of Theorem~\ref{thm:phase_transitions}.

It remains to prove (iii). Assume that $\mu<\mu_c$, or that $\mu=\mu_c$ but $f$ is $C^1$ at $\mu_c$. Let $G=(|u|^2-\rho)^2 + |\nabla u|^2$. From Lemma~\ref{lem:variance_kinetic} we have
\begin{equation}
\lim_{R\to\ii}\sup_{\tau\in\R^d}\frac1{|B_R|}\int_{B(\tau,R)} G=0.
\label{eq:variance_zero_thermo}
\end{equation}
We use the following simple lemma, whose proof is given at the end of the argument.


\begin{lemma}\label{lem:cnst_in_one_ball}
Let $G\in W^{1,\ii}(\R^d)$ satisfy
\begin{equation}
\lim_{R\to\ii}\frac1{|B_R|}\int_{B(0,R)}|G|=0.
\label{eq:variance_zero_thermo2}
\end{equation}
Then, for almost every $\omega\in\bS^{d-1}$, we can find a sequence $s_n\to\ii$ such that
$$\norm{G}_{L^\ii(B(s_n\omega,n))}\to0.$$
\end{lemma}

This gives $\norm{|u|^2-\rho}_{L^\ii(B(s_n \omega, n))}\to0$ and $\norm{\nabla u}_{L^\ii(B(s_n \omega, n))} \to0$ for almost every $\omega\in \bS^{d-1}$ and $s_n\to\ii$ as in the lemma (depending on $\omega$). After extraction a subsequence we can further assume that $u(\cdot+s_n \omega)\to u_0$ in $L^\ii(B(0,r))$ for any fixed $r$, where $|u_0|^2=\rho$. Since $u_0$ is then an infinite ground state, it must solve the GP equation and since $\mu=\rho\int_{\R^d}w$, we obtain $\Delta u_0=0$. Therefore $u_0=e^{i\theta_0}\sqrt\rho$ for some constant $\theta_0\in\R$ and we have shown that the constant function $u_{\rm cnst}=\sqrt\rho$ is an infinite ground state for all $\mu<\mu_c$. Passing to the limit $\mu\to\mu_c^-$ in the characterization~\eqref{eq:const-GS-condition} of an infinite ground state, we deduce that $u_{\rm cnst}$ is also an infinite ground state for $\mu=\mu_c$, in case $f$ was not $C^1$ at $\mu_c$.

Let us finally go back to our sequence $u(\cdot+s_n \omega)$, which converges to $e^{i\theta_0}u_{\rm cnst}$ on $B(s_n \omega,r)$ for any fixed $r>0$. Since $\eps_n:=\norm{\nabla u}_{L^\ii(B(s_n\omega,n))}$ tends to $0$, we can take $r_n=\min(n,\eps_n^{-1/2})\to\ii$ to conclude that
$$\|u- e^{i\theta_0}u_{\rm cnst}\|_{L^\infty(s_n \omega,r_n)}\to 0.$$
This concludes the proof of Theorem~\ref{thm:phase_transitions}.\qed

\medskip
Let us finally provide the

\begin{proof}[Proof of Lemma \ref{lem:cnst_in_one_ball}] For every fixed $r\in \mathbb{N}$, since $\1_{B_R}\ast\1_{B_r}\leq |B_r|\1_{B_{R+r}}$, we can replace $|G|$ by $|G|\ast\1_{B_r}$ in~\eqref{eq:variance_zero_thermo2}. We obtain
\begin{align*}
0&=\lim_{R\to\ii}\frac1{|B_R|}\int_{B_R}|G|\ast\1_{B_r}\\
&=\lim_{R\to\ii}\frac1{|B_R|}\int_{0}^R\int_{\bS^{d-1}}\left(\int_{B(s\omega,r)}|G|\right)\rd\omega\,s^{d-1}\rd s\\
&\geq\lim_{R\to\ii}\frac1{|B_R|}\int_{R/2}^R\int_{\bS^{d-1}}\left(\int_{B(s\omega,r)}|G|\right)\rd\omega\,s^{d-1}\rd s\\
&\geq\lim_{R\to\ii}\frac{(R/2)^d}{|B_R|}\int_{\bS^{d-1}}\left(\inf_{s\geq R/2}\int_{B(s\omega,r)}|G|\right)\rd\omega\\
&=\frac{2^{-d}}{|\bS^{d-1}|}\int_{\bS^{d-1}}\left(\liminf_{s\to\ii}\int_{B(s\omega,r)}|G|\right)\rd\omega.
\end{align*}
The last line is by monotone convergence. Consequently, for almost every $\omega\in\bS^{d-1}$,
\begin{equation}\label{eq:F2-L1-CV}
\liminf_{s\to\ii}\int_{B(s\omega,r)}|G|=0,\qquad\forall r\in\N.
\end{equation}
We now prove that this implies uniform convergence. Let $M:=\|G\|_{L^\ii(\R^d)}+\|\nabla G\|_{L^\ii(\R^d)}$. For any $x\in B(s\omega,r)$, we have
$|G(y)| \ge |G(x)|/2$ for all $y$ in the ball $B(x,r(x))$ of radius $r(x):= |G(x)|/(2M)\le 1/2 \le r$. In addition, we have
$$\Big|B(s\omega,r) \cap B\big(x,r(x)\big)\Big|\ge C_d^{-1}r(x)^d=\frac{|G(x)|^d}{C_d(2M)^d}$$
for some universal constant $C_d$. We have thus proved the inequality
\begin{equation}\label{eq:GN-type-inequality}
\|G\|^{d+1}_{L^\infty(B(s\omega,r))} \le C_d2^{d+1}M^d  \| G \|_{L^1(B(s\omega,r))}.
\end{equation}
From \eqref{eq:F2-L1-CV} and  \eqref{eq:GN-type-inequality},  we obtain
\begin{equation}\label{eq:F2-Lii-CV}
\liminf_{s\to\ii} \| G \|_{L^{\infty} ( B(s\omega,r))}=0,\qquad\forall r\in\N.
\end{equation}
We can thus choose a sequence $s_n\to \ii$ such that $\norm{G}_{L^\ii(B(s_n\omega,n))}\to0.$
\end{proof}

\subsection{Improved lower bounds of $\mu_c$: Proof of Proposition~\ref{prop:improved_mu_c}}\label{sec:improved_mu_c}

In this subsection we discuss the improved lower bound of $\mu_c$ under additional conditions on $w$. We will mimic the proof of  Theorem~\ref{thm:phase_transitions} (ii), with suitable modifications.

\subsubsection*{Case where $s>d+2$.} We assume first that $w$ decays fast enough but do not put any assumption on the sign on $w$. To prove \eqref{eq:const-GS-condition}, we proceed as in \eqref{eq:const-GS-condition-interaction} and ignore the diagonal parts which are non-negative. For the off-diagonal part in \eqref{eq:const-GS-condition-interaction}, there is no need to rearrange the linear and quadratic contributions of $h_+$. Arguing as in \eqref{eq:const-GS-condition-off-diagonal-g1} and \eqref{eq:const-GS-condition-off-diagonal-g22}, and  using also \eqref{eq:Taylor-expansion-fourier}, we have
\begin{multline}\label{eq:const-GS-condition-off-diagonal-special-1}
\left| 2u \iint f_-(x)  h_+(y)  w(x-y)  \dx\, \dy \right| \\
 =  2u (2\pi)^{\frac{d}2} \left| \int \overline{\hat f_-(k)} \hat h_+(k) ( \hat w(0) - \hat w(k) ) \rd k \right| \le u \beta  \int |k|^2 |\hat f_-| \int |\hat h_+|^2
\end{multline}
with
$$
\beta=   \max_{|\omega|=1} \int |\omega\cdot z|^2 |w(z)| {\rm d}z  .
$$
When $w$ is radial the right side is independent of $\omega$ and thus $\beta=d^{-1}\int |z|^2 |w(z)| {\rm d}z$.
For the part involving $h_+^2$, we use the fact that $f_-$ and $h_+$ have disjoint support and the uniform bound $|f_-|\le u^2$ to write
\begin{align*}
\left| \iint f_-(x)  h_+^2(y)  w(x-y) \dx\, \dy \right|
&=  \left| \iint f_-(x)  (h_+(x)-h_+(y))^2   w(x-y)  \dx\, \dy \right| \nn\\
&\le u^2 \iint |w(z)|  |h_+(x)-h_+(x-z)|^2 \dx\, {\rm d}z.
\end{align*}
By Plancherel's theorem and the bound $|1-e^{it}| \le |t|$, we have
$$
 \int_{\R^d} |h_+(x)-h_+(x-z)|^2 \dx
\le  \int_{\R^d} |\hat h_+(k)|^2 |k\cdot z|^2  {\rm d} k.$$
Therefore,
\begin{align}\label{eq:const-GS-condition-off-diagonal-special-2}
\left| \iint f_-(x)  h_+^2(y)  w(x-y)  \dx\, \dy \right| &\le  u^2 \iint  |\hat h_+(k)|^2 |k\cdot z|^2 |w(z)| {\rm d}z\, \rd k \nn\\
&\le  u^2 \beta \int |k|^2 |\hat h_+(k)|^2 {\rm d}k.
\end{align}
For the kinetic term, using \eqref{eq:const-GS-condition-kinetic-1} we have
\begin{align} \label{eq:const-GS-condition-kinetic-special-kinetic}
\int_{\R^d} |\nabla h|^2 \ge  \int |k|^2 |\hat h_+(k)|^2 \rd k  +  \frac{1}{4u^2} \int |k|^2 |\hat f_-(k)|^2 \rd k.
\end{align}
In summary, collecting \eqref{eq:const-GS-condition-off-diagonal-special-1}, \eqref{eq:const-GS-condition-off-diagonal-special-2} and \eqref{eq:const-GS-condition-kinetic-special-kinetic} we conclude that
\begin{align}\label{eq:const-GS-condition-final-special-1}
&\int_{\R^d} |\nabla h|^2 + \frac{1}{2} \iint (2u h(x) + h^2(x)) (2u h(y) + h^2(y)) w(x-y) \dx\, \dy \nn\\
&\ge   \int  |k|^2 \left( (1-u^2\beta)  |\hat h_+(k)|^2  +  \frac{1}{4u^2}   |\hat f_-(k)|^2  - u\beta |\hat f_-(k)||\hat h_+(k)| \right) {\rm d}k  .
\end{align}
The integrand of  the right-hand side of \eqref{eq:const-GS-condition-final-special-1} is non-negative if
$$
\frac{1-u^2\beta}{u^2} \ge (u \beta)^2
 $$
which is verified when $u^2\beta \le (\sqrt{5}-1)/2$, namely
 \begin{align}\label{eq:const-GS-condition-final-special-2}
 \mu =  u^2 \int_{\R^d} w \le  \frac{\sqrt{5}-1}{2 \beta} \int_{\R^d} w.
\end{align}
Thus if $s>d+2$, then the constant function $u$ is an infinite ground state for all $\mu$ as in~\eqref{eq:const-GS-condition-final-special-2} and we obtain the stated inequality~\eqref{eq:mc-mu1-off-factor}.

\subsubsection*{Case where $w\geq0$.}  Next we assume that $w\geq0$ (but we do not assume $s>d+2$). Using again \eqref{eq:const-GS-condition-interaction}, we can estimate the diagonal parts as
\begin{equation}\label{eq:const-GS-condition-positivew-1}
\frac{1}{2} \iint f_-(x) f_-(y) w(x-y) \dx\, \dy =  \frac{(2\pi)^{d/2}}{2} \int |\hat f_-(k)|^2 \hat w(k) {\rm d}k
\end{equation}
and
\begin{align}\label{eq:const-GS-condition-positivew-2}
\frac{1}{2}  \iint f_+(x) f_+(y) w(x-y) \dx\, \dy &\ge 2u^2 \iint h_+(x) h_+(y) w(x-y) \dx\, \dy  \nn\\
&=  2u^2 (2\pi)^{d/2} \int |\hat h_+(k)|^2 \hat w(k) {\rm d}k.
\end{align}
 In the latter bound, we used $f_+\ge 2u h_+$ and the assumption $w\ge 0$. For the off-diagonal term, we take $t\in (0,1)$ and estimate using the orthogonality
\begin{align}\label{eq:const-GS-condition-positivew-3}
&\left| 2u \iint f_-(x)  h_+(y)  w(x-y)  \dx\, \dy \right| \nn\\
&\qquad  = 2u (2\pi)^{d/2}   \left|  \int \overline{\hat f_-(k)} \hat h_+(k) ( \hat w(k) - t \hat w(0) )  {\rm d}k \right| \nn\\
&\qquad\le 2u (2\pi)^{d/2}   \int |\hat f_-(k) | |\hat h_+(k)|  |\hat w(k) -  t \hat w(0)|  {\rm d}k
\end{align}
and
\begin{align}\label{eq:const-GS-condition-positivew-4}
\left| \iint f_-(x)  h_+^2(y)  w(x-y)  \dx\, \dy \right|  \le  u^2 (2\pi)^{d/2} \hat w(0) \int |\hat h_+(k)|^2 {\rm d}k.
\end{align}
Here we used $|f_-|\le u^2$ and $\|w\|_{L^1}=(2\pi)^{d/2} \hat w(0)$ since $w\ge 0$. Putting \eqref{eq:const-GS-condition-positivew-1}, \eqref{eq:const-GS-condition-positivew-2}, \eqref{eq:const-GS-condition-positivew-3}, \eqref{eq:const-GS-condition-positivew-4} and the kinetic bound \eqref{eq:const-GS-condition-kinetic-special-kinetic} together, we arrive at
\begin{align}\label{eq:const-GS-condition-final-special-1aa}
&\int_{\R^d} |\nabla h|^2 + \frac{1}{2} \iint (2u h(x) + h^2(x)) (2u h(y) + h^2(y)) w(x-y) \dx\, \dy \nn\\
&\qquad\ge  \int   \Bigg\{  \frac{1}{4u^2} \Big( |k|^2  + 2u^2 (2\pi)^{d/2} \hat w(k) \Big)  |\hat f_-(k)|^2 \nn\\
&\qquad \qquad + \Big(|k|^2 + u^2 (2\pi)^{d/2} (2\hat w(k) -\hat w(0) )  \Big) |\hat h_+(k)|^2 \nn\\
&\qquad\qquad  - 2u (2\pi)^{d/2} |\hat f_-(k)| |\hat h_+(k)| |t\hat w(0)-\hat w(k)|  \Big\} {\rm d}k.
\end{align}
When
$$ (2\pi)^{d/2}\hat w(0) u^2 = \mu \le \mu_2=\inf_{k\in \R^d} \frac{|k|^2\hat w(0)}{2(\hat w(0) - 2 \hat w(k))_+},$$
we have
$$
|k|^2  + 2u^2 (2\pi)^{d/2} \hat w(k) \ge 2 u^2 (2\pi)^{d/2} \Big( (\hat w(0) - 2 \hat w(k))_+ + \hat w(k) \Big) \ge 0
$$
since $ \frac{1}{2} (\hat w(0) - 2 \hat w(k))_+ + \hat w(k) \ge \frac 1 2  \hat w(0)\ge 0$, and
$$
|k|^2  + 2u^2 (2\pi)^{d/2} \hat w(k) -u^2 \ge u^2 (2\pi)^{d/2} \Big( 2(\hat w(0) - 2 \hat w(k))_+ + 2\hat w(k) -\hat w(0) \Big).
$$
Therefore, when $t=1/2$, the integrand on the right hand side of \eqref{eq:const-GS-condition-final-special-1aa} is non-negative since
\begin{align}\label{eq:const-GS-condition-final-special-2bb}
\Big( (\hat w(0) - 2 \hat w(k))_+ + \hat w(k) \Big) \Big( 2(\hat w(0) - 2 \hat w(k))_+&  + 2\hat w(k) -\hat w(0) \Big)  \nn\\
&\ge \frac{1}{2} \Big( \hat w(0)-2\hat w(k) \Big)^2.
\end{align}
To verify \eqref{eq:const-GS-condition-final-special-2bb}, we may distinguish two cases. If $\hat w(0)-2\hat w(k)\ge 0$, then
\begin{align*}
&\Big( (\hat w(0) - 2 \hat w(k))_+ + \hat w(k) \Big) \times  \Big( 2(\hat w(0) -2 \hat w(k))_+ + 2\hat w(k) -\hat w(0) \Big)\nn\\
&\ge \frac{1}{2}  (\hat w(0) - 2 \hat w(k))_+ \times (\hat w(0) -2 \hat w(k))_+ = \frac{1}{2} \Big( \hat w(0)-2\hat w(k) \Big)^2.
\end{align*}
If $\hat w(0)-2\hat w(k)\le 0$, then $\hat w(k) \ge 2\hat w(k)-\hat w (0) \ge 0$ since $\hat w(0)\ge \hat w(k)$, and hence
\begin{multline*}
\Big( (\hat w(0) - 2 \hat w(k))_+ + \hat w(k) \Big) \Big( 2(\hat w(0) - 2 \hat w(k))_+  + 2\hat w(k) -\hat w(0) \Big) \\
= \hat w(k)  \Big( 2\hat w(k) -\hat w(0) \Big) \ge \Big( \hat w(0)-2\hat w(k) \Big)^2.
\end{multline*}
Thus if $w\ge 0$, then the constant function $u$ is an infinite ground state if $\mu\le \mu_2$. The proof of Proposition~\ref{prop:improved_mu_c} is complete.\qed

\section{Low and high density limits}

\subsection{Uniqueness of real-valued ground states for Ginzburg-Landau}
As a preparation for the proof of Theorem \ref{thm:phase_transitions-low-density}, we start with the special case of the delta interaction. The following is well known in dimensions $d\in\{1,2\}$ but we have not found it explicitly written in all space dimensions.

\begin{lemma}[Uniqueness of real-valued ground states for Ginzburg-Landau]\label{lem:uniqueness_GinzLan}
Let $d\geq1$, $w=\delta$ and $\mu=1$. The only real-valued infinite ground states are constants: $u(x)=\pm1$.
\end{lemma}

Recall that the Ginzburg-Landau equation admits non constant real-valued solutions, for instance the 1D kink repeated in one direction. Lemma~\ref{lem:uniqueness_GinzLan} thus says that those are unstable with respect to the variations of the complex phase, in all dimensions. We refer to ~\cite{YeZho-96,Mironescu-96,Sandier-98,ComMir-99,PacRiv-00,MilPis-10,FarMir-13,IgnNguSlaZar-20} for various uniqueness results concerning the Ginzburg-Landau equation.

\begin{proof}
We have $-\Delta u+(u^2-1)u=0$
as well as the local stability condition
\begin{equation}
 -\Delta +(u^2-1)\geq 0
 \label{eq:GL_stability}
\end{equation}
in the sense of quadratic forms, by Lemma~\ref{lem:1st_2nd_order}. This is the only specific property of (real-valued) ground states we will use in the proof.
We will also use a pointwise bound of Modica~\cite{Modica-85} which states that for real-valued solutions
\begin{equation}
2|\nabla u|^2\leq (1-u^2)^2.
\label{eq:Modica}
\end{equation}
We then let $f:=1-u^2$ and aim at proving that $f\equiv0$. The function $f$ satisfies the equation
\begin{align*}
\Delta f=-\Delta u^2=-2|\nabla u|^2-2u\Delta u&=-2|\nabla u|^2+2u^2f\\
&=-2|\nabla u|^2+2f-2f^2.
\end{align*}
Using that $\Delta f\leq 0$ on the set $\{f<0\}$ one can see that necessarily $f>0$, that is, $|u|< 1$~\cite{BreMerRiv-94,HerHer-96,Farina-98}.
Using~\eqref{eq:Modica} to bound $|\nabla u|^2$, we have
\begin{equation}
\Delta f=-2|\nabla u|^2+2f-2f^2\geq 2f-3f^2.
\label{eq:GL_f}
\end{equation}

Let $\chi_R(x)=\chi(x/R)$ where $\chi\in C^\ii_c(\R^d,[0,1])$ with $\chi\equiv1$ on $B_1$ and $\chi\equiv0$ on $\R^d\setminus B_2$.
First we bound $2f-3f^2\geq f-Cf^3$ in~\eqref{eq:GL_f} and integrate against $\chi_R$ to obtain
\begin{equation}
\int_{B_R}f\leq \int f\Delta \chi_R +C\int\chi_Rf^3\leq \frac{C}{R^2}\int_{B_{2R}}f+C\int_{B_{2R}}f^3.
\label{eq:f_1_3}
\end{equation}
Next we rewrite~\eqref{eq:GL_f} in the form
$$-\Delta f-f^2+2fu^2\leq0$$
and integrate against $\chi_R^2f$ to obtain
$$\pscal{\chi_Rf(-\Delta-f)\chi_R f}+2\int\chi_R^2u^2f^2\leq \int|\nabla\chi_R|^2f^2\leq \frac{C}{R^2}\int_{B_{2R}}f^2$$
From the stability condition~\eqref{eq:GL_stability}, the first term is non-negative and hence we obtain
$$2\int\chi_R^2u^2f^2\leq \frac{C}{R^2}\int_{B_{2R}}f^2.$$
Noticing that $|\nabla f|^2=4u^2|\nabla u|^2\leq 2u^2f^2$, this provides
\begin{equation}
\int\chi_R^2|\nabla f|^2\leq \frac{C}{R^2}\int_{B_{2R}}f^2.
 \label{eq:gradient}
\end{equation}
However, the non-negativity of $-\Delta-f$ implies
$$\int\chi_R^2 f^3\leq \int|\nabla\chi_Rf|^2\leq \frac{C}{R^2}\int_{B_{2R}}f^2+2\int\chi_R^2|\nabla f|^2.$$
Inserting~\eqref{eq:gradient} we arrive at
\begin{equation}
\int_{B_R}f^3\leq \frac{C}{R^2}\int_{B_{2R}}f^2\leq \frac{C}{R^2}\int_{B_{2R}}f
 \label{eq:f_3_1}
\end{equation}
using that $f\leq1$. We can now insert~\eqref{eq:f_3_1} into~\eqref{eq:f_1_3}, which provides
\begin{equation}
 \int_{B_R}f \leq \frac{C}{R^2}\int_{B_{4R}}f.
 \label{eq:bootstrap}
\end{equation}
We can now easily bootstrap to deduce $f\equiv0$. In fact, iterating~\eqref{eq:bootstrap} $d$ times, we obtain
$$\int_{B_R}f \leq \frac{C^d}{R^{2d}}\int_{B_{4^dR}}f\leq CR^{-d}$$
since $f\leq1$. Taking $R\to\ii$, we conclude that $\int_{B_{R_0}}f =0$ for any given $R_0$ and therefore $f\equiv0$.
\end{proof}

\subsection{Low-density: Proof of Theorem \ref{thm:phase_transitions-low-density}}\label{sec:proof_low_density}

We study the limit $\mu\to0$ and prove convergence to the case of a delta interaction, after an appropriate rescaling. Let $u_\mu$ be any real-value infinite ground state for the chemical potential $\mu\ll1$ and write
$$u_\mu(x)=\left(\frac{\mu}{\int_{\R^d} w}\right)^{\frac12} \tilde u_\mu\big(\sqrt\mu x\big).$$
The GP equation \eqref{eq:GP_local} can be rewritten as
\begin{equation}
-\Delta \tilde u_\mu+\big(\tilde w_\mu\ast\tilde u_\mu^2-1\big)\tilde u_\mu=0
\label{eq:rescaled_GP}
\end{equation}
with the rescaled potential
$$\tilde w_\mu(x)=\frac{\mu^{-\frac{d}2}w(x/\sqrt\mu)}{\int_{\R^d}w}\wto\delta_0.$$
After scaling the test function $h$ properly in~\eqref{eq:GP_local}, we see that $\tilde u_\mu$ is an infinite ground state for the chemical potential $\mu=1$ and with the interaction $\tilde w_\mu$.

Next we prove that $\tilde u_\mu$ converges locally uniformly to a Ginzburg-Landau ground state for the $\delta$ interaction. By~\eqref{eq:pointwise_bound_main} and~\eqref{eq:pointwise_bound_derivatives_main} in Theorem~\ref{thm:uniform_bound} we have $\|\tilde u_\mu\|_{L^\ii}+\|\nabla \tilde u_\mu\|_{L^\ii}\leq C$ for some constant $C$ and $\mu$ small enough. By~\eqref{eq:lower_bound} in the same theorem, we have
\begin{align} \label{eq:lower-bound-tumu}
\int_{B(\tau,R)}|\tilde u_\mu|^2=\mu^{\frac{d}2-1}\int_{\R^d}w\int_{B(\tau,R/\sqrt\mu)}|u_\mu|^2\geq \left(C^{-1}-CR^{-2}\right)_+.
\end{align}
The right side is strictly positive for $R$ large enough, which proves that $u_\mu$ must have mass everywhere in space, uniformly with respect to $\mu$. Since $\tilde u_\mu$ is bounded and uniformly Lipschitz, after extraction of a subsequence, we can assume that $\tilde u_\mu\to u$ locally uniformly. The limit $u$ is nontrivial due to \eqref{eq:lower-bound-tumu}, and it solves the Ginzburg-Landau equation
\begin{equation}
-\Delta u+(u^2-1)u=0.
\label{eq:Ginzburg-Landau}
\end{equation}
We can also pass to the limit in the definition of infinite ground states with a fixed $h$, and obtain that $u$ is a ground state for the delta interaction. From Lemma~\ref{lem:uniqueness_GinzLan} we deduce that $u\equiv\pm1$ and thus $\tilde u_\mu\to\pm1$ locally uniformly. Since $\nabla \tilde u_\mu$ is bounded, one sees that the convergence must be uniform over the whole of $\R^d$ and not just local (if for instance $u_\mu(0)\to1$ and $u_\mu(x_\mu)\to -1$ then there must exist $y_\mu\in[0,x_\mu]$ such that $u_\mu(y_\mu)=0$, which cannot be). Upon changing $\tilde u_\mu$ into $-\tilde u_\mu$ we can thus assume
\begin{equation}
\lim_{\mu\to0}\norm{\tilde u_\mu-1}_{L^\ii}=0.
\label{eq:limit_to_cnst}
\end{equation}
In other words, $u_\mu=(\mu/\int_{\R^d}w)^{1/2}+o(\sqrt\mu)_{L^\ii}$.

Our last step is to prove that necessarily $\tilde u_\mu\equiv1$ by a perturbation argument around the constant function.  We let $v_\mu:=\tilde u_\mu-1$ so that
$$-\Delta v_\mu+\Big(\tilde w_\mu\ast(2v_\mu+v_\mu^2\big)\Big)(1+v_\mu)=0,$$
which we can also rewrite in the form
\begin{align}\label{eq:inverted_GP}
v_\mu &=(-\Delta+2)^{-1}\bigg(2(\tilde w_\mu-\delta_0)\ast v_\mu-(\tilde w_\mu\ast v_\mu^2) (1+v_\mu)-2(\tilde w_\mu\ast v_\mu)v_\mu\bigg) \nn\\
&= (2\pi)^{-\frac{d}2}\;Y_2 \ast \bigg(2(\tilde w_\mu-\delta_0)\ast v_\mu-(\tilde w_\mu\ast v_\mu^2) (1+v_\mu)-2(\tilde w_\mu\ast v_\mu)v_\mu\bigg).
\end{align}
Recall that $Y_M$ is the Yukawa potential, that is, $\widehat{Y_M}(k)=(|k|^2+M^2)^{-1}$.
Using $Y_2\in L^1(\R^d)$, this implies
\begin{equation}
\|v_\mu\|_{L^\ii}\leq C\norm{(\tilde w_\mu-\delta_0)\ast v_\mu}_{L^\ii}+C\|v_\mu\|_{L^\ii}^2.
\label{eq:fixed_point}
\end{equation}
We write the first term as
\begin{multline*}
(\tilde w_\mu-\delta_0)\ast v_\mu(x)=\int_{|x-y|\leq R}\tilde w_\mu(x-y)\big(v_\mu(y)-v_\mu(x)\big)\,\dy\\
+\int_{|x-y|\geq R}\tilde w_\mu(x-y)v_\mu(y)\,\dy-v_\mu(x)\int_{|y|\geq R}\tilde w_\mu
\end{multline*}
and obtain
\begin{align*}
\norm{(\tilde w_\mu-\delta_0)\ast v_\mu}_{L^\ii}&\leq CR\norm{\nabla v_\mu}_{L^\ii}
+2\frac{\int_{|y|\geq R/\sqrt\mu}|w(y)|\,\dy}{\int_{\R^d} w}\norm{v_\mu}_{L^\ii}\\
&\leq CR\norm{\nabla v_\mu}_{L^\ii}
+\frac{C\mu^{\frac{s-d}{2}}}{R^{s-d}}\norm{v_\mu}_{L^\ii}.
\end{align*}
Taking the gradient of Equation~\eqref{eq:inverted_GP} and using that $\nabla Y_2\in L^1(\R^d)$, we obtain by a similar reasoning
$$\norm{\nabla v_\mu}_{L^\ii}\leq C\norm{v_\mu}_{L^\ii}.$$
Inserting in~\eqref{eq:fixed_point} we have proved that
$$\|v_\mu\|_{L^\ii}\leq C\left(R+\frac{\mu^{\frac{s-d}{2}}}{R^{s-d}}+\|v_\mu\|_{L^\ii}\right)\|v_\mu\|_{L^\ii}.$$
This suggests to take $R=\mu^{1/4}\to0$. This gives $v_\mu\equiv0$ when $\mu$ is small enough. This concludes the proof that the constant is the unique ground state for sufficiently small values of $\mu$. This concludes the proof of  Theorem \ref{thm:phase_transitions-low-density}. \qed

\subsection{High-density: Proof of Theorem~\ref{thm:high_density}}\label{sec:proof_high_density}
Let $\mu_n\to\ii$ and $u_n$ be any positive infinite ground state. We introduce $v_n:=u_n/\sqrt{\mu_n}$, which satisfies the local bounds
\begin{equation}
\int_{B(z,\ell)}|v_n|^2\leq C,\qquad \int_{B(z,\ell)}|\nabla v_n|^2\leq C\mu_n,
 \label{eq:local_bd_high_density}
\end{equation}
by Theorem~\ref{thm:local_bound}, for any fixed (large enough) $\ell$ and any $z\in\R^d$. Upon extraction of a subsequence, we can thus assume that $|v_n|^2\wto\nu$ locally in the sense of measures. The previous uniform local bound and the continuity of $w$ imply by an adaptation of Lemma~\ref{lem:estim_potential} that
$$|v_n|^2\ast w\to \nu\ast w\qquad\text{locally uniformly.}$$
Since $u_n$ is real-valued, we have the operator inequality $-\Delta+|u_n|^2\ast w-\mu_n\geq0$,  by Lemma~\ref{lem:1st_2nd_order}. Taking the scalar product against some $h\in H^1(\R^d)$ and dividing by $\mu_n$, we find
$$\frac1{\mu_n}\int_{\R^d}|\nabla h|^2+\int_{\R^d}|h|^2\big(|v_n|^2\ast w-1\big)\geq0.$$
Passing to the limit we obtain
\begin{equation}
\nu\ast w\geq1\qquad\text{on $\R^d$,}
 \label{eq:inequality_nu_classical}
\end{equation}
where we recall that $\nu\ast w$ is a continuous function.

Next we prove that $\nu\ast w=1$ on the support of $\nu$ and get a better bound on the local kinetic energy. The GP equation can be written as
\begin{equation}
\left(-\frac\Delta{\mu_n}+|v_n|^2\ast w-1\right)v_n=0.
 \label{eq:GP_high_density}
\end{equation}
Taking the scalar product against $\chi v_n$ for some $0\leq\chi\in C^\ii_c$, we obtain
$$\frac1{\mu_n}\int_{\R^d} \chi|\nabla v_n|^2+\int_{\R^d}\big(v_n^2\ast w-1\big)\chi v_n^2=\frac1{2\mu_n}\int_{\R^d} v_n^2\Delta\chi.$$
The right side converges to 0 by~\eqref{eq:local_bd_high_density} since $\chi$ has compact support. On the other hand, we have from the local uniform convergence
$$\lim_{n\to\ii}\int_{\R^d}\big(v_n^2\ast w-1\big)\chi v_n^2=\int_{\R^d}\big(\nu\ast w-1\big)\chi \rd\nu\geq0$$
by~\eqref{eq:inequality_nu_classical}. Thus we conclude that
$$\lim_{n\to\ii}\frac1{\mu_n}\int_{\R^d} \chi|\nabla v_n|^2=\int_{\R^d}\big(\nu\ast w-1\big)\chi \rd\nu=0.$$
In particular we have $\nu\ast w=1$, $\nu$--almost everywhere, as claimed. We have also shown that $\nabla v_n/\sqrt{\mu_n}$ converges to 0 locally in $L^2$.

Now we are ready to show that $\nu$ is a classical infinite ground state. For any fixed $R>0$ we consider a fixed function $h\in C^\ii_c(B_{R})$. Let $\delta>0$ be such that $h$ is supported on $B_{R-\delta}$ and consider a function $\eta_\delta\in C^\ii(\R^d,[0,1])$ so that $\eta_\delta\equiv1$ outside of the ball $B_R$ and $\eta_\delta\equiv0$ in $B_{R-\delta}$. We finally plug the trial state $\eta_\delta u_n+\sqrt{\mu_n}h$ in the characterization~\eqref{eq:local_min} of an infinite ground states in $B_R$ and obtain
\begin{multline*}
\frac1{\mu_n}\int_{B_R}(|\nabla (\eta_\delta v_n)|^2+|\nabla h|^2-|\nabla v_n|^2)-\int_{\R^d}\big((\eta_\delta^2-1)v_n^2+h^2\big)\\
+\frac12D(h^2,h^2)-\frac12D\big((1-\eta_\delta^2)v_n^2,(1-\eta_\delta^2)v_n^2\big)+D\big(\eta_\delta^2v_n^2,h^2\big)\geq0.
\end{multline*}
Using that both $\eta$ and $w$ are continuous, as well as the previous convergence $\nabla v_n/\sqrt{\mu_n}\to0$ in $L^2(B_R)$, we can pass to the limit and obtain
\begin{multline*}
-\int_{\R^d}h^2-\int_{\R^d}(\eta_\delta^2-1)\rd\nu+\frac12D(h^2,h^2)\\
-\frac12\iint \big(1-\eta_\delta(x)^2\big)\big(1-\eta_\delta(y)^2\big)w(x-y)\,\rd\nu(x)\rd\nu(y)\\
+\int_{\R^d}\eta_\delta(x)^2h(y)^2w(x-y)\,\rd\nu(x)\,\dy\geq0.
\end{multline*}
Taking $\delta\to0$ and recalling that $B_R$ is the open ball, we get
\begin{multline*}
-\int_{\R^d}h^2+\nu(B_R)+\frac12D(h^2,h^2)\\
-\frac12\iint_{(B_R)^2}w(x-y)\,\rd\nu(x)\rd\nu(y)+\int_{\R^d\setminus B_R}h(y)^2w(x-y)\,\rd\nu(x)\,\dy\geq0.
\end{multline*}
After letting now $h^2$ approximate a finite measure in $B_R$, we obtain the characterization of a classical infinite ground state.\qed

\section{Breaking of $U(1)$ symmetry and the 2D vortex}

\subsection{Real-valued ground states: Proof of Theorem~\ref{thm:real-valued}}\label{sec:proof_real-valued}

In this section we prove Theorem~\ref{thm:real-valued}, which states that all infinite ground states are real-value and positive in 1D, whereas in 2D real-value infinite ground states are positive.

\subsubsection*{\textnormal{(i)} Infinite ground states are real-value and positive in 1D}
Let $u$ be an infinite ground state in dimension $d=1$. Let $R>0$. We write $u(-R)=|u(-R)|e^{i\theta_-}$ and $u(R)=|u(R)|e^{i\theta_+}$ for some $\theta_\pm\in[0,2\pi)$. Consider then the function
$$v(x)=\begin{cases}
u(x)&\text{for $|x|\geq2R$,}\\
|u(x)|&\text{for $|x|\leq R$,}\\
u(x)e^{i\theta(x)}&\text{for $R\leq |x|\leq2R$,}
\end{cases}$$
where
$$\theta(x):=\begin{cases}
\frac{x+2R}{R}\theta_-&\text{for $-2R\leq x\leq -R$,}\\
\frac{2R-x}{R}\theta_+&\text{for $R\leq x\leq 2R$.}
\end{cases}$$
Since $|v|=|u|$ everywhere, the infinite ground state property implies
\begin{align*}
0&\leq \int_{\R^d}|v'|^2-|u'|^2\\
&=\int_{-R}^R(|(|u|)'|^2-|u'|^2)+\int_{(-2R,-R)\cup(R,2R)}(|\theta'|^2|u|^2-2\Im \theta'\overline u u')\\
&\leq\int_{-R}^R(|(|u|)'|^2-|u'|^2)+\frac{C}{R}-\frac{2\theta_-}{R}\Im\left(\int_{-2R}^{-R}\overline u u'\right)+\frac{2\theta_+}{R}\Im\left(\int_{R}^{2R}\overline u u'\right).
\end{align*}
In the last estimate we have used that $u$ is bounded by~\eqref{eq:pointwise_bound_main} in Theorem~\ref{thm:uniform_bound}.
From~\eqref{eq:momentum} in Theorem~\ref{thm:prop_GS} we know that the average momentum per unit length vanishes,
$$\int_{R}^{2R}\overline u u'=\int_{-2R}^{-R}\overline u u'=o(R)_{R\to\ii}$$
and therefore obtain
$$\int_{-R}^R(|u'|^2-|(|u|)'|^2)\leq o(1)_{R\to\ii}.$$
Since the function in the parenthesis is non-negative, this implies that $|(|u|)'|=|u'|$ almost everywhere. The regularity of $u$ then gives $u=e^{i\phi}|u|$ for some constant $\phi$. This concludes the proof of (i) in Theorem~\ref{thm:real-valued}.\qed

\subsubsection*{\textnormal{(ii)} Real infinite ground states are positive in 2D}
Let $u$ be a \emph{real-value} infinite ground state in dimension $d=2$. Then we know that $-\Delta+V\geq0$, where $V=w\ast|u|^2-\mu\in L^\ii(\R^2)$, by~\eqref{eq:linearly_stable} in Lemma~\ref{lem:1st_2nd_order}. The following is well known and says that any bounded potential is critical in dimensions $d\in\{1,2\}$, in the sense of~\cite{Simon-81,Murata-86}. The result fails in higher dimensions. 

\begin{lemma}[Criticality]\label{lem:critical}
Let $u\in W^{2,\ii}(\R^d)$ and $V\in L^\ii(\R^d)$ be such that $(-\Delta +V)u=0$. If $d\in\{1,2\}$, there exists a sequence $u_n\in H^1(\R^d)$ of compact support such that $u_n\to u$ locally uniformly and $\pscal{u_n,(-\Delta+V)u_n}\to0$.
\end{lemma}

\begin{proof}[Proof of Lemma~\ref{lem:critical}]
We have for every $\chi\in C^\ii_c(\R^d)$
$$\pscal{\chi u,(-\Delta+V)\chi u}=\int_{\R^d}|\nabla \chi|^2|u|^2\leq \|u\|^2_{L^\ii}\int_{\R^d}|\nabla \chi|^2.$$
We are therefore looking for a sequence $\chi_n\in C^\ii_c(\R^d)$ with $\chi_n\to1$ locally, such that $\nabla\chi_n\to0$ in $L^2(\R^d)$. Such a $\chi_n$ cannot exist in dimension $d\geq3$. In dimension $d=1$, we can simply take $\chi_n(x)=\chi(x/n)$ with any fixed $0\leq \chi\leq 1$ such that $\chi\equiv1$ on $[-1,1]$ and $\chi\equiv0$ on $\R\setminus(-2,2)$. In dimension $d=2$, we take for instance
$$\chi_n(x)=\begin{cases}
1&\text{for $|x|\leq n$,}\\
\frac{\log\left(\frac{n+n^2}{|x|}\right)}{\log(1+n)}&\text{for $n\leq |x|\leq n+n^2$,}\\
0&\text{for $|x|\geq n+n^2$}
\end{cases}$$
and obtain $\int_{\R^2}|\nabla \chi_n|^2\dx=2\pi/\log(1+n)\to0$.
\end{proof}

For the sequence $u_n$ from Lemma~\ref{lem:critical} we obtain using $-\Delta+V\geq0$
$$0\leq \pscal{|u_n|,(-\Delta+V)|u_n|}\leq \pscal{u_n,(-\Delta+V)u_n}\to0$$
and therefore $\pscal{|u_n|,(-\Delta+V)|u_n|}\to0$. For every $h\in C^\ii_c(\R^d)$, this gives
\begin{align*}
0&\leq \pscal{h+|u_n|,(-\Delta+V)(h+|u_n|)}\\
&=\pscal{h(-\Delta+V)h}+2\Re\pscal{(-\Delta+V)h,|u_n|}+o(1).
\end{align*}
and hence, after passing to the limit we find
$$\pscal{h(-\Delta+V)h}+2\Re\pscal{(-\Delta+V)h,|u|}\geq0,\qquad \forall h\in C^\ii_c(\R^d).$$
This implies that $(-\Delta+V)|u|=0$, that is, $|u|$ is also a solution of the GP equation. Since $u$ is a ground states we have the lower bound~\eqref{eq:local_bound_mass} on the local mass and then Corollary~\ref{cor:lower_bound} implies $|u|\geq c$ for some $c>0$. In other words, $u$ does not vanish and hence has a constant sign. This concludes the proof of (ii) in Theorem~\ref{thm:real-valued}.\qed

\subsection{Uniqueness in 1D: Proof of Theorem~\ref{thm:uniqueness_1D}}\label{sec:proof_uniqueness_1D}
We assume that $\widehat{w}\geq0$. Let $u$ be an infinite ground state for some $\mu>0$. From Theorem~\ref{thm:real-valued} proved in the previous section, we can assume that $u>0$ everywhere. From Corollary~\ref{cor:cnst_e_f}, we have $\rho:=\rho_\pm(\mu)=\mu/\int_\R w$ and from Theorem~\ref{thm:prop_GS} we deduce that
$$\lim_{n\to\ii}\frac1n\int_{\tau_n}^{\tau_n+n}u^2=\rho$$
for any sequence $\tau_n\in\R$. This implies that we can find $x_n\to-\ii$ and $y_n\to+\ii$ such that $u(x_n)\to \sqrt\rho$ and $u(y_n)\to \sqrt\rho$. We let $I_n:=(x_n,y_n)$ and define the new function
$$c_n(x):=\begin{cases}
u(x)&\text{if $x\notin I_n$,}\\
u(x_n)+\big(\sqrt\rho-u(x_n)\big)(x-x_n)&\text{if $x_n\leq x\leq x_n+1$,}\\
u(y_n)+\big(\sqrt\rho-u(y_n)\big)(y_n-x)&\text{if $y_n-1\leq x\leq y_n$,}\\
\sqrt\rho&\text{if $x_n+1\leq x\leq y_n-1$,}
\end{cases}$$
that is, we replace $u$ by a constant over $[x_n+1,y_n-1]$ and interpolate linearly in the two transition regions. By definition of $x_n$ and $y_n$, we have $|c_n-\sqrt\rho|\to0$ and $c_n'\to0$ uniformly on $[x_n,y_n]$. The local ground state property implies that
\begin{multline}
\int_{I_n}|u'|^2+\frac12\iint_{I_n^2}(u(x)^2-\rho)(u(y)^2-\rho)w(x-y)\,\dx\,\dy\\
\leq-\int_{I_n}\int_{\R\setminus I_n}(u(x)^2-\rho)(u(y)^2-\rho)w(x-y)\,\dx\,\dy+o(1).
\label{eq:1D_replace_cnst}
\end{multline}
The error $o(1)$ is due to the two finite transitions regions $[x_n,x_n+1]$ and $[y_n-1,y_n]$ where $c_n$ is almost constant. It turns out that the right side of~\eqref{eq:1D_replace_cnst} is bounded whenever $w$ decays fast enough at infinity. Indeed, we have $u\in L^\ii$ by Theorem~\ref{thm:uniform_bound} and, from Assumption~\ref{ass:w},
\begin{align*}
&\int_{x_n}^{y_n}\int_{y_n}^\ii|w(x-y)|\,\dx\,\dy\\
&\qquad \leq\int_{y_n-\kappa}^{y_n}\int_{y_n}^\ii|w(x-y)|\,\dx\,\dy+\kappa\int_{x_n}^{y_n-\kappa}\left(\int_{y_n}^\ii \frac{\dy}{(y-x)^s}\right)\dx\\
&\qquad \leq\kappa \int_\R|w|+\frac{\kappa}{s-1}\int_{x_n}^{y_n-\kappa}\frac{\dx}{(y_n-x)^{s-1}}\leq \kappa \int_\R|w|+\frac{\kappa^{3-s}}{(s-1)(s-2)}.
\end{align*}
This is where we use that $s>d+1=2$. There is a similar estimate for $y\in(-\ii,x_n)$. Using that $\widehat{w}\geq0$, the inequality~\eqref{eq:1D_replace_cnst} thus provides $u'\in L^2(\R)$ and
\begin{equation}
0\leq \iint_{I_n^2}(u(x)^2-\rho)(u(y)^2-\rho)w(x-y)\,\dx\,\dy\leq C.
\label{eq:interaction_bd_1D}
\end{equation}
Let us introduce the function $f:=u^2-\rho$. Then we have $f'=2uu'\in L^2(\R)$ since $u$ is bounded and $u'\in L^2(\R)$. The Fourier transform of the Schwartz distribution $f$ satisfies
$$\int_\R|k|^2|\widehat{f}(k)|^2<\ii.$$
In particular, $\widehat{f}\in L^2((-\ii,-\eps)\cup(\eps,\ii))$ for any $\eps>0$ and $\widehat{f}$ can only be singular at the origin. On the other hand, passing to Fourier coordinates, we can rewrite~\eqref{eq:interaction_bd_1D} as
$$\int_{\R}\widehat{w}(k))|\widehat{f_n}(k)|^2\rd k\leq C,\qquad f_n:=f\1_{I_n}.$$
Since $\widehat{w}(0)>0$ and $\widehat{w}$ is continuous, this gives
$$\int_{-\eps}^\eps|\widehat{f_n}(k)|^2\,\rd k\leq C.$$
Given that $\widehat{f_n}=\widehat{f}\ast \widehat{\1_{I_n}}\wto \widehat{f}$, we conclude that $f$ belongs to $L^2(\R)$, hence in $H^1(\R)$. In particular, $f\to0$ at infinity, that is, $u\to\sqrt\rho$.

Using the fact that $f\in L^2(\R)$, we get from Young's inequality that the right side of~\eqref{eq:1D_replace_cnst} tends to 0:
\begin{multline*}
-\int_{I_n}\int_{\R\setminus I_n}(u(x)^2-\rho)(u(y)^2-\rho)w(x-y)\,\dx\,\dy\\
=-\int(f\1_{I_n})(f\1_{\R\setminus I_n})\ast w\leq \norm{f\1_{I_n}}_{L^2}\norm{f\1_{\R\setminus I_n}}_{L^2}\int_\R|w|\to0.
\end{multline*}
Therefore, passing to the limit in~\eqref{eq:1D_replace_cnst}, we end up with
$$\int_{\R}|u'|^2=D(f,f)=0.$$
This implies that $u\equiv\sqrt\rho$, as was claimed.\qed

\subsection{Existence of vortex solutions in 2D: Proof of Theorem~\ref{thm:exist_vortex}}\label{sec:proof_vortex}

In order to prove the existence of a vortex in 2D, we need first a result which states that solutions of the GP equation with a topological degree $D$ at infinity can never be infinite ground states when $|D|\geq 2$. This is the content of the following theorem, which is a generalization of several results from~\cite{BetBreHel-94,BreMerRiv-94,Sandier-98,Shafrir-94,Mironescu-96}.

\begin{theorem}[Topological degree of 2D infinite ground states]\label{thm:degree_one}
Let $w$ satisfying Assumption~\ref{ass:w} in dimension $d=2$ and $\mu>0$. Let $u$ be an infinite ground state over $\R^2$, satisfying
\begin{equation}
\rho_{\rm cnst}(1-\eps)\leq |u(x)|^2\leq\rho_{\rm cnst}(1+\eps),\qquad \forall |x|\geq R_0
\label{eq:u_not_vanish_infinity}
\end{equation}
for some $R_0$ and some $\eps<1/3$, where $\rho_{\rm cnst}:=\mu/\int_{\R^2}w$.
Then the degree at infinity of $u$
\begin{equation}
D:=\frac{1}{2\pi R_0}\int_{R_0\bS^1}\frac{\Im\big(\overline{u}\partial_\theta u\big)}{|u|^2}\in \Z
\label{eq:def_degree}
\end{equation}
satisfies
$$D\in\{0,-1,1\}.$$
If $u$ is an infinite ground state on the half plane $\R^2_+=\{x_1>0\}$ satisfying $u\equiv \sqrt{\rho_{\rm cnst}}$ at the boundary $\{x_1=0\}$, the same hold  for the extension $\tilde u$ with $\tilde u\equiv\sqrt{\rho_{\rm cnst}}$ on $\{x_1<0\}$, but this time we obtain $D=0$.
\end{theorem}

Note that the statement involves the density $\rho_{\rm cnst}=\mu/\int_{\R^2}w$ of the constant function, which could in principle be different from the thermodynamic density $\rho$ satisfying $\mu(\rho)=\mu$. The latter is also the density of $u$ by Theorem~\ref{thm:prop_GS}. We only know that the two coincide when $\widehat{w}\geq0$ or when $\mu<\mu_c$, by Theorem~\ref {thm:phase_transitions}. Assumption~\eqref{eq:u_not_vanish_infinity} means that $|u|$ is not too far from being constant, but it could still oscillate quite a bit, for instance be periodic. In most works on the Ginzburg-Landau model $w=\delta$, it is assumed (or proved) that $|u|^2\to\rho_{\rm cnst}=\rho$ at infinity. We do not need such a strong condition to infer that the degree is $0$, $-1$ or $1$.

\begin{remark}
If~\eqref{eq:u_not_vanish_infinity} holds for some $1/3\leq \eps<1$ then our proof still applies but it gives $|D|\leq\frac{1+\eps}{1-\eps}$ in the whole space and $|D|\leq\frac{1+\eps}{2(1-\eps)}$ in a half space.
\end{remark}

\begin{proof}[Proof of Theorem~\ref{thm:degree_one}]
We only provide the argument in the whole space and quickly discuss the half space in the end, following~\cite{Sandier-98}. To lighten the notation, in the whole proof we denote
$$\varrho:=\rho_{\rm cnst}.$$
In $\R^2$ the proof is inspired of~\cite{BreMerRiv-94,Sandier-98b} to obtain bounds on the kinetic energy and then~\cite{BetBreHel-94,Shafrir-94} to get a contradiction when $|D|\geq2$. We provide all the details for the convenience of the reader. We start with a simple lower bound on the full kinetic energy.

\begin{lemma}[Kinetic energy lower bound]\label{lem:estim_lower_full_kinetic}
Let $u\in C^1(\R^2)$ be so that $|u(x)|^2\geq \varrho(1-\eps)$ for all $R_0\leq |x|\leq R$ and  some $\eps<1$. Write $u=|u|e^{i(D\theta+\psi)}$ where $D$ is the topological degree~\eqref{eq:def_degree} of $u$ in this region, $\theta$ is the polar angle and $\psi$ is $C^1$ and single-valued. Then we have
\begin{multline}
\int_{B_R}|\nabla u(x)|^2\,\dx\geq \int_{B_R}\big|\nabla |u|(x)\big|^2\dx+\varrho(1-\eps)2\pi D^2\log\frac{R}{R_0}\\
+\varrho(1-\eps)\int_{B_R}|\nabla \psi(x)|^2\dx.
\label{eq:estim_lower_full_kinetic}
\end{multline}
\end{lemma}

\begin{proof}[Proof of Lemma~\ref{lem:estim_lower_full_kinetic}]
We follow~\cite[Thm.~3]{BreMerRiv-94}. We have with $\phi=D\theta+\psi$
\begin{equation}
|\nabla u|^2=|\nabla |u||^2+|u|^2|\nabla \phi|^2\geq |\nabla |u||^2+\varrho(1-\eps)|\nabla \phi|^2.
\label{eq:decompose_gradient_angle}
\end{equation}
On the other hand, we can compute
\begin{equation}
|\nabla \phi|^2=\left|D\frac{x^\perp}{|x|^2}+\nabla \psi\right|^2=\frac{D^2}{|x|^2}+2D\partial_\theta\psi+|\nabla \psi|^2
\label{eq:decompose_vortex}
\end{equation}
where $x^\perp=(x_2,-x_1)$ and $\partial_\theta\psi=x^\perp\cdot\nabla\psi$ denotes the angular derivative. When we integrate over the annulus $\{R_0\leq |x|\leq R\}$, the angular integral of $\partial_\theta\psi$ vanishes and we are left with
\begin{align}
&\int_{B_R\setminus B_{R_0}}|\nabla u|^2\nn\\
&\quad\geq  \int_{B_R\setminus B_{R_0}}|\nabla |u||^2+\varrho(1-\eps)D^2\int_{B_R\setminus B_{R_0}}\frac{\rd x}{|x|^2}+\varrho(1-\eps)\int_{B_R\setminus B_{R_0}}|\nabla \psi|^2\nn\\
&\quad =\int_{B_R\setminus B_{R_0}}|\nabla |u||^2+\varrho(1-\eps)2\pi D^2\log\frac{R}{R_0}+\varrho(1-\eps)\int_{B_R\setminus B_{R_0}}|\nabla \psi|^2\label{eq:estim_lower_with_psi}
\end{align}
as desired.
\end{proof}

Next we derive a preliminary (and far from optimal) upper bound, which is sufficient for our purposes.

\begin{lemma}[Kinetic energy upper bound]\label{lem:estim_upper_full_kinetic}
Let $u$ be an infinite ground state. Then we have
$$\int_{B_R}|\nabla u(x)|^2\,\dx\leq CR+\int_{B_R}\big|\nabla |u|(x)\big|^2\dx$$
for $R$ large enough. If $|u|^2\geq \varrho(1-\eps)$ for all $R_0\leq |x|\leq R$ and some $\eps<1$, we obtain with the same notation as in Lemma~\ref{lem:estim_lower_full_kinetic}
\begin{equation}
\int_{B_R\setminus B_{R_0}}|\nabla \psi(x)|^2\,\dx\leq \frac{C}{\varrho(1-\eps)}R.
\label{eq:estim_linear_psi}
\end{equation}
\end{lemma}

\begin{proof}[Proof of Lemma~\ref{lem:estim_upper_full_kinetic}]
Recall the characterization~\eqref{eq:min_u_v} of an infinite ground state $u$, in terms of the local energy $\cF_{\mu,B_R,u}$ in~\eqref{eq:local_energy_Omega} containing the potential generated by the outside. For simplicity, in this proof we denote it by
\begin{multline}
\cF_R(v):=\cF_{\mu,B_R,u}(v)
=\cE_{B_R}(v)-\mu\int_{B_R}|v|^2\\+\iint_{B_R\times (B_R)^c}|v(x)|^2|u(y)|^2w(x-y)\dx\,\dy.
\label{eq:energy_variation}
\end{multline}
As a trial state we take a function which equals $|u|$ over the ball $B_{R-\ell}$ and then smoothly interpolate to $u$ in the annulus $\{R-\ell\leq|x|\leq R\}$:
$$v(x)=\begin{cases}
|u(x)|&\text{for $|x|\leq R-\ell$,}\\
u(x)\frac{|x|-R+\ell}{\ell}+|u(x)|\frac{R-|x|}{\ell}&\text{for $R-\ell\leq |x|\leq R$.}
\end{cases}$$
Using that $w\in L^1(\R^2)$, $u,\nabla u,\nabla|u|\in L^\ii(\R^2)$, we see that the error terms can all be estimated by the volume of the transition region, which is of order $R\ell$:
$$0\leq \cF_R(v)-\cF_R(u)\leq \int_{B_R}\big|\nabla |u|\big|^2-\int_{B_R}|\nabla u|^2+CR\ell.$$
This is the claimed estimate if we take $\ell$ or order one. To obtain the bound on $|\nabla \psi|$ we use again~\eqref{eq:decompose_gradient_angle} and~\eqref{eq:decompose_vortex}.
\end{proof}

With the previous bounds on the kinetic energy, we are able to conclude the proof of Theorem~\ref{thm:degree_one}, arguing similarly as in~\cite{BetBreHel-94,Shafrir-94}. For $R$ large enough the idea is to define a new function $\tilde u$ which is a perturbation of $u$, obtained by manipulating the phase and perforating $|D|$ holes of finite size. The argument goes as follows. First, after replacing $u$ by $\overline{u}$ we assume without loss of generality that $D\geq0$. Then, since we want to prove that $D\in\{0,1\}$, we can argue by contradiction and assume $D\geq2$ in the whole proof. For $0<\alpha<1$ a fixed number to be chosen later, we take $D$ points $X_1,...,X_D$ in the annulus $B_{\alpha R /2}\setminus B_{\alpha R /4}$, located at a distance of order $\alpha R$ to each other. We then define
$$\tilde u(x)=|u(x)|\prod_{j=1}^d\eta(x-X_j)e^{i\tilde \phi(x)}=:v(x)e^{i\tilde\phi(x)},$$
where $\eta\in C^\ii(\R^2,[0,1])$ vanishes on $B_1$ and equals 1 on $\R^2\setminus B_2$. We choose the phase $\tilde\phi$ so that
$$\tilde \phi(x)=\begin{cases}
\phi(x)&\text{for $|x|\geq R$,}\\
D\theta+\theta_R+\Big(\psi(Rx/|x|)-\theta_R\Big)\frac{\log\frac{|x|}{\alpha R}}{\log \frac1{\alpha}}&\text{for $\alpha R \leq |x|\leq R$,}
\end{cases}$$
where $\theta_R:=(2\pi R)^{-1}\int_{\partial B_R}\psi$
is the average of $\psi$ over the sphere of radius $R$. In other words, we slowly turn off the regular part $\psi$ of the phase $\phi$ in the normal direction, without touching the angular variations on $\partial B_R$. We also do not change the topological degree $D$. For $|x|\leq \alpha R $ we choose any real-valued function  $\tilde \phi$ such that $\tilde\phi=D\theta+\theta_R$ on the sphere $\partial B_{\alpha R }$ and
\begin{equation}
\int_{B_{\alpha R }\setminus\cup_{j=1}^dB(X_j,1)}|\nabla\tilde\phi|^2\leq 2\pi D\log R+C,\qquad  \int_{B_{R_0}}|\nabla\tilde\phi|^2\leq C.
 \label{eq:trial_phase_BBH}
\end{equation}
The existence of such a function is proved in~\cite[Thm.~I.8]{BetBreHel-94}. In fact, we do not need to define $\tilde\phi$ in the balls $B(X_j,1)$ since $v\equiv0$ there. One function that works is the Dirichlet minimizer with the boundary condition $\tilde \phi=\theta_j+\theta_R$ at $\partial B(X_j,1)$, where $\theta_j$ is the local angle around the center $X_j$.

Next we write that $u$ is an infinite ground state and obtain
\begin{equation*}
\cF_R(u)\leq \cF_R(\tilde u)= \cF_R(v)+\int_{B_R}|v|^2|\nabla\tilde\phi|^2\leq  \cF_R(|u|)+\int_{B_R}|v|^2|\nabla\tilde\phi|^2+C
\end{equation*}
with $\cF_R$ as in~\eqref{eq:energy_variation}. The last bound is because we have perforated finitely many holes in $|u|$, which costs a finite amount of energy. Everywhere $C$ is a generic constant that may depend on $u$ and $D$ but not on $R$. Simplifying the interaction terms, we end up with
\begin{equation*}
\int_{B_R}|\nabla u|^2\leq  \int_{B _R}|\nabla |u||^2+\int_{B_R}|v|^2|\nabla\tilde\phi|^2+C.
\end{equation*}
Inserting our lower bound~\eqref{eq:estim_lower_full_kinetic} we obtain
\begin{equation}
\varrho(1-\eps)2\pi D^2\log R+\varrho(1-\eps)\int_{B_R\setminus B_{R_0}}|\nabla \psi|^2\leq\int_{B_R}|v|^2|\nabla\tilde\phi|^2+C
\label{eq:GS_perforate}
\end{equation}
where $C$ now also contains the finite quantity $\int_{B _{R_0}}|\nabla|u||^2$.

To estimate the angular kinetic energy in~\eqref{eq:GS_perforate}, we note that
$$|v|^2\leq |u|^2\1_{B_R\setminus \cup_{j=1}^dB(X_j,1)}\leq \varrho(1+\eps)\1_{B_R\setminus \cup_{j=1}^dB(X_j,1)}\qquad\text{on $B_R\setminus B_{R_0}$.}$$
Using~\eqref{eq:trial_phase_BBH} and $u\in L^\ii$ we obtain
\begin{align}
\int_{B_R}|v|^2|\nabla\tilde\phi|^2 &\leq (1+\eps)\varrho\int_{B_{\alpha R }\setminus\cup_{j=1}^dB(X_j,1)}|\nabla\tilde\phi|^2+2\|u\|_\ii^2\int_{B_{R_0}}|\nabla\tilde\phi|^2\nn\\
&\qquad +(1+\eps)\varrho\int_{B_R\setminus B_{\alpha R }}|\nabla\tilde\phi|^2\nn\\
&\leq 2\pi \varrho(1+\eps)D\log R +(1+\eps)\varrho\int_{B_R\setminus B_{\alpha R }}|\nabla\tilde\phi|^2+C.\label{eq:GS_perforate2}
\end{align}
Finally we compute the new angular kinetic energy cost in the annulus $\{\alpha R \leq|x|\leq R\}$. We let $t(x):=\frac{\log(|x|/\alpha R)}{\log (1/\alpha)}$ and recall that $\nu=x/|x|$. By Hölder's inequality we have
\begin{align*}
& \int_{B_R\setminus B_{\alpha R }}|\nabla\tilde\phi|^2\\
 &\qquad\leq D^2(1+\alpha^{-1})\int_{B_R\setminus B_{\alpha R }}\frac{\dx}{|x|^2}+(1+\alpha)\int_{B_R\setminus B_{\alpha R }}|\nabla t(\psi(R\nu)-\theta_R)|^2\\
& \qquad=2\pi D^2\frac{1+\alpha}{\alpha}\log (\alpha^{-1})+(1+\alpha)\int_{B_R\setminus B_{\alpha R }}|\nabla t|^2(\psi(R\nu)-\theta_R)^2\\
&\qquad \qquad +(1+\alpha)\int_{B_R\setminus B_{\alpha R }}t^2|\nabla \psi(R\nu)|^2\\
& \qquad\leq 2\pi D^2\frac{1+\alpha}{\alpha}\log (\alpha^{-1})+\frac{1+\alpha}{\log(\alpha^{-1})R}\int_{\partial B_R}(\psi(R\nu)-\theta_R)^2\\
&\qquad \qquad +\frac{(1+\alpha)(1-\alpha^2)}{2}R\int_{\partial B_R}|\nabla \psi|^2.
\end{align*}
In the last inequality, we have just used $t^2\leq1$. Using Poincaré's inequality on the circle of radius $R$
$$\int_{\partial B_R}(\psi-\theta_R)^2\leq R^2\int_{\partial B_R}|\tau\cdot\nabla \psi|^2\leq R^2\int_{\partial B_R}|\nabla \psi|^2$$
and inserting in~\eqref{eq:GS_perforate}, we arrive at our final inequality
\begin{multline}
2\pi D\left(D-\frac{1+\eps}{1-\eps}\right)\log R+\int_{B_R\setminus B_{R_0}}|\nabla \psi|^2\\
\leq \frac{1+\eps}{1-\eps}(1+\alpha)\left(\frac{1}{\log(\alpha^{-1})}+\frac{1-\alpha^2}{2}\right) R\int_{\partial B_R}|\nabla \psi|^2+\frac{C}{\varrho(1-\eps)}.
 \label{eq:final_psi}
\end{multline}
At this step we require $\eps<1/3$ to ensure that $\frac{1+\eps}{1-\eps}<2 \le D$, and hence the logarithmic term on the left-hand side is then positive. 
For any such $\eps$ we can choose $\alpha$ small enough so that
$$\beta:=\frac{1+\eps}{1-\eps}(1+\alpha)\left(\frac{1}{\log(\alpha^{-1})}+\frac{1-\alpha^2}{2}\right)<1.$$
Next we argue as in~\cite{BreMerRiv-94}. We let
$$\zeta(R):=\int_{B_R\setminus B_{R_0}}|\nabla \psi|^2$$
so that~\eqref{eq:final_psi} provides
$$\zeta(R)\leq \beta R\zeta'(R)+C.$$
From Lemma~\ref{lem:estim_upper_full_kinetic}, we know that $\zeta(R)\leq CR$ but then we conclude immediately that $\zeta$ is bounded, by~\cite[Lem.~1]{BreMerRiv-94}. Hence there exists a sequence $R_n\to\ii$ such that
$$\zeta'(R_n)=\int_{\partial B_{R_n}}|\nabla\psi|^2\leq \frac{C}{R_n}$$
(otherwise we would get $\zeta'(R)>C/R$ at infinity and thus $\zeta\geq C\log R$). The right side of~\eqref{eq:final_psi} is thus finite and
$$2\pi D\left(D-\frac{1+\eps}{1-\eps}\right)\log R_n\leq C.$$
This is a contradiction which concludes the proof that $D\in\{0,-1,1\}$.

In the half space the proof is the same, except that in Lemma~\ref{lem:estim_lower_full_kinetic} the coefficient of $\log R$ is $2\pi (2D)^2/2=4\pi D^2$ (see~\cite[Lem.~XI.1]{BetBreHel-94} and~\cite[Thm.~1.3]{Sandier-98}) whereas the right side of~\eqref{eq:GS_perforate} is unchanged. Hence in this case we must have $D=0$.
\end{proof}

With Theorem~\ref{thm:degree_one} at hand, we are able to prove the existence of degree-one infinite ground states, using a thermodynamic limit. This was the content of Theorem~\ref{thm:exist_vortex}.

\begin{proof}[Proof of Theorem~\ref{thm:exist_vortex}]
Let us consider a minimizer $u_L$ of the following problem
$$\min_{\substack{u\in H^1(B_L)\\ u_{|\partial B_L}=g}}\int_{B_L}|\nabla u|^2+\frac12\iint_{B_L\times B_L}(|u(x)|^2-\rho)(|u(y)|^2-\rho)w(x-y)\,\dx\,\dy,$$
where
$$g(x)=\sqrt\rho\, \frac{x_1+ix_2}{L},\qquad \rho=\frac{\mu}{\int_{\R^2}w}.$$
Recall that $\widehat{w}\geq0$, hence $\varrho=\rho$ in this proof. We have
\begin{equation}
\begin{cases}
\big(-\Delta +|u|^2\ast w \big)u_L=\rho(\1_{B_L}\ast w )u_L&\text{on the disk $B_L$,}\\
u_L=g&\text{on the circle $\partial B_L$.}
  \end{cases}
 \label{eq:problem_g_disk}
\end{equation}
This is not exactly the GP equation in the disk, since the right side is only close to $\mu u_L$. However, we can obtain similar uniform bounds as before.

\begin{lemma}\label{lem:uniform_bound_u_L}
For $\mu\leq 1/C$ and $L\geq C/\sqrt\mu$ with $C$ large enough, we have
$$\|u_L\|_{L^\ii(B_L)}\leq C\sqrt\mu,\qquad \|\nabla u_L\|_{L^\ii(B_L)}\leq C\mu,\qquad \|\Delta u_L\|_{L^\ii(B_L)}\leq C\mu^{\frac32}.$$
\end{lemma}

\begin{proof}
We can argue as in Corollary~\ref{cor:pointwise_bounds_I} but since we have to handle the boundary, we give a simpler argument specific to the space dimension $d=2$. We take
$$G(x)=\sqrt\rho\, \frac{x_1+ix_2}{L}\chi\big(\sqrt\mu(L-|x|)\big),$$
where $\chi\in C^\ii_c(\R)$ satisfies $\chi(0)=1$. Then we have $|G|\leq C\sqrt\mu$, $|\nabla G|\leq C(\sqrt\mu/L+\mu)\leq C\mu$ and $|\Delta G|\leq C\mu^{3/2}$ since we assume that $L\geq C/\sqrt\mu$. Next we note that the proof of Corollary~\ref{cor:inh-bdr} applies without changes for $u_L$ and $G$, using $|\1_{B_L}\ast w|\leq \int_{\R^2}|w|$ in~\eqref{eq:estim_mu_inhomog}. This provides the local bounds
$$\int_{B(z,\ell)}|u_L|^2\leq C\left(\mu+\frac1{\ell^2}\right)\ell^2,\quad  \int_{B(z,\ell)}|\nabla u_L|^2
 \leq C\left(\mu+\frac1{\ell^2}\right)^2\ell^2$$
for any disk $B(z,\ell)$ with $\ell\geq C$. The Gagliardo-Nirenberg inequality in the disk $B=B(z,\ell)$ and the above local bounds give
$$\norm{u_L}_{L^4(B(z,\ell))} \leq C\ell^{-\frac12}\left(1+\ell^2\mu\right),\qquad \norm{u_L}_{L^8(B(z,\ell))} \leq C\ell^{-\frac34}\left(1+\ell^2\mu\right).$$
At this step we choose $\ell=1/\sqrt\mu$. Decomposing $|w\ast|u|^2|\leq |w|\ast (|u_L|^2\1_{B(z,2\ell)})+C\ell^{-s}\sup_y\int_{B(y,\ell)}|u_L|^2$ as in Lemma~\ref{lem:estim_potential}, we can conclude that
\begin{multline*}
\norm{u_L(w\ast|u_L|^2)}_{L^2(B(z,\ell))}\\
\leq C\norm{u_L}_{L^4(B(z,\ell))}\norm{u_L}_{L^8(B(z,2\ell))}^2+\frac{C}{\ell^s}\sup_y\left(\int_{B(y,\ell)}|u_L|^2\right)^{\frac32}\leq C\mu.
\end{multline*}
Using Equation~\eqref{eq:problem_g_disk} we conclude that
$$\int_{B(z,\ell)}|\Delta u_L|^2\leq C\mu^2,\qquad \ell=1/\sqrt\mu.$$
At this step it is easier to rescale everything by $\sqrt\mu$ and introduce
$$v_L(x)=\mu^{-\frac12}\Big(u(x/\sqrt\mu)-G(x/\sqrt\mu)\Big).$$
Then $v_L$ is defined on the ball $B_{L\sqrt\mu}$, satisfies the Dirichlet boundary condition and the bounds
\begin{equation}
\sup_{z\in\R^2}\int_{B(z,1)}\big(|v_L|^2+|\nabla v_L|^2+|\Delta v_L|^2\big)\leq C.
\label{eq:rescale_bound_v_L}
\end{equation}
We claim that this implies $|v_L|\leq C$. One simple way of proving this is to write as in the proof of Lemma~\ref{lem:Landau-Kolomogorov_ball}
$$v_L=(-\Delta_D+1)^{-1}(-\Delta_D+1)v_L$$
where $-\Delta_D$ denotes the Dirichlet Laplacian on the ball $B_{L\sqrt\mu}$, and to use the pointwise bound on the kernel of the resolvent
$$(-\Delta_D+1)^{-1}(x,y)\leq (2\pi)^{-1} Y_1(x-y)$$
where $Y_1$ is the Yukawa potential in $\R^2$. This implies the domination
$$|v_L(x)|\leq (2\pi)^{-1}\int_{\R^2}Y_1(y) \big(|\Delta v_L(x-y)|+|v_L(x-y)|\big)\,\dy,$$
which proves that $|v_L|\leq C$ after splitting space into cubes $C_j$ and using that $\sum_j \norm{Y_1}_{L^2(C_j)}<\ii$ since $Y_1$ decays exponentially fast. Thus we have proved that $|u_L|\leq C\sqrt\mu$ which, after inserting in the equation, implies $|\Delta u_L|\leq C\mu^{3/2}$ as claimed. Finally, the estimate on the gradient follows from Lemma~\ref{lem:Landau-Kolomogorov_ball}.
\end{proof}

Now that we have proved uniform bounds, we multiply Equation~\eqref{eq:problem_g_disk} by $x\cdot\nabla \overline{u_L}$ and integrate. After letting $f_L:=(\rho-|u_L|^2)\1_{B_L}\in H^1_0(B_L)\subset H^1(\R^2)$, we obtain the Pohozaev identity
\begin{equation}
\frac{L}2\int_{\partial B_L}|\nu\cdot \nabla u_L|^2-\frac12\int_{B_L} (w\ast f_L)x\cdot\nabla f_L=\frac{L}2\int_{\partial B_L}|\tau\cdot \nabla u_L|^2=\pi\rho.
\label{eq:Pohozaev_pre}
\end{equation}
Using that $f_L=0$ on the circle $\partial B_L$, we can write
\begin{align*}
 &-\Re\int_{B_L}(w\ast f_L)u_L\, x\cdot\nabla \overline{u_L}\\
 &\qquad\qquad =\frac12\int_{B_L} (w\ast f_L)x\cdot\nabla f_L\\
&\qquad\qquad =-\frac12\int_{B_L} f_L(\nabla\cdot x w)\ast f_L-\frac12\int_{B_L} f_Lw\ast (2+x\cdot \nabla f_L).
\end{align*}
Therefore we find the formula
\begin{equation*}
- \Re\int_{B_L}(w\ast f_L)u_L\, x\cdot\nabla \overline{u_L}=-\frac14\int_{B_L} f_L(\nabla\cdot x w)\ast f_L-\frac{1}2\int_{B_L} f_Lw\ast f_L
\end{equation*}
which, after inserting in~\eqref{eq:Pohozaev_pre}, provides us with
\begin{equation}
\frac{L}2\int_{\partial B_L}|\nu\cdot\nabla u_L|^2+\frac14\int_{B_L} f_L(2w+\nabla\cdot x w)\ast f_L=\pi\rho.
\label{eq:Pohozaev_w}
\end{equation}
The Fourier transform of the modified interaction equals
$$\widehat{(2w+\nabla\cdot xw)}(k)=2\widehat{w}(k)-k\cdot\nabla_k \widehat{w}(k)$$
which controls $c|\widehat{w}(k)|^2$ due to our assumption~\eqref{eq:assumption_Fourier_w_Pohozaev}. Therefore we have proved the uniform bound
\begin{equation}
\int_{\R^2}(f_L\ast w)^2=4\pi^2\int_{\R^2}|\widehat{w}(k)|^2|\widehat{f_L}(k)|^2\,\rd k\leq C\rho.
\label{eq:bound_f_Pohozaev}
\end{equation}

Next we prove that $f_L$ must be small everywhere in the ball $B_L$, except perhaps in finitely many bounded balls, this number depending on $\mu$ but not on $L$. Lemma~\ref{lem:uniform_bound_u_L} implies for $\mu$ small enough and $L$ large enough
$$\|\nabla f_L\|_{L^\ii}\leq 2\|u_L\|_{L^\ii}\|\nabla u_L\|_{L^\ii}\leq C\rho \sqrt\mu$$
and
$$\norm{\nabla f_L\ast w}_{L^\ii}\leq \|\nabla f_L\|_{L^\ii}\int_{\R^2}|w|\leq C\rho\sqrt \mu.$$
This gives
$$\norm{f_L\int_{\R^2} w-f_L\ast w}_{L^\ii}\leq\norm{\nabla f_L}_{L^\ii}\int_{\R^2}|x|\,|w(x)|\,\dx\leq C\rho\sqrt \mu.$$
Let us now assume that $\mu$ is small enough so that
$$\|\nabla f_L\|_{L^\ii}\leq \frac{\rho}{16},\qquad \norm{f_L\int_{\R^2} w-f_L\ast w}_{L^\ii}\leq\frac{\rho}{16}\int_{\R^2}w.$$
From the above bounds we conclude that if $|f_L(x_L^{(1)})|\geq \rho/4$ at one point $x_L^{(1)}$ in the ball $B_L$, then $|f(x)|\geq 3\rho/16$ over the ball $B(x_L^{(1)},1)$, and hence $|f_L\ast w(x)|\geq \rho/8\int_{\R^2}w$ over the same ball. In particular,
$$\int_{B(x_L^{(1)},1)}(f_L\ast w)^2\geq \frac{\pi\rho^2}{64}\left(\int_{\R^2}w\right)^2.$$
Next we iterate the procedure outside of $B(x_L^{(1)},2)$. If there exists a point $x_L^{(2)}$ in $B_L\setminus B(x_L^{(1)},2)$ such that $|f_L(x_L^{(2)})|\geq \rho/4$ we get again $\int_{B(x_L^{(2)},1)}(f_L\ast w)^2\geq \frac{\pi\rho^2}{64}(\int w)^2$ and we go on like this. Due to~\eqref{eq:bound_f_Pohozaev}, the procedure must terminate after at most $C/\rho$ steps since the balls $B(x_L^{(j)},1)$ are disjoint by construction. We have thus proved that the set $\{|f_L|\geq  \rho/4\}$ is contained in the finite union $\cup_{j=1}^JB(x_L^{(j)},2)$ with $J\leq C/\rho$. On the complement we have $|f_L|=|\rho-|u_L|^2|\leq \rho/4<\rho/3$.

Next we study the thermodynamic limit. We first group the $x^{(j)}_L$ into $K\leq J$ clusters so that, after extraction of a subsequence, $|x^{(j)}_L-x^{(j')}_L|$ stays bounded for two centers in the same cluster but diverges as $L\to\ii$ whenever they belong to different clusters. We pick $X^{(k)}_L$ to be any center $x_L^{(j)}$ in the $k$th cluster and let $R_0$ be large enough so that the large ball $B_L^{(k)}:=B(X_L^{(k)},R_0)$ contains the union of the balls $B(x_L^{(j)},2)$ within the clusters. Then we have $|f_L|=|\rho-|u_L|^2|\leq \rho/4<\rho/3$ outside of the $B_L^{(k)}$'s and $|X_L^{(k)}-X_L^{(k')}|\to\ii$ when $L\to\ii$ and $k\neq k'$.

We know that, after extraction, $u_L(\cdot+X_L^{(k)})$ converges locally uniformly, as well as all its derivative, to an infinite ground state $u^{(k)}$. From the local convergence we have $|\rho-|u^{(k)}|^2|\leq \rho/4$ outside of the ball $B_{R_0}$ and from the definition~\eqref{eq:def_degree} of the degree we see that ${\rm deg}(u_L,B_L^{(k)})\to {\rm deg}(u^{(k)},B_{R_0})=:D^{(k)}$. Due to the boundary condition, we know that $u_L$ has total degree $1$ and therefore obtain
$$\sum_{k=1}^K D^{(k)}=1$$
after passing to the limit. If $X_L^{(k)}$ stays at a finite distance to the boundary $\partial B_L$ in the limit $L\to\ii$, our infinite ground state $u^{(k)}$ lives over a half space with the boundary condition 1 (up to a global phase factor). By Theorem~\ref{thm:degree_one} we infer $D^{(k)}=0$. Therefore there must exist at least one non-zero $D^{(k)}$ corresponding to a center $X^{(k)}_L$ so that $\rd(X_L^{(k)},\partial B_L)\to\ii$, hence an infinite ground state in $\R^2$. From Theorem~\ref{thm:degree_one} the corresponding $u^{(k)}$ must have degree $D^{(k)}\in\{-1,1\}$. If $D^{(k)}=-1$ we can replace $u^{(k)}$ by its complex conjugate and the theorem is proved.
\end{proof}

\appendix
\section{Landau-Kolmogorov inequality in a ball}\label{app:Landau-Kolmogorov}

\begin{lemma}[Landau-Kolmogorov inequality in a ball]\label{lem:Landau-Kolomogorov_ball}
  There exists a constant $C$ such that
$$\norm{\nabla f}_{L^\ii(B_R)}^2\leq C\norm{\Delta f}_{L^\ii(B_R)}\norm{f}_{L^\ii(B_R)},$$
for any $R>0$ and any $f\in H^2(B_R)$ satisfying the Dirichlet condition $f_{|\partial B_R}=0$.
\end{lemma}

\begin{proof}
By scaling we can assume that $R=1$. First we recall the inequality
\begin{equation}
\norm{f}_{L^\ii(B_1)}\leq C\norm{\Delta f}_{L^\ii(B_1)},
 \label{eq:Sobolev_infinity_ball}
\end{equation}
for any $f$ in $H^2(B_1)$ vanishing at the boundary $\partial B_1$, that is, $f$ in the operator domain of the Dirichlet Laplacian $\Delta_{\rm D}$. In dimensions $d\in\{1,2,3\}$ this immediately follows from the Sobolev embedding and elliptic regularity
$$\norm{f}_{L^\ii(B_1)}\leq C\norm{f}_{H^2(B_1)}\leq C\norm{\Delta f}_{L^2(B_1)}\leq C\norm{\Delta f}_{L^\ii(B_1)}.$$
In higher dimensions, we cannot use the $L^2$ norm and we instead use the pointwise bound on the Dirichlet heat kernel on $B_1$
$$0\leq e^{t\Delta_D}(x,y)\leq e^{t\Delta_{\R^d}}(x-y)$$
(see, e.g.,~\cite[Thm.~2.1.7]{Davies-90}). Therefore, using $x^{-1}=\int_0^\ii e^{-tx}\,\rd t$ we get after integrating over $t$,
$$0\leq (-\Delta_D)^{-1}(x,y)\leq (2\pi)^{-\frac{d}2} Y_0(x-y)=\frac{C}{|x-y|^{d-2}}.$$
We thus write $f=(-\Delta_{\rm D})^{-1}(-\Delta_{\rm D})f$ and obtain
$$|f(x)|\leq C\int_{B_1}\frac{|\Delta f(y)|}{|x-y|^{d-2}}\dy\leq C\norm{\Delta f}_{L^\ii(B_1)}\int_{B_2}\frac{\dy}{|y|^{d-2}}$$
as was claimed in~\eqref{eq:Sobolev_infinity_ball}.

Next we turn to the estimate on $\nabla f$. By~\cite{Wang-04} we know that the heat kernel of the Dirichlet Laplacian on $B_1$ satisfies
$$\norm{\nabla e^{t\Delta_D}}_{L^\ii(B_1)\to L^\ii(B_1)} \leq C\left(\frac1{\sqrt t}+1\right).$$
As before this gives after multiplying by $e^{-tM^2}$ and integrating over $t$,
$$\norm{\nabla (-\Delta_D+M^2)^{-1}}_{L^\ii(B_1)\to L^\ii(B_1)}\leq C\left(\frac1M+\frac1{M^2}\right).$$
Thus we obtain by writing $\nabla f=\nabla (-\Delta_{\rm D}+M^2)^{-1}(-\Delta_{\rm D}+M^2)f$
$$\norm{\nabla f}_{L^\ii(B_1)}\leq C\left(1+M^{-1}\right)\left(M^{-1}\norm{\Delta f}_{L^\ii(B_1)}+M\norm{f}_{L^\ii(B_1)}\right).$$
Taking $M^2=\norm{\Delta f}_{L^\ii(B_1)}/\norm{f}_{L^\ii(B_1)}$ gives
$$\norm{\nabla f}_{L^\ii(B_1)}\leq C\left(\norm{\Delta f}^{\frac12}_{L^\ii(B_1)}\norm{f}^{\frac12}_{L^\ii(B_1)}+\norm{f}_{L^\ii(B_1)}\right).$$
Using finally~\eqref{eq:Sobolev_infinity_ball} for the last term we obtain the claim.
\end{proof}

The same result holds for a smooth domain $\Omega$, with $R$ replaced by $|\Omega|^{1/d}$, for a constant $C$ depending on the shape of the normalized set $\omega:=|\Omega|^{-1/d}\Omega$. In this paper we only need the result for a ball.

\section{Thermodynamic limit}\label{app:existence-thermo}

In this section we prove the existence of the limits $e(\rho)$ and $f(\mu)$ in \eqref{eq:def_e_lambda} and \eqref{eq:def_f_mu}.

\begin{lemma}[Existence of thermodynamic limit]\label{lem:existence-thermo}
Let $w$ be a potential satisfying Assumption~\ref{ass:w}. Let $\Omega_n\subset\R^d$ be a sequence of domains satisfying~\eqref{eq:Fisher-0} and \eqref{eq:Fisher}.
Then for all $\rho,\mu>0$, the limits $e(\rho)$ and $f(\mu)$ in \eqref{eq:def_e_lambda} and \eqref{eq:def_f_mu}, respectively, exist  and are independent of the sequence $\{\Omega_n\}$ (as well as the constant $C$ in \eqref{eq:Fisher-0} and \eqref{eq:Fisher}). 
\end{lemma}

\begin{proof}
This is standard in statistical mechanics~\cite{Ruelle} and we only outline the proof. For simplicity we will use the local bounds from Theorems~\ref{thm:uniform_bound} and~\ref{thm:local_bound} to control the interactions between different sub-systems, but it is also possible to argue without using them.

We start with the Dirichlet boundary condition. Let us consider a tiling of $\R^d$ with cubes $C_j=\ell j+(-\ell/2,\ell/2)^d$, $j\in\Z^d$, of fixed side length $\ell$. We consider all the cubes $C_j\subset\Omega_n$. From the assumption~\eqref{eq:Fisher} on the boundary, the number of such cubes satisfies
$$K_n:=\#\{j\in\Z^d\ :\ C_j\subset\Omega_n\}=\ell^{-d}|\Omega_n|+O\big((R_n/\ell)^{d-1}\big).$$
Consider any sequence $\lambda_n\sim |\Omega_n|\rho$ and let
$\lambda'_n:=\lambda_n/K_n=\rho\ell^d+o(1)$.
Let $u'_n$ be a minimizer for $E_{\rm D}(\lambda'_n,C_0)$ in the cube $C_0=(-\ell/2,\ell/2)^d$. We use as trial state
$$u_n:=\sum_{\substack{j\in\Z^d\\C_j\subset\Omega_n}}u'_n(\cdot-j)\in H^1_0(\Omega_n)$$
and obtain
$$E_{\rm D}(\lambda_n,\Omega_n)\leq K_n E_{\rm D}(\lambda'_n,C_0)+\frac12\sum_{j\neq k}\int_{C_0\times C_0}w(x-y)u'_n(x+j)^2u'_n(y+k)^2.$$
Passing to the limit $n\to\ii$ provides
\begin{multline}
\limsup_{n\to\ii}\frac{E_{\rm D}(\lambda_n,\Omega_n)}{|\Omega_n|}\leq \frac{E_{\rm D}(\rho\ell^d,C_0)}{\ell^d}\\
+\frac12\sum_{j\in\Z^d\setminus\{0\}}\int_{C_0\times C_0}w(x-y)u'(x)^2u'(y+j)^2,
\label{eq:upper_thermo_lim}
\end{multline}
for some minimizer $u'$ of $E_{\rm D}(\rho\ell^d,C_0)$. We have used here that $\lambda\mapsto E_{\rm D}(\lambda,A)$ is continuous in a fixed bounded domain $A$ and that corresponding minimizers converge, up to extraction. The last term in~\eqref{eq:upper_thermo_lim} is the interaction of $u'$ with infinitely many copies of itself outside of $C_0$, which we claim is a $o(\ell^d)$. In fact, since $u'$ is uniformly bounded by Corollary~\ref{cor:pointwise_bounds_I}, we can bound it by
\begin{multline*}
\frac12\sum_{j\in\Z^d\setminus\{0\}}\int_{C_0\times C_0}w(x-y)u'(x)^2u'(y+j)^2\\
\leq C\int_{C_0}\int_{\R^d\setminus C_0}|w(x-y)|\,\dx\,\dy\leq C\left(\ell^{d-1}\delta+\frac{\ell^d}{\delta^{s-d}}\right)\leq C\ell^{d-\frac{s-d}{1+s-d}}.
\end{multline*}
where we have distinguished whether $x$ is at distance $\delta$ from the boundary of $C_0$ or not and then have optimized over $\delta$. Passing to the limit $\ell\to\ii$, we find
\begin{equation}
\limsup_{n\to\ii}\frac{E_{\rm D}(\lambda_n,\Omega_n)}{|\Omega_n|}\leq \liminf_{\ell\to\ii}\frac{E_{\rm D}(\rho \ell^d,(-\ell/2,\ell/2)^d)}{\ell^d}.
\label{eq:first_estim_upper_thermo}
\end{equation}
Taking for $\Omega_n$ a big cube $\Omega_n=(-\ell_n/2,\ell_n/2)^d$ and the mass $\lambda_n=\rho|\Omega_n|$, we conclude that the limit for cubes
$$e(\rho):=\lim_{\ell\to\ii}\frac{E_{\rm D}(\rho \ell^d,(-\ell/2,\ell/2)^d)}{\ell^d}$$
exists. The inequality~\eqref{eq:first_estim_upper_thermo} becomes
$$\limsup_{n\to\ii}\frac{E_{\rm D}(\lambda_n,\Omega_n)}{|\Omega_n|}\leq e(\rho).$$

To get the similar lower bound with the liminf, we choose a very large cube $C_n'$ of size proportional to $R_n$ containing $\Omega_n$ and fill the missing space using small cubes of side length $\ell$ as before. We use as trial state for the large cube the minimizer for $E_{\rm D}(\lambda_n,\Omega_n)$ to which we add the same minimizers $u'$ of the small cubes. The interaction is estimated as before (using that the Dirichlet minimizer in $\Omega_n$ is bounded) and this concludes the proof that
$$\lim_{n\to\ii}\frac{E_{\rm D}(\lambda_n,\Omega_n)}{|\Omega_n|}= e(\rho).$$
The proof for the grand-canonical problem $F_{\rm D}(\mu,\Omega)$ is similar and even easier since we do not have any mass constraint to be fulfilled for our trial states.

To prove that the Neumann problem has the same limit as Dirichlet, we choose the localization function
$$\chi(x)=\begin{cases}
1&\text{if $x\in\Omega_n$ and $\rd(x,\partial\Omega_n)\geq 2R$,}\\
\frac{\rd(x,\partial\Omega_n)-R}{R}&\text{if $x\in\Omega_n$ and $R\leq \rd(x,\partial\Omega_n)\leq 2R$,}\\
0&\text{otherwise.}
          \end{cases}$$
We then let $\eta:=\sqrt{1-\chi^2}$. From the IMS formula, we deduce that for a minimizer $u_n$ of the Neumann problem
\begin{multline*}
E_{\rm N}(\lambda_n,\Omega_n)\geq \cE_{\Omega_n}(u_n\chi)-\frac{C}{R^2}\int_{\rd(x,\partial\Omega_n)\leq 2R}u_n^2\\
+\int_{\Omega_n\times\Omega_n}u_n(x)^2\chi^2(x)u_n(y)^2\eta^2(y)w(x-y)\dx,\dy
\end{multline*}
due to the positivity of the kinetic and interaction energies. Using the local bounds from Theorem~\ref{thm:local_bound}, the errors are $o(R_n^d)$ and therefore
$$E_{\rm D}(\lambda_n,\Omega_n)\geq E_{\rm N}(\lambda_n\,,\,\Omega_n)\geq E_{\rm D}\left(\int\chi^2u_n^2,\Omega_n\right)+o(R_n^{d}).$$
Since $\int\chi^2u_n^2=\lambda_n+o(R_n^{d})=\rho|\Omega_n|+o(|\Omega_n|)$ we obtain the result~\eqref{eq:def_e_lambda}. The proof in the grand-canonical case is similar. The proof of Lemma~\ref{lem:existence-thermo} is complete.
\end{proof}


\begin{thebibliography}{100}

\bibitem{AftBlaJer-07}
{\sc A.~Aftalion, X.~Blanc, and R.~L. Jerrard}, {\em Nonclassical rotational
  inertia of a supersolid}, Phys. Rev. Lett., 99 (2007), p.~135301.

\bibitem{AftBlaJer-09}
{\sc A.~Aftalion, X.~Blanc, and R.~L. Jerrard}, {\em Mathematical issues in the
  modelling of supersolids}, Nonlinearity, 22 (2009), p.~1589.

\bibitem{Agmon}
{\sc S.~Agmon}, {\em Lectures on exponential decay of solutions of second-order
  elliptic equations}, Princeton University Press, 1982.

\bibitem{AmbCab-00}
{\sc L.~Ambrosio and X.~Cabr\'{e}}, {\em Entire solutions of semilinear
  elliptic equations in {$\bold R^3$} and a conjecture of {D}e {G}iorgi}, J.
  Amer. Math. Soc., 13 (2000), pp.~725--739.

\bibitem{AncRosToi-13}
{\sc F.~Ancilotto, M.~Rossi, and F.~Toigo}, {\em Supersolid structure and
  excitation spectrum of soft-core bosons in three dimensions}, Phys. Rev. A,
  88 (2013), p.~033618.

\bibitem{BagDelBerHol-13}
{\sc L.~Baguet, F.~Delyon, B.~Bernu, and M.~Holzmann}, {\em {Hartree-Fock
  Ground State Phase Diagram of Jellium}}, Phys. Rev. Lett., 111 (2013),
  p.~166402.

\bibitem{BagDelBerHol-14}
\leavevmode\vrule height 2pt depth -1.6pt width 23pt, {\em Properties of
  {Hartree-Fock} solutions of the three-dimensional electron gas}, Phys. Rev.
  B, 90 (2014), p.~165131.

\bibitem{BatStiTor-09}
{\sc R.~D. Batten, F.~H. Stillinger, and S.~Torquato}, {\em Novel
  low-temperature behavior in classical many-particle systems}, Phys. Rev.
  Lett., 103 (2009), p.~050602.

\bibitem{BelRadShl-10}
{\sc J.~Bellissard, C.~Radin, and S.~Shlosman}, {\em The characterization of
  ground states}, J. Phys. A: Math. Theor., 43 (2010), p.~305001.

\bibitem{BerDelDunHol-08}
{\sc B.~Bernu, F.~Delyon, M.~Duneau, and M.~Holzmann}, {\em Metal-insulator
  transition in the {Hartree-Fock} phase diagram of the fully polarized
  homogeneous electron gas in two dimensions}, Phys. Rev. B, 78 (2008),
  p.~245110.

\bibitem{BerDelHolBag-11}
{\sc B.~Bernu, F.~Delyon, M.~Holzmann, and L.~Baguet}, {\em {Hartree-Fock phase
  diagram of the two-dimensional electron gas}}, Phys. Rev. B, 84 (2011),
  p.~115115.

\bibitem{BetBreHel-94}
{\sc F.~Bethuel, H.~Brezis, and F.~H{\'e}lein}, {\em Ginzburg-{L}andau
  vortices}, Progress in Nonlinear Differential Equations and their
  Applications, 13, Birkh{\"a}user Boston, Inc., Boston, MA, 1994.

\bibitem{BetBreOrl-01}
{\sc F.~Bethuel, H.~Brezis, and G.~Orlandi}, {\em Asymptotics for the
  {G}inzburg-{L}andau equation in arbitrary dimensions}, J. Funct. Anal., 186
  (2001), pp.~432--520.

\bibitem{BlaBriLio-03}
{\sc X.~Blanc, C.~{Le Bris}, and P.-L. Lions}, {\em A definition of the ground
  state energy for systems composed of infinitely many particles}, Comm.
  Partial Differential Equations, 28 (2003), pp.~439--475.

\bibitem{BlaLew-15}
{\sc X.~Blanc and M.~Lewin}, {\em The crystallization conjecture: A review},
  EMS Surv. Math. Sci., 2 (2015), pp.~255--306.

\bibitem{BotSchWenHerGuoLanPfa-19}
{\sc F.~B\"ottcher, J.-N. Schmidt, M.~Wenzel, J.~Hertkorn, M.~Guo, T.~Langen,
  and T.~Pfau}, {\em Transient supersolid properties in an array of dipolar
  quantum droplets}, Phys. Rev. X, 9 (2019), p.~011051.

\bibitem{BreMerRiv-94}
{\sc H.~Brezis, F.~Merle, and T.~Rivi{\`e}re}, {\em {Quantization effects for
  $-\Delta u = u(1 -|u|^2)$ in $\mathbb{R}^2$}}, Arch. Rat. Mech. Anal., 126
  (1994), pp.~35--58.

\bibitem{ButLeb-05}
{\sc P.~Butt{\`a} and J.~L. Lebowitz}, {\em Local mean field models of uniform
  to nonuniform density fluid--crystal transitions}, The Journal of Physical
  Chemistry B, 109 (2005), pp.~6849--6854.

\bibitem{CahHil-58}
{\sc J.~W. Cahn and J.~E. Hilliard}, {\em {Free Energy of a Nonuniform System.
  I. Interfacial Free Energy}}, J. Chem. Phys., 28 (1958), pp.~258--267.

\bibitem{CanEhr-11}
{\sc E.~Canc\`es and V.~Ehrlacher}, {\em Local defects are always neutral in
  the {T}homas--{F}ermi--von {W}eisz\"{a}cker theory of crystals}, Arch.
  Ration. Mech. Anal., 202 (2011), pp.~933--973.

\bibitem{Casotti_etal-24}
{\sc E.~Casotti, E.~Poli, L.~Klaus, A.~Litvinov, C.~Ulm, C.~Politi, M.~J. Mark,
  T.~Bland, and F.~Ferlaino}, {\em Observation of vortices in a dipolar
  supersolid}, Nature, 635 (2024), pp.~327--331.

\bibitem{ChaWei-18}
{\sc H.~Chan and J.~Wei}, {\em {On De Giorgi's conjecture: Recent progress and
  open problems}}, Science China Mathematics, 61 (2018), pp.~1925--1946.

\bibitem{ChaHalRea-13}
{\sc M.~H.~W. Chan, R.~B. Hallock, and L.~Reatto}, {\em {Overview on Solid
  $^4$He and the Issue of Supersolidity}}, J. Low Temp. Phys., 172 (2013),
  pp.~317--363.

\bibitem{Chomaz_etal-19}
{\sc L.~Chomaz, D.~Petter, P.~Ilzh\"ofer, G.~Natale, A.~Trautmann, C.~Politi,
  G.~Durastante, R.~M.~W. van Bijnen, A.~Patscheider, M.~Sohmen, M.~J. Mark,
  and F.~Ferlaino}, {\em Long-lived and transient supersolid behaviors in
  dipolar quantum gases}, Phys. Rev. X, 9 (2019), p.~021012.

\bibitem{CinJaiBinMicZolPup-10}
{\sc F.~{Cinti}, P.~{Jain}, M.~{Boninsegni}, A.~{Micheli}, P.~{Zoller}, and
  G.~{Pupillo}}, {\em {Supersolid Droplet Crystal in a Dipole-Blockaded Gas}},
  Phys. Rev. Lett., 105 (2010), p.~135301.

\bibitem{CinMacLecPupPoh-14}
{\sc F.~{Cinti}, T.~{Macr{\`\i}}, W.~{Lechner}, G.~{Pupillo}, and T.~{Pohl}},
  {\em {Defect-induced supersolidity with soft-core bosons}}, Nat. Commun., 5
  (2014), p.~3235.

\bibitem{ComMir-99}
{\sc M.~Comte and P.~Mironescu}, {\em Minimizing properties of arbitrary
  solutions to the {G}inzburg-{L}andau equation}, Proc. Roy. Soc. Edinburgh
  Sect. A, 129 (1999), pp.~1157--1169.

\bibitem{ConJer-17}
{\sc A.~Contreras and R.~L. Jerrard}, {\em Nearly parallel vortex filaments in
  the 3{D} {G}inzburg-{L}andau equations}, Geom. Funct. Anal., 27 (2017),
  pp.~1161--1230.

\bibitem{Davies-74}
{\sc E.~B. Davies}, {\em Properties of the {G}reen's functions of some
  {S}chr\"{o}dinger operators}, J. London Math. Soc. (2), 7 (1974),
  pp.~483--491.

\bibitem{Davies-90}
\leavevmode\vrule height 2pt depth -1.6pt width 23pt, {\em Heat kernels and
  spectral theory}, vol.~92 of Cambridge Tracts in Mathematics, Cambridge
  University Press, Cambridge, 1990.

\bibitem{DeGiorgi-79}
{\sc E.~De~Giorgi}, {\em Convergence problems for functionals and operators},
  in Proceedings of the {I}nternational {M}eeting on {R}ecent {M}ethods in
  {N}onlinear {A}nalysis ({R}ome, 1978), 1979, pp.~131--188,.

\bibitem{DeLaire-12}
{\sc A.~de~Laire}, {\em Nonexistence of traveling waves for a nonlocal
  {G}ross-{P}itaevskii equation}, Indiana Univ. Math. J., 61 (2012),
  pp.~1451--1484.

\bibitem{DeLaiLop-22}
{\sc A.~de~Laire and S.~L\'{o}pez-Mart\'{\i}nez}, {\em Existence and decay of
  traveling waves for the nonlocal {G}ross-{P}itaevskii equation}, Comm.
  Partial Differential Equations, 47 (2022), pp.~1732--1794.

\bibitem{DeLaiMen-20}
{\sc A.~de~Laire and P.~Mennuni}, {\em Traveling waves for some nonlocal 1{D}
  {G}ross-{P}itaevskii equations with nonzero conditions at infinity}, Discrete
  Contin. Dyn. Syst., 40 (2020), pp.~635--682.

\bibitem{PinKowWei-11}
{\sc M.~del Pino, M.~Kowalczyk, and J.~Wei}, {\em On {D}e {G}iorgi's conjecture
  in dimension {$N\geq 9$}}, Ann. of Math. (2), 174 (2011), pp.~1485--1569.

\bibitem{Dobrushin-68b}
{\sc R.~Dobru\v{s}in}, {\em The problem of uniqueness of a gibbsian random
  field and the problem of phase transitions}, Funct. Anal. Appl., 2 (1968),
  pp.~302--312.

\bibitem{Dobrushin-68a}
{\sc R.~L. Dobru\v{s}in}, {\em Gibbsian random fields for lattice systems with
  pairwise interactions}, Funkcional. Anal. i Prilo\v{z}en., 2 (1968),
  pp.~31--43.

\bibitem{Dobrushin-69}
\leavevmode\vrule height 2pt depth -1.6pt width 23pt, {\em Gibbsian random
  fields. {G}eneral case}, Funkcional. Anal. i Prilo\v{z}en, 3 (1969),
  pp.~27--35.

\bibitem{DobMin-67}
{\sc R.~L. Dobru\v{s}in and R.~A. Minlos}, {\em Existence and continuity of
  pressure in classical statistical physics}, Teor. Verojatnost. i Primenen.,
  12 (1967), pp.~595--618.

\bibitem{EspNicPul-82}
{\sc R.~Esposito, F.~Nicol{\`o}, and M.~Pulvirenti}, {\em Superstable
  interactions in quantum statistical mechanics: {M}axwell-{B}oltzmann
  statistics}, Ann. Inst. H. Poincar{\'e} Sect. A (N.S.), 36 (1982),
  pp.~127--158.

\bibitem{Farina-98}
{\sc A.~Farina}, {\em Finite-energy solutions, quantization effects and
  {L}iouville-type results for a variant of the {G}inzburg-{L}andau systems in
  {${\bf R}^K$}}, Differential Integral Equations, 11 (1998), pp.~875--893.

\bibitem{FarMir-13}
{\sc A.~Farina and P.~Mironescu}, {\em Uniqueness of vortexless
  {G}inzburg-{L}andau type minimizers in two dimensions}, Calc. Var. Partial
  Differential Equations, 46 (2013), pp.~523--554.

\bibitem{FraHaiSeiSol-12}
{\sc R.~L. {Frank}, C.~{Hainzl}, R.~{Seiringer}, and J.~P. {Solovej}}, {\em
  Microscopic derivation of ginzburg-landau theory}, J. Amer. Math. Soc., 25
  (2012), pp.~667--713.

\bibitem{Gates-72}
{\sc D.~Gates}, {\em Rigorous results in the mean-field theory of freezing},
  Ann. Phys., 71 (1972), pp.~395 -- 420.

\bibitem{GatPen-69}
{\sc D.~J. Gates and O.~Penrose}, {\em The van der {W}aals limit for classical
  systems. {I}. {A} variational principle}, Comm. Math. Phys., 15 (1969),
  pp.~255--276.

\bibitem{GhoGui-98}
{\sc N.~Ghoussoub and C.~Gui}, {\em On a conjecture of {D}e {G}iorgi and some
  related problems}, Math. Ann., 311 (1998), pp.~481--491.

\bibitem{Ginibre-67}
{\sc J.~Ginibre}, {\em Rigorous lower bound on the compressibility of a
  classical system}, Phys. Lett. A, 24 (1967), pp.~223--224.

\bibitem{GinVel-80}
{\sc J.~Ginibre and G.~Velo}, {\em On a class of nonlinear {S}chr{\"o}dinger
  equations with nonlocal interaction}, Math. Z., 170 (1980), pp.~109--136.

\bibitem{GinLan-50}
{\sc V.~Ginzburg and L.~Landau}, {\em On the theory of superconductivity}, Zh.
  Eksp. Teor. Fiz., 20 (1950), pp.~1064--82.

\bibitem{GiuLebLie-09}
{\sc A.~Giuliani, J.~L. Lebowitz, and E.~H. Lieb}, {\em Periodic minimizers in
  1d local mean field theory}, Comm. Math. Phys., 286 (2009), pp.~163--177.

\bibitem{Gross-58}
{\sc E.~Gross}, {\em Classical theory of boson wave fields}, Ann. Phys., 4
  (1958), pp.~57--74.

\bibitem{Gross-61}
\leavevmode\vrule height 2pt depth -1.6pt width 23pt, {\em Structure of a
  quantized vortex in boson systems}, Nuovo Cimento, 20 (1961), pp.~454--477.

\bibitem{Gross-57}
{\sc E.~P. {Gross}}, {\em {Unified Theory of Interacting Bosons}}, Phys. Rev.,
  106 (1957), pp.~161--162.

\bibitem{Hallock-15}
{\sc R.~Hallock}, {\em Is solid helium a supersolid?}, Phys. Today, 68 (2015),
  pp.~30--35.

\bibitem{HenCinJaiPupPoh-12}
{\sc N.~Henkel, F.~Cinti, P.~Jain, G.~Pupillo, and T.~Pohl}, {\em {Supersolid
  Vortex Crystals in Rydberg-Dressed Bose-Einstein Condensates}}, Phys. Rev.
  Lett., 108 (2012), p.~265301.

\bibitem{HenNatPoh-10}
{\sc N.~Henkel, R.~Nath, and T.~Pohl}, {\em {Three-Dimensional Roton
  Excitations and Supersolid Formation in Rydberg-Excited Bose-Einstein
  Condensates}}, Phys. Rev. Lett., 104 (2010), p.~195302.

\bibitem{HerHer-96}
{\sc R.~M. Herv\'{e} and M.~Herv\'{e}}, {\em Quelques propri\'{e}t\'{e}s des
  solutions de l'\'{e}quation de {G}inzburg-{L}andau sur un ouvert de {${\bf
  R}^2$}}, Potential Anal., 5 (1996), pp.~591--609.

\bibitem{IgnNguSlaZar-20}
{\sc R.~Ignat, L.~Nguyen, V.~Slastikov, and A.~Zarnescu}, {\em On the
  uniqueness of minimisers of {G}inzburg-{L}andau functionals}, Ann. Sci.
  \'{E}c. Norm. Sup\'{e}r. (4), 53 (2020), pp.~589--613.

\bibitem{JexLewMad-24}
{\sc M.~Jex, M.~Lewin, and P.~Madsen}, {\em Classical {D}ensity {F}unctional
  {T}heory: {T}he {L}ocal {D}ensity {A}pproximation}, Rev. Math. Phys.,
  (2024).
\newblock in press.

\bibitem{JosPomRic-07b}
{\sc C.~Josserand, Y.~Pomeau, and S.~Rica}, {\em Patterns and supersolids},
  Eur. Phys. J. Special Topics, 146 (2007), pp.~47--61.

\bibitem{Kadau-etal-16}
{\sc H.~Kadau, M.~Schmitt, M.~Wenzel, C.~Wink, T.~Maier, I.~Ferrier-Barbut, and
  T.~Pfau}, {\em Observing the {R}osensweig instability of a quantum
  ferrofluid}, Nature, 530 (2016), pp.~194--197.

\bibitem{KirNep-71}
{\sc D.~A. {Kirzhnits} and Y.~A. {Nepomnyashchi{\v{i}}}}, {\em {Coherent
  Crystallization of Quantum Liquid}}, Soviet Journal of Experimental and
  Theoretical Physics, 32 (1971), p.~1191.

\bibitem{Klein_etal-94}
{\sc W.~Klein, H.~Gould, R.~A. Ramos, I.~Clejan, and A.~I. Mel'cuk}, {\em
  Repulsive potentials, clumps and the metastable glass phase}, Physica A, 205
  (1994), pp.~738--746.

\bibitem{KunKat-12}
{\sc M.~Kunimi and Y.~Kato}, {\em Mean-field and stability analyses of
  two-dimensional flowing soft-core bosons modeling a supersolid}, Phys. Rev.
  B, 86 (2012), p.~060510.

\bibitem{LanRue-69}
{\sc O.~E. Lanford, III and D.~Ruelle}, {\em Observables at infinity and states
  with short range correlations in statistical mechanics}, Comm. Math. Phys.,
  13 (1969), pp.~194--215.

\bibitem{Legett-70}
{\sc A.~J. Leggett}, {\em Can a solid be "superfluid"?}, Phys. Rev. Lett., 25
  (1970), pp.~1543--1546.

\bibitem{LeoMorZupEssDon-17}
{\sc J.~L\'eonard, A.~Morales, P.~Zupancic, T.~Esslinger, and T.~Donner}, {\em
  Supersolid formation in a quantum gas breaking a continuous translational
  symmetry}, Nature, 543 (2017), pp.~87--90.

\bibitem{Lewin-15}
{\sc M.~Lewin}, {\em Mean-field limit of {B}ose systems: rigorous results}, in
  Proceedings of the {I}nternational {C}ongress of {M}athematical {P}hysics,
  {S}antiago de {C}hile, 2015.
\newblock ArXiv e-prints.

\bibitem{Lewin-22}
\leavevmode\vrule height 2pt depth -1.6pt width 23pt, {\em Coulomb and {R}iesz
  gases: {T}he known and the unknown}, J. Math. Phys., 63 (2022), p.~061101.
\newblock Special collection in honor of Freeman Dyson.

\bibitem{LewNamRou-14}
{\sc M.~Lewin, P.~T. Nam, and N.~Rougerie}, {\em Derivation of {H}artree's
  theory for generic mean-field {B}ose systems}, Adv. Math., 254 (2014),
  pp.~570--621.

\bibitem{LiLeeBurShtTopJamKet-17}
{\sc J.-R. {Li}, J.~{Lee}, W.~{Huang}, S.~{Burchesky}, B.~{Shteynas},
  F.~{\c{C}}. {Top}, A.~O. {Jamison}, and W.~{Ketterle}}, {\em {A stripe phase
  with supersolid properties in spin-orbit-coupled Bose-Einstein condensates}},
  Nature, 543 (2017), pp.~91--94.

\bibitem{LieLos-01}
{\sc E.~H. Lieb and M.~Loss}, {\em Analysis}, vol.~14 of Graduate Studies in
  Mathematics, American Mathematical Society, Providence, RI, 2nd~ed., 2001.

\bibitem{LieSeiSolYng-05}
{\sc E.~H. Lieb, R.~Seiringer, J.~P. Solovej, and J.~Yngvason}, {\em The
  mathematics of the {B}ose gas and its condensation}, Oberwolfach {S}eminars,
  Birkh{\"a}user, 2005.

\bibitem{LieSeiYng-00}
{\sc E.~H. Lieb, R.~Seiringer, and J.~Yngvason}, {\em Bosons in a trap: A
  rigorous derivation of the {G}ross-{P}itaevskii energy functional}, Phys.
  Rev. A, 61 (2000), p.~043602.

\bibitem{LieThi-86}
{\sc E.~H. Lieb and W.~E. Thirring}, {\em Universal nature of {V}an {D}er
  {W}aals forces for {C}oulomb systems}, Phys. Rev. A, 34 (1986), pp.~40--46.

\bibitem{LieYng-98}
{\sc E.~H. Lieb and J.~Yngvason}, {\em Ground state energy of the low density
  {B}ose gas}, Phys. Rev. Lett., 80 (1998), pp.~2504--2507.

\bibitem{LikLanWatLow-01}
{\sc C.~N. Likos, A.~Lang, M.~Watzlawek, and H.~L{\"o}wen}, {\em Criterion for
  determining clustering versus reentrant melting behavior for bounded
  interaction potentials}, Phys. Rev. E, 63 (2001), p.~031206.

\bibitem{Likos-07}
{\sc C.~N. Likos, B.~M. Mladek, D.~Gottwald, and G.~Kahl}, {\em Why do
  ultrasoft repulsive particles cluster and crystallize? analytical results
  from density-functional theory}, J. Chem. Phys., 126 (2007).

\bibitem{LikWatLow-98}
{\sc C.~N. Likos, M.~Watzlawek, and H.~L\"owen}, {\em Freezing and clustering
  transitions for penetrable spheres}, Phys. Rev. E, 58 (1998), pp.~3135--3144.

\bibitem{MacMacCinPoh-13}
{\sc T.~{Macr{\`\i}}, F.~{Maucher}, F.~{Cinti}, and T.~{Pohl}}, {\em
  {Elementary excitations of ultracold soft-core bosons across the
  superfluid-supersolid phase transition}}, Phys. Rev. A, 87 (2013), p.~061602.

\bibitem{MasJos-13}
{\sc P.~Mason and C.~Josserand}, {\em Phase-jump nucleation in a
  one-dimensional model of a supersolid}, Phys. Rev. B, 88 (2013), p.~224506.

\bibitem{MatTsu-70}
{\sc H.~Matsuda and T.~Tsuneto}, {\em {Off-Diagonal Long-Range Order in
  Solids}}, Prog. Theor. Phys. Supp., 46 (1970), pp.~411--436.

\bibitem{CorNatLi-20}
{\sc G.~McCormack, R.~Nath, and W.~Li}, {\em {Dynamical excitation of maxon and
  roton modes in a Rydberg-dressed Bose-Einstein condensate}}, Phys. Rev. A,
  102 (2020), p.~023319.

\bibitem{MilPis-10}
{\sc V.~Millot and A.~Pisante}, {\em Symmetry of local minimizers for the
  three-dimensional {G}inzburg-{L}andau functional}, J. Eur. Math. Soc. (JEMS),
  12 (2010), pp.~1069--1096.

\bibitem{Miranville-19}
{\sc A.~Miranville}, {\em The {C}ahn-{H}illiard equation}, vol.~95 of CBMS-NSF
  Regional Conference Series in Applied Mathematics, Society for Industrial and
  Applied Mathematics (SIAM), Philadelphia, PA, 2019.
\newblock Recent advances and applications.

\bibitem{Mironescu-96}
{\sc P.~Mironescu}, {\em Les minimiseurs locaux pour l'\'{e}quation de
  {G}inzburg-{L}andau sont \`a sym\'{e}trie radiale}, C. R. Acad. Sci. Paris
  S\'{e}r. I Math., 323 (1996), pp.~593--598.

\bibitem{Modica-85}
{\sc L.~Modica}, {\em A gradient bound and a liouville theorem for nonlinear
  poisson equations}, Commun. Pure Appl. Math., 38 (1985), pp.~679--684.

\bibitem{Moser-60}
{\sc J.~Moser}, {\em A new proof of {D}e {G}iorgi's theorem concerning the
  regularity problem for elliptic differential equations}, Comm. Pure Appl.
  Math., 13 (1960), pp.~457--468.

\bibitem{Moser-61}
\leavevmode\vrule height 2pt depth -1.6pt width 23pt, {\em On {H}arnack's
  theorem for elliptic differential equations}, Comm. Pure Appl. Math., 14
  (1961), pp.~577--591.

\bibitem{Murata-86}
{\sc M.~Murata}, {\em Structure of positive solutions to {$(-\Delta+V)u=0$} in
  {${\bf R}^n$}}, Duke Math. J., 53 (1986), pp.~869--943.

\bibitem{Nepomnyashchii-71}
{\sc Y.~A. {Nepomnyashchi{\v{i}}}}, {\em {Coherent crystals with
  one-dimensional and cubic lattices}}, Theor. Math. Phys., 8 (1971),
  pp.~928--938.

\bibitem{NepNep-71}
{\sc Y.~A. {Nepomnyashchi{\v{i}}} and A.-A. {Nepomnyashchi{\v{i}}}}, {\em
  {Collective spectrum of coherent crystals}}, Theor. Math. Phys., 9 (1971),
  pp.~1033--1041.

\bibitem{NijRui-85b}
{\sc B.~Nijboer and T.~Ruijgrok}, {\em On the minimum-energy configuration of a
  one- dimensional system of particles interacting with the potential
  $\phi(x)=(1+x 4) {-1}$}, Physica A, 133 (1985), pp.~319--329.

\bibitem{Norcia_etal-21}
{\sc M.~A. {Norcia}, C.~{Politi}, L.~{Klaus}, E.~{Poli}, M.~{Sohmen}, M.~J.
  {Mark}, R.~N. {Bisset}, L.~{Santos}, and F.~{Ferlaino}}, {\em
  {Two-dimensional supersolidity in a dipolar quantum gas}}, Nature, 596
  (2021), pp.~357--361.

\bibitem{PacRiv-00}
{\sc F.~Pacard and T.~Rivi\`ere}, {\em Linear and nonlinear aspects of
  vortices}, vol.~39 of Progress in Nonlinear Differential Equations and their
  Applications, Birkh\"{a}user Boston, Inc., Boston, MA, 2000.
\newblock The Ginzburg-Landau model.

\bibitem{Park-84}
{\sc Y.~M. Park}, {\em Bounds on exponentials of local number operators in
  quantum statistical mechanics}, Comm. Math. Phys., 94 (1984), pp.~1--33.

\bibitem{Park-85}
\leavevmode\vrule height 2pt depth -1.6pt width 23pt, {\em Quantum statistical
  mechanics for superstable interactions: {B}ose-{E}instein statistics}, J.
  Statist. Phys., 40 (1985), pp.~259--302.

\bibitem{PetRot-18}
{\sc M.~Petrache and S.~Rota~Nodari}, {\em Equidistribution of jellium energy
  for {C}oulomb and {R}iesz interactions}, Constr. Approx., 47 (2018),
  pp.~163--210.

\bibitem{PetReb-07}
{\sc S.~N. Petrenko and A.~L. Rebenko}, {\em Superstable criterion and
  superstable bounds for infinite range interaction. {I}. {T}wo-body
  potentials}, Methods Funct. Anal. Topology, 13 (2007), pp.~50--61.

\bibitem{Pisante-11}
{\sc Pisante}, {\em {Two results on the equivariant Ginzburg-Landau vortex in
  arbitrarydimension}}, J. Funct. Anal., 260 (2011), pp.~892--905.

\bibitem{Pisante-14}
{\sc A.~Pisante}, {\em Symmetry in nonlinear {PDE}s: some open problems}, J.
  Fixed Point Theory Appl., 15 (2014), pp.~299--320.

\bibitem{Pitaevskii-61}
{\sc L.~P. Pitaevskii}, {\em Vortex lines in an imperfect {B}ose gas}, Zh.
  Eksper. Teor. fiz., 40 (1961), pp.~646--651.

\bibitem{PitStr-03}
{\sc L.~P. Pitaevskii and S.~Stringari}, {\em Bose-{E}instein condensation},
  no.~116, Oxford University Press, 2003.

\bibitem{PolBaiFerBla-24}
{\sc E.~Poli, D.~Baillie, F.~Ferlaino, and P.~B. Blakie}, {\em Excitations of a
  two-dimensional supersolid}, Phys. Rev. A, 110 (2024), p.~053301.

\bibitem{PomRic-93}
{\sc Y.~Pomeau and S.~Rica}, {\em Model of superflow with rotons}, Phys. Rev.
  Lett., 71 (1993), pp.~247--250.

\bibitem{PomRic-94}
\leavevmode\vrule height 2pt depth -1.6pt width 23pt, {\em Dynamics of a model
  of supersolid}, Phys. Rev. Lett., 72 (1994), pp.~2426--2429.

\bibitem{PreSerBru-18}
{\sc S.~Prestipino, A.~Sergi, and E.~Bruno}, {\em Clusterization of
  weakly-interacting bosons in one dimension: an analytic study at zero
  temperature}, J. Phys. A: Math. Theor., 52 (2018), p.~015002.

\bibitem{Radin-84}
{\sc C.~Radin}, {\em Classical ground states in one dimension}, J. Statist.
  Phys., 35 (1984), pp.~109--117.

\bibitem{Radin-04}
\leavevmode\vrule height 2pt depth -1.6pt width 23pt, {\em Existence of ground
  state configurations}, Math. Phys. Electron. J., 10 (2004), pp.~Paper 6, 7.

\bibitem{Rebenko-98}
{\sc A.~L. Rebenko}, {\em A new proof of {R}uelle's superstability bounds}, J.
  Statist. Phys., 91 (1998), pp.~815--826.

\bibitem{ReeSim2}
{\sc M.~Reed and B.~Simon}, {\em Methods of {M}odern {M}athematical {P}hysics.
  {II}. {F}ourier analysis, self-adjointness}, Academic Press, New York, 1975.

\bibitem{RipBaiBla-23}
{\sc B.~T.~E. Ripley, D.~Baillie, and P.~B. Blakie}, {\em {Two-dimensional
  supersolidity in a planar dipolar Bose gas}}, Phys. Rev. A, 108 (2023),
  p.~053321.

\bibitem{Roman-19}
{\sc C.~Rom\'{a}n}, {\em Three dimensional vortex approximation construction
  and {$\varepsilon$}-level estimates for the {G}inzburg-{L}andau functional},
  Arch. Ration. Mech. Anal., 231 (2019), pp.~1531--1614.

\bibitem{RotSer-15}
{\sc S.~{Rota Nodari} and S.~Serfaty}, {\em Renormalized energy
  equidistribution and local charge balance in 2d {C}oulomb system}, Int. Math.
  Res. Not. (IMRN), 11 (2015), pp.~3035--3093.

\bibitem{Rowlinson-79}
{\sc J.~S. {Rowlinson}}, {\em {Translation of J. D. van der Waals' ``The
  thermodynamik theory of capillarity under the hypothesis of a continuous
  variation of density''}}, J. Statist. Phys., 20 (1979), pp.~197--200.

\bibitem{Ruelle-70}
{\sc D.~Ruelle}, {\em Superstable interactions in classical statistical
  mechanics}, Comm. Math. Phys., 18 (1970), pp.~127--159.

\bibitem{Ruelle}
\leavevmode\vrule height 2pt depth -1.6pt width 23pt, {\em Statistical
  mechanics. Rigorous results}, {Singapore: World Scientific. London: Imperial
  College Press}, 1999.

\bibitem{Sandier-98}
{\sc E.~Sandier}, {\em Locally minimising solutions of {$-\Delta u=u(1-|u|^2)$}
  in {${\bf R}^2$}}, Proc. Roy. Soc. Edinburgh Sect. A, 128 (1998),
  pp.~349--358.

\bibitem{Sandier-98b}
\leavevmode\vrule height 2pt depth -1.6pt width 23pt, {\em Lower bounds for the
  energy of unit vector fields and applications}, J. Funct. Anal., 152 (1998),
  pp.~379--403.

\bibitem{SanSha-17}
{\sc E.~Sandier and I.~Shafrir}, {\em Small energy {G}inzburg-{L}andau
  minimizers in {$\Bbb{R}^3$}}, J. Funct. Anal., 272 (2017), pp.~3946--3964.

\bibitem{Savin-09}
{\sc O.~Savin}, {\em Regularity of flat level sets in phase transitions}, Ann.
  of Math. (2), 169 (2009), pp.~41--78.

\bibitem{SepJosRic-10}
{\sc N.~{Sep{\'u}lveda}, C.~{Josserand}, and S.~{Rica}}, {\em {Superfluid
  density in a two-dimensional model of supersolid}}, Eur. Phys. J. B, 78
  (2010), pp.~439--447.

\bibitem{Shafrir-94}
{\sc I.~Shafrir}, {\em Remarks on solutions of {$-\Delta u=(1-|u|^2)u$} in
  {${\bf R}^2$}}, C. R. Acad. Sci. Paris S\'{e}r. I Math., 318 (1994),
  pp.~327--331.

\bibitem{Simon-81}
{\sc B.~Simon}, {\em Large time behavior of the $l^p$ norm of {S}chr{\"o}dinger
  semigroups}, J. Funct. Anal., 40 (1981), pp.~66--83.

\bibitem{Suto-05}
{\sc A.~S{\"u}t\H{o}}, {\em Crystalline ground states for classical particles},
  Phys. Rev. Lett., 95 (2005), p.~265501.

\bibitem{Suto-11}
\leavevmode\vrule height 2pt depth -1.6pt width 23pt, {\em Ground state at high
  density}, Comm. Math. Phys., 305 (2011), pp.~657--710.

\bibitem{Suto-11b}
\leavevmode\vrule height 2pt depth -1.6pt width 23pt, {\em Superimposed
  particles in 1{D} ground states}, J. Phys. A, 44 (2011), pp.~035205, 7.

\bibitem{TanLucFamCatFioGabbBisSanMod-19}
{\sc L.~Tanzi, E.~Lucioni, F.~Fam\`a, J.~Catani, A.~Fioretti, C.~Gabbanini,
  R.~N. Bisset, L.~Santos, and G.~Modugno}, {\em Observation of a dipolar
  quantum gas with metastable supersolid properties}, Phys. Rev. Lett., 122
  (2019), p.~130405.

\bibitem{TanRocLucFamFioGabModRecStr-19}
{\sc L.~{Tanzi}, S.~M. {Roccuzzo}, E.~{Lucioni}, F.~{Fam{\`a}}, A.~{Fioretti},
  C.~{Gabbanini}, G.~{Modugno}, A.~{Recati}, and S.~{Stringari}}, {\em
  {Supersolid symmetry breaking from compressional oscillations in a dipolar
  quantum gas}}, Nature, 574 (2019), pp.~382--385.

\bibitem{TorSti-08}
{\sc S.~Torquato and F.~H. Stillinger}, {\em New duality relations for
  classical ground states}, Phys. Rev. Lett., 100 (2008), p.~020602.

\bibitem{Triay-18}
{\sc A.~Triay}, {\em Derivation of the dipolar {G}ross-{P}itaevskii energy},
  SIAM J. Math. Anal., 50 (2018), pp.~33--63.

\bibitem{Trudinger-67}
{\sc N.~S. Trudinger}, {\em On {H}arnack type inequalities and their
  application to quasilinear elliptic equations}, Comm. Pure Appl. Math., 20
  (1967), pp.~721--747.

\bibitem{Trudinger-73}
{\sc N.~S. {Trudinger}}, {\em {Linear elliptic operators with measurable
  coefficients}}, {Ann. Sc. Norm. Super. Pisa, Sci. Fis. Mat., III. Ser.}, 27
  (1973), pp.~265--308.

\bibitem{VanDerWaals-79}
{\sc J.~D. van~der Waals}, {\em {The thermodynamik theory of capillarity under
  the hypothesis of a continuous variation of density}}, J. Statist. Phys., 20
  (1979), pp.~200--244.
\newblock Translation of the original 1893 article, by J.S. Rowlinson.

\bibitem{VenNij-79-2}
{\sc W.~J. Ventevogel and B.~R.~A. Nijboer}, {\em On the configuration of
  systems of interacting particles with minimum potential energy per particle},
  Phys. A, 98 (1979), pp.~274--288.

\bibitem{VenNij-79-1}
\leavevmode\vrule height 2pt depth -1.6pt width 23pt, {\em On the configuration
  of systems of interacting particles with minimum potential energy per
  particle}, Phys. A, 99 (1979), pp.~569--580.

\bibitem{Wang-04}
{\sc F.-Y. Wang}, {\em Gradient estimates of {D}irichlet heat semigroups and
  application to isoperimetric inequalities}, Ann. Probab., 32 (2004),
  pp.~424--440.

\bibitem{WatBra-12}
{\sc H.~Watanabe and T.~c.~v. Brauner}, {\em Spontaneous breaking of continuous
  translational invariance}, Phys. Rev. D, 85 (2012), p.~085010.

\bibitem{YeZho-96}
{\sc D.~Ye and F.~Zhou}, {\em Uniqueness of solutions of the
  {G}inzburg-{L}andau problem}, Nonlinear Anal., 26 (1996), pp.~603--612.

\bibitem{ZhaMauPoh-19}
{\sc Y.-C. Zhang, F.~Maucher, and T.~Pohl}, {\em {Supersolidity around a
  Critical Point in Dipolar Bose-Einstein Condensates}}, Phys. Rev. Lett., 123
  (2019), p.~015301.

\end{thebibliography}

\end{document}